\documentclass[12pt]{statsoc2}
\usepackage[colorlinks,citecolor=blue,urlcolor=blue,breaklinks]{hyperref}
\usepackage{latexsym,amsmath,amssymb,amsfonts,dsfont,stmaryrd,url,breakcites,geometry}
\usepackage{graphicx}
\usepackage{vmargin} % required for statsoc+pdflatex
\usepackage{dsfont,url}
\newtheorem{theorem}{Theorem}

\newtheorem{proposition}{Proposition}
\newtheorem{corollary}{Corollary}

\DeclareMathOperator*{\argmin}{argmin}
\DeclareMathOperator*{\argmax}{argmax}
\DeclareMathOperator*{\sargmin}{sargmin}
\DeclareMathOperator*{\sargmax}{sargmax}

\newenvironment{prooftitle}[1]{{\noindent \textsc{Proof #1}}}

\title{Random-projection ensemble classification}
\author{Timothy I. Cannings and Richard J. Samworth}
\coaddress{Statistical Laboratory, Centre for Mathematical Sciences, Wilberforce Road, Cambridge, UK. CB3 0WB.}

\address{University of Cambridge, UK}
\email{ \{t.cannings, r.samworth\}@statslab.cam.ac.uk}

\date{}

\begin{document}
\maketitle

\begin{abstract}
We introduce a very general method for high-dimensional classification, based on careful combination of the results of applying an arbitrary base classifier to random projections of the feature vectors into a lower-dimensional space.  In one special case that we study in detail, the random projections are divided into disjoint groups, and within each group we select the projection yielding the smallest estimate of the test error.  Our random projection ensemble classifier then aggregates the results of applying the base classifier on the selected projections, with a data-driven voting threshold to determine the final assignment.  Our theoretical results elucidate the effect on performance of increasing the number of projections.  Moreover, under a boundary condition implied by the sufficient dimension reduction assumption, we show that the test excess risk of the random projection ensemble classifier can be controlled by terms that do not depend on the original data dimension and a term that becomes negligible as the number of projections increases.  The classifier is also compared empirically with several other popular high-dimensional classifiers via an extensive simulation study, which reveals its excellent finite-sample performance.
\end{abstract}


\section{Introduction}
Supervised classification concerns the task of assigning an object (or a number of objects) to one of two or more groups, based on a sample of labelled training data.  The problem was first studied in generality in the famous work of \citet{Fisher:36}, where he introduced some of the ideas of Linear Discriminant Analysis (LDA), and applied them to his Iris data set.  Nowadays, classification problems arise in a plethora of applications, including spam filtering, fraud detection, medical diagnoses, market research, natural language processing and many others.  

In fact, LDA is still widely used today, and underpins many other modern classifiers; see, for example, \citet*{Friedman:89} and \citet{Tibshirani:02}.  Alternative techniques include support vector machines \citep{Cortes:95}, tree classifiers and random forests \citep{Breiman:84,Breiman2001}, kernel methods \citep{Hall:05a} and nearest neighbour classifiers \citep{Fix:51}.  More substantial overviews and in-depth discussion of these techniques, and others, can be found in \citet*{PTPR:96} and \citet{ESL:09}. 

An increasing number of modern classification problems are \emph{high-dimensional}, in the sense that the dimension $p$ of the feature vectors may be comparable to or even greater than the number of training data points, $n$.  In such settings, classical methods such as those mentioned in the previous paragraph tend to perform poorly \citep{Bickel:04}, and may even be intractable; for example, this is the case for LDA, where the problems are caused by the fact that the sample covariance matrix is not invertible when $p \geq n$. 

Many methods proposed to overcome such problems assume that the optimal decision boundary between the classes is linear, e.g.\ \citet{Friedman:89} and \citet{Hastie:95}. Another common approach assumes that only a small subset of features are relevant for classification. Examples of works that impose such a sparsity condition include \citet*{Fan:08}, where it is also assumed that the features are independent, as well as \citet{Tibshirani:03}, where soft thresholding is used to obtain a sparse boundary.  More recently, \citet*{Witten:11} and \citet*{Fan:12} both solve an optimisation problem similar to Fisher's linear discriminant, with the addition of an $\ell_1$ penalty term to encourage sparsity. 

In this paper we attempt to avoid the curse of dimensionality by projecting the feature vectors at random into a lower-dimensional space. The use of random projections in high-dimensional statistical problems is motivated by the celebrated Johnson--Lindenstrauss Lemma \citep[e.g.][]{Dasgupta:02}. This lemma states that, given $x_1,\ldots,x_n \in \mathbb{R}^p$, $\epsilon \in (0,1)$ and $d>\frac{8\log{n}}{\epsilon^2}$, there exists a linear map $f: \mathbb{R}^p \to \mathbb{R}^d$ such that 
\[
(1-\epsilon)\|x_i - x_j\|^2 \leq \|f(x_i) - f(x_j)\|^2 \leq (1+\epsilon)\|x_i - x_j\|^2,
\]
for all $i,j = 1,\ldots,n$.  In fact, the function $f$ that nearly preserves the pairwise distances can be found in randomised polynomial time using random projections distributed according to Haar measure, as described in Section~\ref{sec--ChooseRP} below. It is interesting to note that the lower bound on $d$ in the Johnson--Lindenstrauss lemma does not depend on $p$; this lower bound is optimal up to constant factors \citep{Larsen:16}. As a result, random projections have been used successfully as a computational time saver: when $p$ is large compared to $\log{n}$, one may project the data at random into a lower-dimensional space and run the statistical procedure on the projected data, potentially making great computational savings, while achieving comparable or even improved statistical performance.  As one example of the above strategy, \citet{Durrant:13} obtained Vapnik--Chervonenkis type bounds on the generalisation error of a linear classifier trained on a single random projection of the data.  See also \citet{Dasgupta:99}, \citet{Ailon:06} and \citet*{McWilliams:14} for other instances.

Other works have sought to reap the benefits of aggregating over many random projections.  For instance, \citet*{Marzetta:11} considered estimating a $p \times p$ population inverse covariance (precision) matrix using $B^{-1}\sum_{b=1}^B \mathbf{A}_b^T(\mathbf{A}_b\hat{\Sigma}\mathbf{A}_b^T)^{-1}\mathbf{A}_b$, where $\hat{\Sigma}$ denotes the sample covariance matrix and $\mathbf{A}_1,\ldots,\mathbf{A}_B$ are random projections from $\mathbb{R}^p$ to $\mathbb{R}^d$.  \citet*{Lopes:11} used this estimate when testing for a difference between two Gaussian population means in high dimensions, while \citet*{Durrant:15} applied the same technique in Fisher's linear discriminant for a high-dimensional classification problem.  

Our proposed methodology for high-dimensional classification has some similarities with the techniques described above, in the sense that we consider many random projections of the data, but is also closely related to \emph{bagging} \citep{Breiman:96}, since the ultimate assignment of each test point is made by aggregation and a vote.  Bagging has proved to be an effective tool for improving unstable classifiers.  Indeed, a bagged version of the (generally inconsistent) $1$-nearest neighbour classifier is universally consistent as long as the resample size is carefully chosen, see \citet{Hall:05b}; for a general theoretical analysis of majority voting approaches, see also \citet{Lopes:16}.  Bagging has also been shown to be particularly effective in high-dimensional problems such as variable selection \citep{MeinshausenBuhlmann2010,ShahSamworth2013}.  Another related approach to ours is \citet{BlaserFryzlewicz2015}, who consider ensembles of random rotations, as opposed to projections.  

\begin{figure}[h]
  \centering
   \makebox{\includegraphics[width=\textwidth]{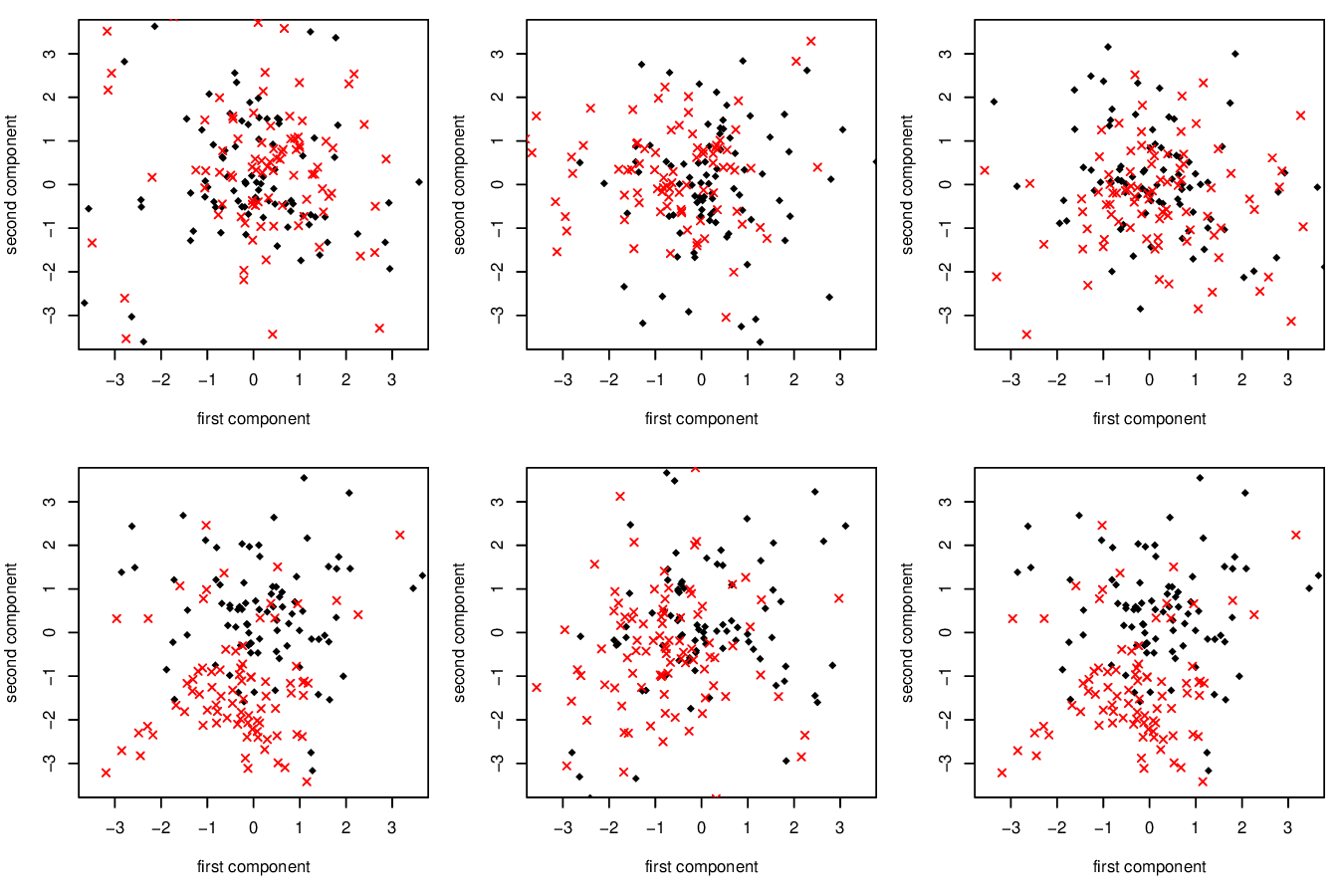}}
 \caption{\label{Fig:Useless}Different two-dimensional projections of $200$ observations in $p = 50$ dimensions.  Top row: three projections drawn from Haar measure; bottom row: the projected data after applying the projections with smallest estimate of test error out of 100 Haar projections with LDA (left), Quadratic Discriminant Analysis (middle) and $k$-nearest neighbours (right).}
 \end{figure}

One of the basic but fundamental observations that underpins our proposal is the fact that aggregating the classifications of all random projections is not always sensible, since many of these projections will typically destroy the class structure in the data; see the top row of Figure~\ref{Fig:Useless}.  For this reason, we advocate partitioning the projections into disjoint groups, and within each group we retain only the projection yielding the smallest estimate of the test error.  The attraction of this strategy is illustrated in the bottom row of Figure~\ref{Fig:Useless}, where we see a much clearer partition of the classes.  Another key feature of our proposal is the realisation that a simple majority vote of the classifications based on the retained projections can be highly suboptimal; instead, we argue that the voting threshold should be chosen in a data-driven fashion in an attempt to minimise the test error of the infinite-simulation version of our random projection ensemble classifier.  In fact, this estimate of the optimal threshold turns out to be remarkably effective in practice; see Section~\ref{sec--alpha} for further details.  We emphasise that our methodology can be used in conjunction with any base classifier, though we particularly have in mind classifiers designed for use in low-dimensional settings.  The random projection ensemble classifier can therefore be regarded as a general technique for either extending the applicability of an existing classifier to high dimensions, or improving its performance.  The methodology is implemented in an \texttt{R} package \texttt{RPEnsemble} \citep{CanningsSamworth2016b}.   

Our theoretical results are divided into three parts.  In the first, we consider a generic base classifier and a generic method for generating the random projections into $\mathbb{R}^d$ and quantify the difference between the test error of the random projection ensemble classifier and its infinite-simulation counterpart as the number of projections increases.  We then consider selecting random projections from non-overlapping groups by initially drawing them according to Haar measure, and then within each group retaining the projection that minimises an estimate of the test error.  Under a condition implied by the widely-used sufficient dimension reduction assumption \citep{Li:91,Cook:98,Lee:13}, we can then control the difference between the test error of the random projection classifier and the Bayes risk as a function of terms that depend on the performance of the base classifier based on projected data and our method for estimating the test error, as well as a term that becomes negligible as the number of projections increases.  The final part of our theory gives risk bounds for the first two of these terms for specific choices of base classifier, namely Fisher's linear discriminant and the $k$-nearest neighbour classifier.  The key point here is that these bounds only depend on $d$, the sample size $n$ and the number of projections, and not on the original data dimension $p$.  
    
The remainder of the paper is organised as follows.  Our methodology and general theory are developed in Sections~\ref{sec--RP} and~\ref{sec--ChooseRP}.   Specific choices of base classifier as well as a general sample splitting strategy are discussed in Section~\ref{sec--base}, while Section~\ref{sec--prac} is devoted to a consideration of the practical issues of computational complexity, choice of voting threshold, projected dimension and the number of projections used.  In Section~\ref{sec--empirical} we present results from an extensive empirical analysis on both simulated and real data where we compare the performance of the random projection ensemble classifier with several popular techniques for high-dimensional classification.  The outcomes are very encouraging, and suggest that the random projection ensemble classifier has excellent finite-sample performance in a variety of different high-dimensional classification settings.  We conclude with a discussion of various extensions and open problems.  Proofs are given in the Appendix and the supplementary material \citet{CanningsSamworth:16}, which appears below the reference list.  

Finally in this section, we introduce the following general notation used throughout the paper.  For a sufficiently smooth real-valued function $g$ defined on a neighbourhood of $t \in \mathbb{R}$, let $\dot{g}(t)$ and $\ddot{g}(t)$ denote its first and second derivatives at $t$, and let $\lfloor t \rfloor$ and $\llbracket t \rrbracket := t- \lfloor t \rfloor $ denote the integer and fractional part of $t$ respectively.

\section{A generic random projection ensemble classifier}
\label{sec--RP}

We start by describing our setting and defining the relevant notation.  Suppose that the pair $(X,Y)$ takes values in $\mathbb{R}^p \times \{0,1\}$, with joint distribution $P$, characterised by $\pi_1 := \mathbb{P}(Y = 1)$, and $P_r$, the conditional distribution of $X|Y = r$, for $r = 0, 1$.  For convenience, we let $\pi_0 := \mathbb{P}(Y = 0) = 1- \pi_1$.  In the alternative characterisation of $P$, we let $P_X$ denote the marginal distribution of $X$ and write $\eta(x) := \mathbb{P}(Y=1 |X=x)$ for the regression function.   Recall that a \emph{classifier} on $\mathbb{R}^p$ is a Borel measurable function $C:\mathbb{R}^p \rightarrow \{0,1\}$, with the interpretation that we assign a point $x\in \mathbb{R}^p$ to class $C(x)$.  We let $\mathcal{C}_p$ denote the set of all such classifiers.  

The test error of a classifier $C$ is\footnote{We define $R(C)$ through an integral rather than $R(C) := \mathbb{P}\{C(X) \neq Y\}$ to make it clear that when $C$ is random (depending on training data or random projections), it should be conditioned on when computing $R(C)$.}
\[
R(C) := \int_{\mathbb{R}^p \times \{0,1\}} \mathds{1}_{\{C(x) \neq y\}} \, dP(x,y),
\]
and is minimised by the \emph{Bayes} classifier
\[
C^{\mathrm{Bayes}}(x) := \left\{ \begin{array}{ll}
         1 & \mbox{if $\eta(x) \geq 1/2$};\\
         0 & \mbox{otherwise}\end{array} \right. 
\]
\citep*[e.g.][p.~10]{PTPR:96}. Its risk is $R(C^{\mathrm{Bayes}}) =  \mathbb{E}[\min\{\eta(X), 1 - \eta(X)\}]$.

Of course, we cannot use the Bayes classifier in practice, since $\eta$ is unknown.  Nevertheless, we often have access to a sample of training data that we can use to construct an approximation to the Bayes classifier.  Throughout this section and Section~\ref{sec--ChooseRP}, it is convenient to consider the training sample $\mathcal{T}_n := \{(x_1,y_1),\ldots,(x_n,y_n)\}$ to be fixed points in $\mathbb{R}^p \times \{0,1\}$.  Our methodology will be applied to a base classifier ${C}_n = {C}_{n,\mathcal{T}_{n,d}}$, which we assume can be constructed from an arbitrary training sample $\mathcal{T}_{n,d}$ of size $n$ in $\mathbb{R}^d \times \{0,1\}$; thus ${C}_n$ is a measurable function from $(\mathbb{R}^d \times \{0,1\})^n$ to~$\mathcal{C}_d$.

Now assume that $d \leq p$.  We say a matrix $A \in \mathbb{R}^{d\times p}$ is a \emph{projection} if $AA^T = I_{d\times d}$, the $d$-dimensional identity matrix. Let $\mathcal{A} = \mathcal{A}_{d \times p} := \{A \in \mathbb{R}^{d\times p} : AA^T = I_{d\times d} \}$ be the set of all such matrices.  Given a projection $A \in \mathcal{A}$, define projected data $z_i^A := Ax_i$ and $y_i^A := y_i$ for $i=1,\ldots,n$, and let $\mathcal{T}_n^A := \{(z_1^A,y_1^A),\ldots,(z_n^A,y_n^A)\}$.  The projected data base classifier corresponding to ${C}_n$ is ${C}_n^A: (\mathbb{R}^d \times \{0,1\})^n \rightarrow \mathcal{C}_p$, given by
\[
{C}_n^A(x) = {C}_{n,\mathcal{T}_n^A}^A(x) := {C}_{n,\mathcal{T}_n^A}(Ax).
\]
Note that although ${C}_n^A$ is a classifier on $\mathbb{R}^p$, the value of ${C}_n^A(x)$ only depends on $x$ through its $d$-dimensional projection $Ax$.

We now define a generic ensemble classifier based on random projections.  For $B_1 \in \mathbb{N}$, let $\mathbf{A}_1,\ldots,\mathbf{A}_{B_1}$ denote independent and identically distributed projections in $\mathcal{A}_{d \times p} $, independent of $(X,Y)$.  The distribution on $\mathcal{A}$ is left unspecified at this stage, and in fact our proposed method ultimately involves choosing this distribution depending on $\mathcal{T}_n$.

Now set  
\begin{equation}
\label{eq--nu}
\nu_{n}(x) =  \nu_n^{(B_1)}(x) :=  \frac{1}{B_1}\sum_{b_1=1}^{B_1} \mathds{1}_{\{{C}_n^{\mathbf{A}_{b_1}}(x) = 1 \}}.
\end{equation}
For $\alpha \in (0,1)$, the \emph{random projection ensemble} classifier is defined to be
\begin{equation}
\label{eq--RP}
{C}_n^{\mathrm{RP}}(x) :=  \left\{ \begin{array}{ll}
         1 & \mbox{if $ \nu_n(x) \geq \alpha$;}\\
         0 & \mbox{otherwise.}\end{array} \right. 
\end{equation}
We emphasise again here the additional flexibility afforded by not pre-specifying the voting threshold $\alpha$ to be $1/2$.  Our analysis of the random projection ensemble classifier will require some further definitions.  Let\footnote{In order to distinguish between different sources of randomness, we will write $\mathbf{P}$ and $\mathbf{E}$ for the probability and expectation, respectively, taken over the randomness from the projections $\mathbf{A}_1, \ldots, \mathbf{A}_{B_1}$.  If the training data is random, then we condition on $\mathcal{T}_n$ when computing $\mathbf{P}$ and $\mathbf{E}$.}
\[
\mu_n(x) := \mathbf{E}\{ \nu_n(x)\} = \mathbf{P}\{{C}_n^{\mathbf{A}_1}(x) = 1\}.
\]
For $r= 0,1$, define distribution functions $G_{n,r}: [0,1] \rightarrow [0,1]$ by $G_{n,r}(t) := P_r(\{x \in \mathbb{R}^p : \mu_n(x) \leq t\})$.  Note that since $G_{n,r}$ is non-decreasing it is differentiable almost everywhere; in fact, however, the following assumption will be convenient:
\begin{description}
\item \textit{Assumption 1.} $G_{n,0}$ and $G_{n,1}$ are twice differentiable at $\alpha$.   
\end{description}
The first derivatives of $G_{n,0}$ and $G_{n,1}$, when they exist, are denoted as $g_{n,0}$ and $g_{n,1}$ respectively; under assumption~1, these derivatives are well-defined in a neighbourhood of $\alpha$.  
Our first main result below gives an asymptotic expansion for the expected test error $\mathbf{E}\{R({C}_n^{\mathrm{RP}})\}$ of our generic random projection ensemble classifier as the number of projections increases.  In particular, we show that this expected test error can be well approximated by the test error of the infinite-simulation random projection classifier
\[
{C}_n^{\mathrm{RP}^*}(x) :=  \left\{ \begin{array}{ll}
         1 & \mbox{if $\mu_n(x) \geq \alpha$;}\\
         0 & \mbox{otherwise.}\end{array} \right. 
\]
Note that provided $G_{n,0}$ and $G_{n,1}$ are continuous at $\alpha$, we have
\begin{equation}
\label{TestErrorInfiniteSim}
R({C}_n^{\mathrm{RP}^*})  = \pi_1 G_{n,1}(\alpha)  + \pi_0\{1-G_{n,0}(\alpha)\}.
\end{equation}
\begin{theorem}
\label{thm--Bvar}
Assume assumption 1.  Then
\[
\mathbf{E}\{R({C}_n^{\mathrm{RP}})\}- R({C}_n^{\mathrm{RP}^*}) = \frac {\gamma_n(\alpha)} {B_1} + o\Bigl(\frac{1}{B_1}\Bigr)
\]
as $B_1 \to \infty$, where 
\[
\gamma_n(\alpha) :=  (1-\alpha-\llbracket B_1\alpha \rrbracket)\{\pi_1 g_{n,1}(\alpha) - \pi_0  g_{n,0}(\alpha)\} + \frac{\alpha(1-\alpha)}{2}\{\pi_1 \dot g_{n,1}(\alpha) - \pi_0  \dot g_{n,0}(\alpha)\}.
\]
\end{theorem}
The proof of Theorem~\ref{thm--Bvar} in the Appendix is lengthy, and involves a one-term Edgeworth approximation to the distribution function of a standardised Binomial random variable.  One of the technical challenges is to show that the error in this approximation holds uniformly in the binomial proportion.  Related techniques can also be used to show that $\mathbf{Var}\{R({C}_n^{\mathrm{RP}})\} = O(B_1^{-1})$ under assumption~1; see Proposition~\ref{Prop:variance} in the supplementary material.  Very recently, \citet{Lopes:16} has obtained similar results to this and to Theorem~\ref{thm--Bvar} in the context of majority vote classification, with stronger assumptions on the relevant distributions and on the form of the voting scheme.  In Figure~\ref{Fig:Errors}, we plot the average error (plus/minus two standard deviations) of the random projection ensemble classifier in one numerical example, as we vary $B_1 \in \{2, \ldots, 500\}$; this reveals that the Monte Carlo error stabilises rapidly, in agreement with what \citet{Lopes:16} observed for a random forest classifier.  

\begin{figure}[h]
  \centering
 \makebox{\includegraphics[width=\textwidth]{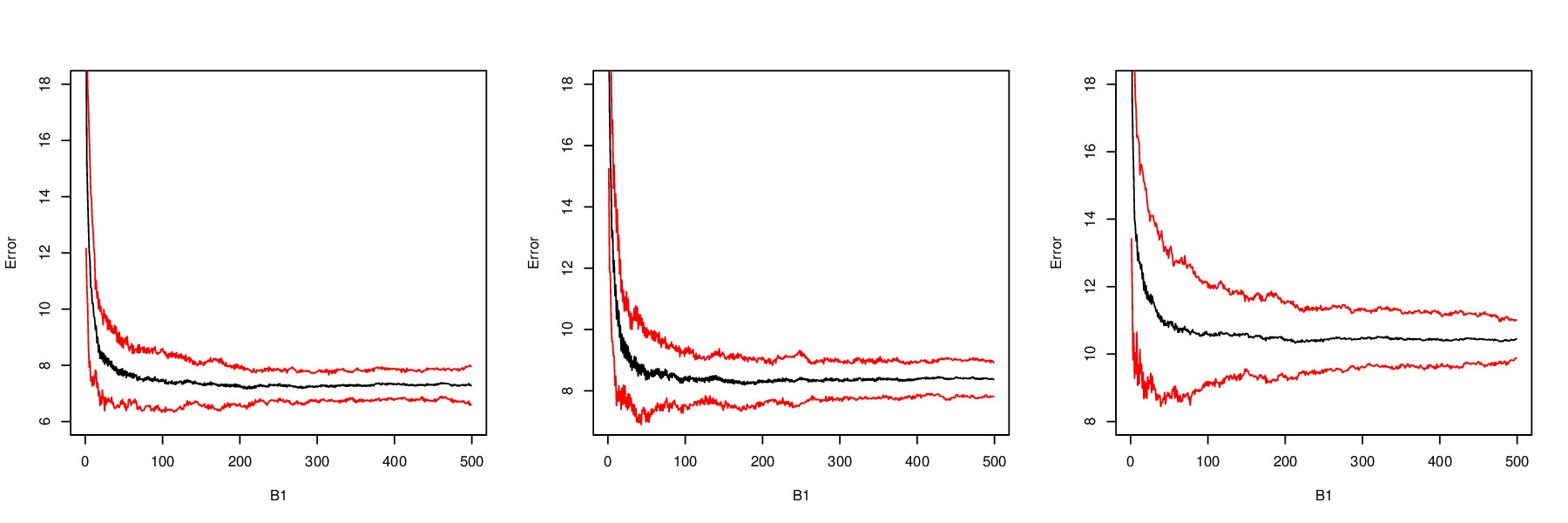}}
 \caption{\label{Fig:Errors} The average error (black) plus/minus two standard deviations (red) over 20 sets of $B_1B_2$ projections for $B_1 \in \{2,\ldots, 500\}$. We use the LDA (left), QDA (middle) and $k$nn (right) base classifiers.  The plots show the test error for one training dataset from Model~2; the other parameters are $n = 50$, $p = 100$, $d = 5$ and $B_2 = 50$.}
 \end{figure}

Our next result controls the test excess risk, i.e.\ the difference between the expected test error and the Bayes risk, of the random projection classifier in terms of the expected test excess risk of the classifier based on a single random projection.  An attractive feature of this result is its generality: no assumptions are placed on the configuration of the training data $\mathcal{T}_n$, the distribution $P$ of the test point $(X,Y)$ or on the distribution of the individual projections.
\begin{theorem}
\label{Thm:BaggingBound}
For each $B_1 \in \mathbb{N} \cup \{\infty\}$, we have
\begin{equation}
\label{Eq:BaggingBound}
\mathbf{E}\{R({C}_n^{\mathrm{RP}})\} - R(C^{\mathrm{Bayes}}) \leq \frac{1}{\min(\alpha,1-\alpha)}\bigl[\mathbf{E}\{R({C}_n^{\mathbf{A}_1})\} - R(C^{\mathrm{Bayes}})\bigr].
\end{equation}
\end{theorem}
When $B_1 = \infty$, we interpret $R({C}_n^{\mathrm{RP}})$ in Theorem~\ref{Thm:BaggingBound} as $R({C}_n^{\mathrm{RP}^*})$.  In fact, when $B_1 = \infty$ and $G_{n,0}$ and $G_{n,1}$ are continuous, the bound in Theorem~\ref{Thm:BaggingBound} can be improved if one is using an `oracle' choice of the voting threshold $\alpha$, namely
\begin{equation}
\label{eq--alpha}
\alpha^* \in \argmin_{\alpha' \in [0,1]} R({C}_{n,\alpha'}^{\mathrm{RP}^*}) = \argmin_{\alpha' \in [0,1]} \bigl[\pi_1 G_{n,1}(\alpha') + \pi_0\{1 - G_{n,0}(\alpha')\}\bigr],
\end{equation} 
where we write ${C}_{n, \alpha}^{\mathrm{RP}^*}$ to emphasise the dependence on the voting threshold $\alpha$.  In this case, by definition of $\alpha^{*}$ and then applying Theorem~\ref{Thm:BaggingBound},
\begin{equation}
\label{eq:BaggingBound'}
R({C}_{n, \alpha^*}^{\mathrm{RP}^*}) - R(C^{\mathrm{Bayes}}) \leq R({C}_{n, 1/2}^{\mathrm{RP}^*}) - R(C^{\mathrm{Bayes}}) \leq 2\bigl[\mathbf{E}\{R({C}_n^{\mathbf{A}_1})\} - R(C^{\mathrm{Bayes}})\bigr],
\end{equation}
which improves the bound in~\eqref{Eq:BaggingBound} since $2 \leq \frac{1}{\min\{\alpha^*,(1-\alpha^*)\}}$.  It is also worth mentioning that if assumption~1 holds at $\alpha^* \in (0,1)$, and $G_{n,0}$ and $G_{n,1}$ are continuous, then $\pi_1 g_{n,1}(\alpha^*) = \pi_0  g_{n,0}(\alpha^*)$ and the constant in Theorem~\ref{thm--Bvar} simplifies to
\[
\gamma_n(\alpha^*) = \frac{\alpha^*(1-\alpha^*)}{2}\{\pi_1 \dot g_{n,1}(\alpha^*) - \pi_0 \dot g_{n,0}(\alpha^*)\} \geq 0.
\]

\section{Choosing good random projections}
\label{sec--ChooseRP}
In this section, we study a special case of the generic random projection ensemble classifier introduced in Section~\ref{sec--RP}, where we propose a screening method for choosing the random projections.  
Let ${R}_n^A$ be an estimator of $R({C}_n^A)$, based on $\{(z_1^A,y_1^A), \dots, (z_n^A, y_n^A)\}$, that takes values in the set $\{0,1/n,\ldots,1\}$.  Examples of such estimators include the training error and leave-one-out estimator; we discuss these choices in greater detail in Section~\ref{sec--base}.  For $B_1, B_2 \in \mathbb{N}$, let $\{\mathbf{A}_{b_1,b_2}:b_1=1,\ldots,B_1,b_2=1,\ldots,B_2\}$ denote independent projections, independent of $(X,Y)$, distributed according to Haar measure on $\mathcal{A}$.  One way to simulate from Haar measure on the set $\mathcal{A}$ is to first generate a matrix $\mathbf{Q} \in \mathbb{R}^{d \times p}$, where each entry is drawn independently from a standard normal distribution, and then take $\mathbf{A}^T$ to be the matrix of left singular vectors in the singular value decomposition of $\mathbf{Q}^T$ \citep[see, for example,][Theorem~1.5.4]{Chikuse:03}.  For $b_1=1,\ldots,B_1$, let
\begin{equation}
\label{eq--ordering}
b_2^*(b_1) := \sargmin_{b_2\in \{1,\dots,B_2\}} {R}_n^{\mathbf{A}_{b_1,b_2}},
\end{equation}
where $\sargmin$ denotes the smallest index where the minimum is attained in the case of a tie.  We now set $\mathbf{A}_{b_1} := \mathbf{A}_{b_1,b_2^*(b_1)}$, and consider the random projection ensemble classifier from Section~\ref{sec--RP} constructed using the independent projections $\mathbf{A}_1,\ldots,\mathbf{A}_{B_1}$.  

Let
\[
R_n^* := \min_{A \in \mathcal{A}} {R}_n^A
\]
denote the optimal test error estimate over all projections.  The minimum is attained here, since $R_n^A$ takes only finitely many values.  We assume the following: 
\begin{description}
\item \textit{Assumption 2.} There exists $\beta \in (0,1]$ such that
\[
\mathbf{P}\bigl({R}_n^{\mathbf{A}_{1,1}} \leq {R}_n^* + |\epsilon_n| \bigr) \geq \beta,
\]
where $\epsilon_n = \epsilon_n^{(B_2)} := \mathbf{E}\{R({C}_n^{\mathbf{A}_1}) - {R}_n^{\mathbf{A}_1}\}$.
\end{description}

The quantity $\epsilon_{n}$, which depends on $B_2$ because $\mathbf{A}_1$ is selected from $B_2$ independent random projections, can be interpreted as a measure of overfitting.  Assumption~2 asks that there is a positive probability that $R_n^{\mathbf{A}_{1,1}}$ is within $|\epsilon_{n}|$ of its minimum value $R_n^{*}$. The intuition here is that spending more computational time choosing a projection by increasing $B_{2}$ is potentially futile: one may find a projection with a lower error estimate, but the chosen projection will not necessarily result in a classifier with a lower test error.  Under this condition, the following result controls the test excess risk of our random projection ensemble classifier in terms of the test excess risk of a classifier based on $d$-dimensional data, as well as a term that reflects our ability to estimate the test error of classifiers based on projected data and a term that depends on the number of projections.
\begin{theorem}
\label{Thm:Main}
Assume assumption~2.  Then, for each $B_1,B_2 \in \mathbb{N}$, and every $A \in \mathcal{A}$,
\begin{equation}
\label{Eq:Main}
\mathbf{E}\{R({C}_n^{\mathrm{RP}})\} - R(C^{\mathrm{Bayes}})  \leq \frac{R({C}_n^{A}) - R(C^{\mathrm{Bayes}})}{\min(\alpha,1-\alpha)} + \frac{2|\epsilon_n| - \epsilon_n^{A}}{\min(\alpha,1-\alpha)} + \frac{(1-\beta)^{B_{2}}}{\min(\alpha,1-\alpha)},
\end{equation}
where $\epsilon_n^{A} := R({C}_n^{A}) - {R}_n^{A}$.
\end{theorem}
Regarding the bound in Theorem~\ref{Thm:Main} as a sum of three terms, we see that the final one can be seen as the price we have to pay for the fact that we do not have access to an infinite sample of random projections.  This term can be made negligible by choosing $B_2$ to be sufficiently large, though the value of $B_2$ required to ensure it is below a prescribed level may depend on the training data.  It should also be noted that $\epsilon_n$ in the second term may increase with $B_2$, which reflects the fact mentioned previously that this quantity is a measure of overfitting.  The behaviour of the first two terms depends on the choice of base classifier, and our aim is to show that under certain conditions, these terms can be bounded (in expectation over the training data) by expressions that do not depend on $p$.

To this end, define the regression function on $\mathbb{R}^d$ induced by the projection $A \in \mathcal{A}$ to be $\eta^A(z) := \mathbb{P}(Y=1|AX=z)$.  The corresponding induced Bayes classifier, which is the optimal classifier knowing only the distribution of $(AX,Y)$, is given by
\[
C^{A-\mathrm{Bayes}}(z) := \left\{ \begin{array}{ll}
         1 & \mbox{if $\eta^A(z) \geq 1/2$;}\\
         0 & \mbox{otherwise.}\end{array} \right. 
\]
In order to give a condition under which there exists a projection $A \in \mathcal{A}$ for which $R({C}_n^{A})$ is close to the Bayes risk, we will invoke an additional assumption on the form of the Bayes classifier:
\begin{description}
\item \textit{Assumption 3.} There exists a projection $A^* \in \mathcal{A}$ such that 
\[
P_X(\{x \in \mathbb{R}^p: \eta(x) \geq 1/2\} \triangle \{x \in \mathbb{R}^p: \eta^{A^*}(A^*x) \geq 1/2\}) = 0,
\]
where $B \triangle C := (B \cap C^c) \cup (B^c \cap C)$ denotes the symmetric difference of two sets $B$ and $C$.
\end{description}
Assumption~3 requires that the set of points $x \in \mathbb{R}^p$ assigned by the Bayes classifier to class 1 can be expressed as a function of a $d$-dimensional projection of $x$.  Note that if the Bayes decision boundary is a hyperplane, then assumption~3 holds with $d=1$.  Moreover, Proposition~\ref{Prop:Cond} below shows that, in fact, assumption~3 holds under the sufficient dimension reduction condition, which states that $Y$ is conditionally independent of $X$ given $A^*X$; see \citet{Cook:98} for many statistical settings where such an assumption is natural.     
\begin{proposition}
\label{Prop:Cond}
If $Y$ is conditionally independent of $X$ given $A^*X$, then assumption~3 holds.
\end{proposition}
The following result confirms that under assumption~3, and for a sensible choice of base classifier, we can hope for $R({C}_n^{A^*})$ to be close to the Bayes risk.
\begin{proposition}
\label{Prop:A^*}
Assume assumption 3.  Then $R(C^{A^*-\mathrm{Bayes}}) = R(C^{\mathrm{Bayes}})$.
\end{proposition}
We are therefore now in a position to study the first two terms in the bound in Theorem~\ref{Thm:Main} in more detail for specific choices of base classifier.
 
\section{Possible choices of the base classifier}
\label{sec--base}
In this section, we change our previous perspective and regard the training data as independent random pairs with distribution $P$, so our earlier statements are interpreted conditionally on $\mathcal{T}_n :=  \{(X_1,Y_1), \dots, (X_n, Y_n)\}$.  For $A \in \mathcal{A}$, we write our projected data as $\mathcal{T}_n^A :=  \{(Z_1^A,Y_1^A), \dots, (Z_n^A, Y_n^A)\}$, where $Z_i^A := AX_i$ and $Y_i^A := Y_i$.  We also write $\mathbb{P}$ and $\mathbb{E}$ to refer to probabilities and expectations over all random quantities.  We consider particular choices of base classifier, and study the first two terms in the bound in Theorem~\ref{Thm:Main}.

\subsection{Linear Discriminant Analysis}
\label{sec--LDA}
Linear Discriminant Analysis (LDA), introduced by \citet{Fisher:36}, is arguably the simplest classification technique.  Recall that in the special case where $X|Y=r \sim N_p(\mu_r, \Sigma)$, we have
\[
\mathrm{sgn}\{\eta(x) - 1/2\} = \mathrm{sgn}\biggl\{\log\frac{\pi_1}{\pi_0}  + \Bigl(x-\frac{\mu_1 +\mu_0}{2}\Bigr)^T \Sigma^{-1}(\mu_1 - \mu_0)\biggr\},
\]
so assumption~3 holds with $d=1$ and $A^* = \frac{(\mu_1 - \mu_0)^T\Sigma^{-1}}{\|\Sigma^{-1}(\mu_1 - \mu_0)\|}$, a $1 \times p$ matrix.  In LDA, $\pi_r$, $\mu_r$ and $\Sigma$ are estimated by their sample versions, using a pooled estimate of $\Sigma$.  Although LDA cannot be applied directly when $p \geq n$ since the sample covariance matrix is singular, we can still use it as the base classifier for a random projection ensemble, provided that $d < n$.  Indeed, noting that for any $A \in \mathcal{A}$, we have $AX|Y=r \sim N_d(\mu_r^A,\Sigma^A)$, where $\mu_r^A := A\mu_r$ and $\Sigma^A := A\Sigma A^T$, we can define\begin{equation}
\label{eq--LDA}
{C}_n^{A}(x) = {C}_n^{A-\mathrm{LDA}}(x) := \left\{ \begin{array}{ll}
         1 & \mbox{if $\log\frac{\hat{\pi}_1}{\hat{\pi}_0}  + \bigl(Ax-\frac{\hat{\mu}_1^A +\hat\mu_0^A}{2}\bigr)^T \hat{\Omega}^A(\hat{\mu}^A_1 - \hat{\mu}^A_0) \geq 0$};\\
         0 & \mbox{otherwise.}\end{array} \right. 
\end{equation}
Here, $\hat{\pi}_r := n_r/n$, where $n_r := \sum_{i=1}^n \mathds{1}_{\{Y_i = r\}}$, $\hat{\mu}_r^A := n_r^{-1}\sum_{i=1}^n AX_i\mathds{1}_{\{Y_i = r\}}$, 
\[
\hat{\Sigma}^A := \frac{1}{n-2}\sum_{i=1}^n\sum_{r=0}^1 (AX_i - \hat{\mu}_r^A)(AX_i - \hat{\mu}_r^A)^T\mathds{1}_{\{Y_i = r\}}
\]
and $\hat{\Omega}^A := (\hat{\Sigma}^A)^{-1}$.

Write $\Phi$ for the standard normal distribution function.  Under the normal model specified above, the test error of the LDA classifier can be written as
\begin{align*}
R({C}_n^{A}) = \pi_0 \Phi\Biggl(\frac{\log \frac{\hat{\pi}_1}{\hat{\pi}_0} + (\hat{\delta}^A)^T \hat{\Omega}^A(\bar{\hat{\mu}}^A - \mu^A_0)}{\sqrt{(\hat{\delta}^A)^T \hat{\Omega}^A\Sigma^A \hat{\Omega}^A\hat{\delta}^A}}\Biggr) +  \pi_1 \Phi\Biggl(\frac{\log \frac{\hat{\pi}_0}{\hat{\pi}_1} -(\hat{\delta}^A)^T \hat{\Omega}^A(\bar{\hat{\mu}}^A - \mu^A_1)}{\sqrt{(\hat{\delta}^A)^T \hat{\Omega}^A\Sigma^A \hat{\Omega}^A\hat{\delta}^A}}\Biggr),
\end{align*}
where $\hat{\delta}^A := \hat{\mu}_0^A - \hat{\mu}_1^A$ and $\bar{\hat{\mu}}^A := (\hat{\mu}_0^A + \hat{\mu}_1^A)/2$.  

\citet{Efron:75} studied the excess risk of the LDA classifier in an asymptotic regime in which $d$ is fixed as $n$ diverges. Specialising his results for simplicity to the case where $\pi_0 = \pi_1$, he showed that using the LDA classifier~\eqref{eq--LDA} with $A=A^*$ yields
\begin{equation}
\label{Eq:RiskLDA}
\mathbb{E}\{R({C}_n^{A^*})\} - R(C^{\mathrm{Bayes}}) = \frac{d}{n}\phi\Bigl(-\frac{\Delta}{2}\Bigr)\biggl(\frac{\Delta}{4} + \frac{1}{\Delta}\biggr)\{1 + o(1)\} 
\end{equation}
as $n \rightarrow \infty$, where $\Delta := \|\Sigma^{-1/2}(\mu_0 - \mu_1)\| = \|(\Sigma^{A^*})^{-1/2}(\mu_0^{A^*} - \mu_1^{A^*})\|$. 

It remains to control the errors $\epsilon_n$ and $\epsilon_n^{A^*}$ in Theorem \ref{Thm:Main}.  For the LDA classifier, we consider the training error estimator
\begin{equation}
\label{eq--riskest1}
{R}_{n}^{A} := \frac{1}{n} \sum_{i=1}^n  \mathds{1}_{\{{C}_{n}^{A-\mathrm{LDA}}(X_i) \neq Y_{i} \}}.
\end{equation}
\citet{DevroyeWagner1976} provided a Vapnik--Chervonenkis bound for ${R}_n^A$ under no assumptions on the underlying data generating mechanism: for every $n \in \mathbb{N}$ and $\epsilon >0$,
\begin{equation}
\label{Eq:InSampleBound}
\sup_{A\in \mathcal{A}} \mathbb{P}\{|R({C}_n^{A}) - {R}_{n}^{A}| > \epsilon \} \leq 8(n+1)^{d+1}e^{-n\epsilon^2/32};
\end{equation}
see also \citet[][Theorem~23.1]{PTPR:96}.  We can then conclude that 
\begin{align}
\label{Eq:epsilonnALDA}
\mathbb{E}|\epsilon_n^{A^*}| \leq \mathbb{E}\bigl|R({C}_n^{A^*}) - {R}_{n}^{A^*}\bigr| & \leq \inf_{\epsilon_0 \in (0,1)} \epsilon_0 + 8(n+1)^{d+1}\int_{\epsilon_0}^1 e^{-ns^2/32} \, ds \nonumber
\\ & \leq  8 \sqrt{\frac{(d+1)\log(n+1) + 3 \log 2 + 1}{2n}}. 
\end{align}
The more difficult term to deal with is 
\[
\mathbb{E}|\epsilon_n| = \mathbb{E}\bigl|\mathbf{E}\{R({C}_n^{\mathbf{A}_1}) - {R}_{n}^{\mathbf{A}_1}\}\bigr| \leq \mathbb{E}\bigl|R({C}_n^{\mathbf{A}_1}) - {R}_{n}^{\mathbf{A}_1}\bigr|.
\]  
In this case, the bound~\eqref{Eq:InSampleBound} cannot be applied directly, because $(X_1,Y_1),\ldots,(X_n,Y_n)$ are no longer independent conditional on $\mathbf{A}_1$; indeed $\mathbf{A}_1 = \mathbf{A}_{1,b_2^*(1)}$ is selected from $\mathbf{A}_{1,1},\ldots,\mathbf{A}_{1,B_2}$ so as to minimise an estimate of test error, which depends on the training data.  Nevertheless, since $\mathbf{A}_{1,1},\ldots,\mathbf{A}_{1,B_2}$ are independent of $\mathcal{T}_n$, we still have that
\begin{align*}
\mathbb{P}\Bigl\{\max_{b_2=1,\ldots,B_2}  |R({C}_n^{\mathbf{A}_{1,b_2}}) - {R}_{n}^{\mathbf{A}_{1,b_2}}| > \epsilon \Bigm| \mathbf{A}_{1,1},\ldots,\mathbf{A}_{1,B_2}\Bigr\} & \leq \sum_{b_2=1}^{B_2} \mathbb{P}\bigl\{|R({C}_n^{\mathbf{A}_{1,b_2}}) - {R}_{n}^{\mathbf{A}_{1,b_2}}| > \epsilon \bigm|\mathbf{A}_{1,b_2}\bigr\} \\
&\leq 8(n+1)^{d+1} B_2 e^{-n\epsilon^2/32}.
\end{align*}
We can therefore conclude by almost the same argument as that leading to~\eqref{Eq:epsilonnALDA} that
\begin{equation}
\label{Eq:epsilonn2}
\mathbb{E}|\epsilon_n| \leq \mathbb{E}\biggl\{\max_{b_2=1,\ldots,B_2}\bigl|R({C}_n^{\mathbf{A}_{1,b_2}}) - {R}_{n}^{\mathbf{A}_{1,b_2}}\bigr|\biggr\} \leq 8 \sqrt{\frac{(d+1)\log(n+1) + 3 \log 2 + \log B_2 + 1}{2n}}.
\end{equation}
Note that none of the bounds~\eqref{Eq:RiskLDA}, \eqref{Eq:epsilonnALDA} and~\eqref{Eq:epsilonn2} depend on the original data dimension $p$.  Moreover,~\eqref{Eq:epsilonn2}, together with Theorem~\ref{Thm:Main}, reveals a trade-off in the choice of $B_2$ when using LDA as the base classifier.  Choosing $B_2$ to be large gives us a good chance of finding a projection with a small estimate of test error, but we may incur a small overfitting penalty as reflected by~\eqref{Eq:epsilonn2}.

Finally, we remark that an alternative method of fitting linear classifiers is via empirical risk minimisation.  In this context, \citet[][Theorem~3.1]{Durrant:13} give high probability bounds on the test error of a linear empirical risk minimisation classifier based on a single random projection, where the bounds depend on what those authors refer to as the `flipping probability', namely the probability that the class assignment of a point based on the projected data differs from the assignment in the ambient space.  In principle, these bounds could be combined with our Theorem~\ref{Thm:BaggingBound}, though the resulting expressions would depend on probabilistic bounds on flipping probabilities.

\subsection{Quadratic Discriminant Analysis}
\label{Sec:QDA}
Quadratic Discriminant Analysis (QDA) is designed to handle situations where the class-conditional covariance matrices are unequal.  Recall that when $X|Y=r \sim N_p(\mu_r, \Sigma_r)$, and $\pi_r = \mathbb{P}(Y=r)$, for $r=0,1$, the Bayes decision boundary is given by $\{x \in \mathbb{R}^p:\Delta(x ;\pi_0, \mu_0, \mu_1, \Sigma_0, \Sigma_1) = 0\}$, where
\begin{align*}
\Delta(x ; \pi_0,\mu_0, \mu_1, \Sigma_0, \Sigma_1) &:= \log\frac{\pi_1}{\pi_0} - \frac 1 2 \log \Bigl(\frac{\det \Sigma_1}{\det \Sigma_0}\Bigr) - \frac 1 2 x^T(\Sigma_1^{-1}-\Sigma_0^{-1})x \\ 
& \hspace{30 pt} + x^T (\Sigma_1^{-1}\mu_1 - \Sigma_0^{-1}\mu_0)  - \frac{1}{2} \mu_1^T\Sigma_1^{-1}\mu_1 +  \frac 1 2 \mu_0^T\Sigma_0^{-1}\mu_0.
\end{align*}
In QDA, $\pi_r$, $\mu_r$ and $\Sigma_r$ are estimated by their sample versions.  If $p \geq \min(n_0,n_1)$, where we recall that $n_r := \sum_{i=1}^n \mathds{1}_{\{Y_i=r\}}$, then at least one of the sample covariance matrix estimates is singular, and QDA cannot be used directly.  Nevertheless, we can still choose $d < \min(n_0,n_1)$ and use QDA as the base classifier in a random projection ensemble.  Specifically, we can set
\[
{C}_n^{A}(x) = {C}_n^{A - \mathrm{QDA}}(x) := \left\{ \begin{array}{ll}
         1 & \mbox{if $\Delta(x ; \hat{\pi}_0,\hat{\mu}_0^A,\hat{\mu}^A_1, \hat{\Sigma}^A_0, \hat{\Sigma}^A_1) \geq 0$};\\
         0 & \mbox{otherwise,}\end{array} \right.
\]
where $\hat{\pi}_r$ and $\hat{\mu}_r^A$ were defined in Section~\ref{sec--LDA}, and where
\[
\hat{\Sigma}_r^A := \frac{1}{n_r-1}\sum_{\{i:Y_i^A = r\}} (AX_i - \hat{\mu}_r^A)(AX_i - \hat{\mu}_r^A)^T
\]
for $r=0,1$.  Unfortunately, analogous theory to that presented in Section~\ref{sec--LDA} does not appear to exist for the QDA classifier; unlike for LDA, the risk does not have a closed form (note that $\Sigma_1-\Sigma_0$ is non-definite in general).  Nevertheless, we found in our simulations presented in Section~\ref{sec--empirical} that the QDA random projection ensemble classifier can perform very well in practice.  In this case, we estimate the test errors using the leave-one-out method given by
\begin{equation}
\label{eq--leave1out}
{R}_{n}^{A} := \frac{1}{n} \sum_{i=1}^n  \mathds{1}_{\{{C}_{n,i}^A(X_i) \neq Y_{i} \}},
\end{equation}
where ${C}_{n,i}^A$ denotes the classifier ${C}^A_n$, trained without the $i$th pair, i.e.\ based on $\mathcal{T}_n^A \setminus \{Z_i^A, Y_i^A\}$.  For a method like QDA that involves estimating more parameters than LDA, we found that the leave-one-out estimator was less susceptible to overfitting than the training error estimator.
  
\subsection{The $k$-nearest neighbour classifier}
\label{Sec:knn}
The $k$-nearest neighbour classifier ($k$nn), first proposed by \citet{Fix:51}, is a nonparametric method that classifies the test point $x \in \mathbb{R}^p$ according to a majority vote over the classes of the $k$ nearest training data points to $x$. The enormous popularity of the $k$nn classifier can be attributed partly due to its simplicity and intuitive appeal; however, it also has the attractive property of being universally consistent: for every fixed distribution $P$, as long as $k\to \infty$ and $k/n \to 0$, the risk of the $k$nn classifier converges to the Bayes risk \citep[][Theorem~6.4]{PTPR:96}.

\citet*{Hall:08} studied the rate of convergence of the excess risk of the $k$-nearest neighbour classifier under regularity conditions that require, inter alia, that $p$ is fixed and that the class-conditional densities have two continuous derivatives in a neighbourhood of the $(p-1)$-dimensional manifold on which they cross.  In such settings, the optimal choice of $k$, in terms of minimising the excess risk, is $O(n^{4/(p+4)})$, and the rate of convergence of the excess risk with this choice is $O(n^{-4/(p+4)})$.  Thus, in common with other nonparametric methods, there is a `curse of dimensionality' effect that means the classifier typically performs poorly in high dimensions. \citet{Samworth:12} found the optimal way of assigning decreasing weights to increasingly distant neighbours, and quantified the asymptotic improvement in risk over the unweighted version, but the rate of convergence remains the same.  

We can use the $k$nn classifier as the base classifier for a random projection ensemble as follows: given $z \in \mathbb{R}^d$, let $(Z^A_{(1)},Y^A_{(1)}), \dots, (Z^A_{(n)}, Y^A_{(n)})$ be a re-ordering of the training data such that $\|Z^A_{(1)}-z\| \leq \ldots \leq \|Z^A_{(n)} - z\|$, with ties split at random.  Now define
\[
{C}_n^A(x) = {C}_n^{A-k\mathrm{nn}}(x) := \left\{ \begin{array}{ll}
         1 & \mbox{if $S_n^A(Ax) \geq 1/2$};\\
         0 & \mbox{otherwise,}\end{array} \right. 
\]         
where $S_n^A(z) : = k^{-1} \sum_{i=1}^k \mathds{1}_{\{Y^A_{(i)} = 1\}}$.  The theory described in the previous paragraph can be applied to show that, under regularity conditions, $\mathbb{E}\{R({C}_n^{A^*})\} - R(C^{\mathrm{Bayes}}) = O(n^{-4/(d+4)})$.

Once again, a natural estimate of the test error in this case is the leave-one-out estimator defined in~\eqref{eq--leave1out}, where we use the same value of $k$ on the leave-one-out datasets as for the base classifier, and where distance ties are split in the same way as for the base classifier.  For this estimator, \citet[][Theorem~4]{DevroyeWagner:1979} showed that for every $n \in \mathbb{N}$,
\[
\sup_{A \in \mathcal{A}} \mathbb{E}[\{R({C}_n^{A}) - {R}_{n}^{A}\}^2] \leq \frac{1}{n} + \frac{24k^{1/2}}{n\sqrt{2\pi}};
\]
see also \citet[Chapter~24]{PTPR:96}.  It follows that 
\[
\mathbb{E}|\epsilon_n^{A^*}| \leq  \Bigl( \frac{1}{n} + \frac{24k^{1/2}}{n\sqrt{2\pi}} \Bigr)^{1/2} \leq \frac{1}{n^{1/2}} + \frac{2^{5/4}\sqrt{3}k^{1/4}}{n^{1/2}\pi^{1/4}}.
\]  
In fact, \citet[][Theorem~1]{DevroyeWagner:1979} also provided a tail bound analogous to~\eqref{Eq:InSampleBound} for the leave-one-out estimator: for every $n \in \mathbb{N}$ and $\epsilon > 0$,
\[
\sup_{A \in \mathcal{A}}\mathbb{P}\{|R(C_n^{A})  - R_n^{A}| > \epsilon\} \leq 2 \exp\Bigl(-\frac{n\epsilon^2}{18}\Bigr) + 6\exp\Bigl(-\frac{n\epsilon^3}{108k(3^d+1)}\Bigr) \leq 8\exp\Bigl(-\frac{n\epsilon^3}{108k(3^d+1)}\Bigr).
\]
Arguing very similarly to Section~\ref{sec--LDA}, we can deduce that under no conditions on the data generating mechanism, and choosing $\epsilon_0 := \bigl\{\frac{108k(3^d+1)}{n}\log(8B_2)\bigr\}^{1/3}$,
\begin{align*}
\mathbb{E}|\epsilon_n| &= \int_0^1 \mathbb{P}\biggl\{\max_{b_2 = 1, \dots, B_2} |R(C_n^{\mathbf{A}_{1,b_2}})  - R_n^{\mathbf{A}_{1, b_2}}| > \epsilon\biggr\} \, d\epsilon \\
&\leq \epsilon_0 + 8B_2 \int_{\epsilon_0}^\infty \exp\Bigl(-\frac{n\epsilon^3}{108k(3^d+1)}\Bigr) \, d\epsilon \leq 3\{4(3^d+1)\}^{1/3}\biggl\{\frac{k(1+\log B_2 + 3 \log 2)}{n}\biggr\}^{1/3}.
\end{align*}
We have therefore again bounded the expectations of the first two terms on the right-hand side of~\eqref{Eq:Main} by quantities that do not depend on $p$. 

\subsection{A general strategy using sample splitting}
In Sections~\ref{sec--LDA},~\ref{Sec:QDA} and~\ref{Sec:knn}, we focused on specific choices of the base classifier to derive bounds on the expected values of the first two terms in the bound in Theorem~\ref{Thm:Main}.  The aim of this section, on the other hand, is to provide similar guarantees that can be applied for any choice of base classifier in conjunction with a sample splitting strategy.  The idea is to split the sample $\mathcal{T}_n$ into $\mathcal{T}_{n,1}$ and $\mathcal{T}_{n,2}$, say, where $|\mathcal{T}_{n,1}| =: n^{(1)}$ and $|\mathcal{T}_{n,2}| =: n^{(2)}$.  To estimate the test error of ${C}_{n^{(1)}}^A$, the projected data base classifier trained on $\mathcal{T}_{n,1}^A := \{(Z_i^A,Y_i^A):(X_i,Y_i) \in \mathcal{T}_{n,1}\}$, we use 
\[
{R}_{n^{(1)},n^{(2)}}^A := \frac{1}{n^{(2)}} \sum_{(X_i,Y_i) \in \mathcal{T}_{n,2}} \mathds{1}_{\{{C}_{n^{(1)}}^{A} (X_i) \neq Y_i\}};
\]
in other words, we construct the classifier based on the projected data from $\mathcal{T}_{n,1}$, and count the proportion of points in $\mathcal{T}_{n,2}$ that are misclassified.  Similar to our previous approach, for the $b_1$th group of projections, we then select a projection $\mathbf{A}_{b_1}$ that minimises this estimate of test error, and construct the random projection ensemble classifier $C_{n^{(1)},n^{(2)}}^{\mathrm{RP}}$ from 
\[
\nu_{n^{(1)}}(x) := \frac{1}{B_1}\sum_{b_1=1}^{B_1} \mathds{1}_{\{{C}_{n^{(1)}}^{\mathbf{A}_{b_1}}(x) = 1 \}}.
\]
Writing $R_{n^{(1)},n^{(2)}}^* := \min_{A \in \mathcal{A}} R_{n^{(1)},n^{(2)}}^A$, we introduce the following assumption analogous to assumption~2:
\begin{description}
\item \textit{Assumption 2'.} There exists $\beta \in (0,1]$ such that
\[
\mathbf{P}\bigl({R}_{n^{(1)}, n^{(2)}}^{\mathbf{A}_{1,1}} \leq {R}_{n^{(1)},n^{(2)}}^* + |\epsilon_{n^{(1)},n^{(2)}}| \bigr) \geq \beta,
\]
where $\epsilon_{n^{(1)},n^{(2)}} := \mathbf{E}\{R({C}_{n^{(1)}}^{\mathbf{A}_1}) - {R}_{n^{(1)},n^{(2)}}^{\mathbf{A}_1}\}$.
\end{description}
The following bound for the random projection ensemble classifier with sample splitting is then immediate from Theorem~\ref{Thm:Main} and Proposition~\ref{Prop:A^*}.
\begin{corollary} 
\label{Thm:Mainsamplesplit}
Assume assumptions~2'~and~3.  Then, for each $B_1, B_2 \in \mathbb{N}$,
\begin{align*}
\mathbf{E}\{R({C}_{n^{(1)},n^{(2)}}^{\mathrm{RP}})\} - R(C^{\mathrm{Bayes}}) & \leq \frac{R({C}_{n^{(1)}}^{A^*}) - R(C^{A^*-\mathrm{Bayes}})}{\min(\alpha,1-\alpha)}  + \frac{2|\epsilon_{n^{(1)},n^{(2)}}| - \epsilon_{n^{(1)}, n^{(2)}}^{A^*}}{\min(\alpha,1-\alpha)}  + \frac{(1-\beta)^{B_{2}}}{\min(\alpha,1-\alpha)},
\end{align*}
where $\epsilon_{n^{(1)}, n^{(2)}}^{A^*} := R({C}_{n^{(1)}}^{A^*}) - {R}_{n^{(1)},n^{(2)}}^{A^*}$.
\end{corollary}
The main advantage of this approach is that, since $\mathcal{T}_{n,1}$ and $\mathcal{T}_{n,2}$ are independent, we can apply Hoeffding's inequality to deduce that
\[
\sup_{A \in \mathcal{A}} \mathbb{P}\bigl\{|R({C}_{n^{(1)}}^A) - {R}_{n^{(1)},n^{(2)}}^A| \geq \epsilon \bigm| \mathcal{T}_{n,1}\bigr\} \leq 2e^{-2n^{(2)}\epsilon^2}.
\]
It then follows by very similar arguments to those given in Section~\ref{sec--LDA} that
\begin{align}
\label{Eq:SampleSplit}
\mathbb{E}(|\epsilon_{n^{(1)},n^{(2)}}^{A^*}| \bigm| \mathcal{T}_{n,1}) = \mathbb{E}\bigl\{|R({C}_{n^{(1)}}^{A^*}) - {R}_{n^{(1)},n^{(2)}}^{A^*}|\bigm|\mathcal{T}_{n,1}\bigr\} &\leq \Bigl(\frac{1 + \log 2}{2n^{(2)}}\Bigr)^{1/2}, \nonumber \\
\mathbb{E}(|\epsilon_{n^{(1)},n^{(2)}}| \bigm| \mathcal{T}_{n,1}) =  \mathbb{E}\bigl\{|R({C}_{n^{(1)}}^{\mathbf{A}_1}) - {R}_{n^{(1)},n^{(2)}}^{\mathbf{A}_1}|\bigm|\mathcal{T}_{n,1}\bigr\} &\leq \Bigl(\frac{1 + \log 2 + \log B_2}{2n^{(2)}}\Bigr)^{1/2}.
\end{align}
These bounds hold for any choice of base classifier (and still without any assumptions on the data generating mechanism); moreover, since the bounds on the terms in~\eqref{Eq:SampleSplit} merely rely on Hoeffding's inequality as opposed to Vapnik--Chervonenkis theory, they are typically sharper.  The disadvantage is that the first term in the bound in Corollary~\ref{Thm:Mainsamplesplit} will typically be larger than the corresponding term in Theorem~\ref{Thm:Main} due to the reduced effective sample size.  

\section{Practical considerations}
\label{sec--prac}

\subsection{Computational complexity}
\label{sec--compcomp}
The random projection ensemble classifier aggregates the results of applying a base classifier to many random projections of the data.  Thus we need to compute the training error (or leave-one-out error) of the base classifier after applying each of the $B_1B_2$ projections.  The test point is then classified using the $B_1$ projections that yield the minimum error estimate within each block of size $B_{2}$.  

Generating a random projection from Haar measure involves computing the left singular vectors of a $p \times d$ matrix, which requires $O(pd^2)$ operations \citep[Lecture~31]{Trefethen:97}.  However, if computational cost is a concern, one may simply generate a matrix with $pd$ independent $N(0,1/p)$ entries.  If $p$ is large, such a matrix will be approximately orthonormal with high probability.  In fact, when the base classifier is affine invariant (as is the case for LDA and QDA), this will give the same results as using Haar projections, in which case one can forgo the orthonormalisation step altogether when generating the random projections.  Even in very high-dimensional settings, multiplication by a random Gaussian matrix can be approximated in a computationally efficient manner \citep[e.g.][]{LSS2013}.  Once a projection is generated, we need to map the training data to the lower dimensional space, which requires $O(npd)$ operations.  We then classify the training data using the base classifier. The cost of this step, denoted $T_{\mathrm{train}}$, depends on the choice of the base classifier; for LDA and QDA it is $O(nd^2)$, while for $k$nn it is $O(n^{2}d)$. Finally the test points are classified using the chosen projections; this cost, denoted $T_{\mathrm{test}}$, also depends on the choice of base classifier. For LDA, QDA and $k$nn it is $O(n_{\mathrm{test}}d)$, $O(n_{\mathrm{test}}d^{2})$ and $O((n+n_{\mathrm{test}})^{2}d)$, respectively, where $n_{\mathrm{test}}$ is the number of test points.  Overall we therefore have a cost of $O(B_1 \{ B_2 (npd + T_{\mathrm{train}}) + n_{\mathrm{test}}pd +  T_{\mathrm{test}}\})$ operations.   

An appealing aspect of the proposal studied here is the scope to incorporate parallel computing.  We can simultaneously compute the projected data base classifier for each of the $B_1B_2$ projections.  In the supplementary material we present the running times of the random projection ensemble classifier and the other methods considered in the empirical comparison in Section~\ref{sec--empirical}.

\subsection{Choice of $\alpha$}
\label{sec--alpha}
We now discuss the choice of the voting threshold $\alpha$.  In~\eqref{eq--alpha},  at the end of Section~\ref{sec--RP}, we defined the oracle choice $\alpha^{*}$, which minimises the test error of the infinite-simulation random projection classifier.  Of course, $\alpha^*$ cannot be used directly, because we do not know $G_{n,0}$ and $G_{n,1}$ (and we may not know $\pi_0$ and $\pi_1$ either).  Nevertheless, for the LDA base classifier we can estimate $G_{n,r}$ using 
\[
\hat{G}_{n,r}(t) := \frac{1}{n_r}\sum_{\{i:Y_i=r\}} \mathds{1}_{\{\nu_n(X_i) < t\}}
\]
for $r=0,1$.  For the QDA and $k$-nearest neighbour base classifiers, we use the leave-one-out-based estimate $\tilde{\nu}_n(X_{i}) :=   B_1^{-1}\sum_{b_1=1}^{B_1} \mathds{1}_{\{{C}_{n,i}^{\mathbf{A}_{b_1}}(X_{i}) = 1 \}}$ in place of $\nu_n(X_i)$.  We also estimate $\pi_r$ by $\hat{\pi}_r := n^{-1}\sum_{i=1}^n \mathds{1}_{\{Y_i = r\}}$, and then set the cut-off in~(\ref{eq--RP}) as
\begin{equation}
\label{Eq:alphahat}
\hat{\alpha} \in \argmin_{\alpha' \in [0,1]} \bigl[\hat{\pi}_1 \hat{G}_{n,1}(\alpha') + \hat{\pi}_0\{1 - \hat{G}_{n,0}(\alpha')\}\bigr].
\end{equation}
Since empirical distribution functions are piecewise constant, the objective function in~\eqref{Eq:alphahat} does not have a unique minimum, so we choose $\hat{\alpha}$ to be the average of the smallest and largest minimisers.  An attractive feature of the method is that $\{\nu_n(X_i):i=1,\ldots,n\}$ (or $\{\tilde{\nu}_n(X_i):i=1,\ldots,n\}$ in the case of QDA and $k$nn) are already calculated in order to choose the projections, so calculating $\hat{G}_{n,0}$ and $\hat{G}_{n,1}$ carries negligible extra computational cost. 

Figure~\ref{fig--GPlots0.5} illustrates $\hat{\pi}_1 \hat{G}_{n,1}(\alpha') + \hat{\pi}_0\{1 - \hat{G}_{n,0}(\alpha')\}$ as an estimator of $\pi_1 G_{n,1}(\alpha') + \pi_0\{1 - G_{n,0}(\alpha')\}$, for the QDA base classifier and for different values of $n$ and $\pi_1$.  Here, a very good approximation to the estimand was obtained using an independent data set of size 5000.  Unsurprisingly, the performance of the estimator improves as $n$ increases, but the most notable feature of these plots is the fact that for all classifiers and even for small sample sizes, $\hat{\alpha}$ is an excellent estimator of $\alpha^*$, and may be a substantial improvement on the naive choice $\hat{\alpha} = 1/2$ (or the appropriate prior weighted choice), which may result in a classifier that assigns every point to a single class.  

\begin{figure}[h]
  \centering
  \makebox{\includegraphics[width=0.9\textwidth]{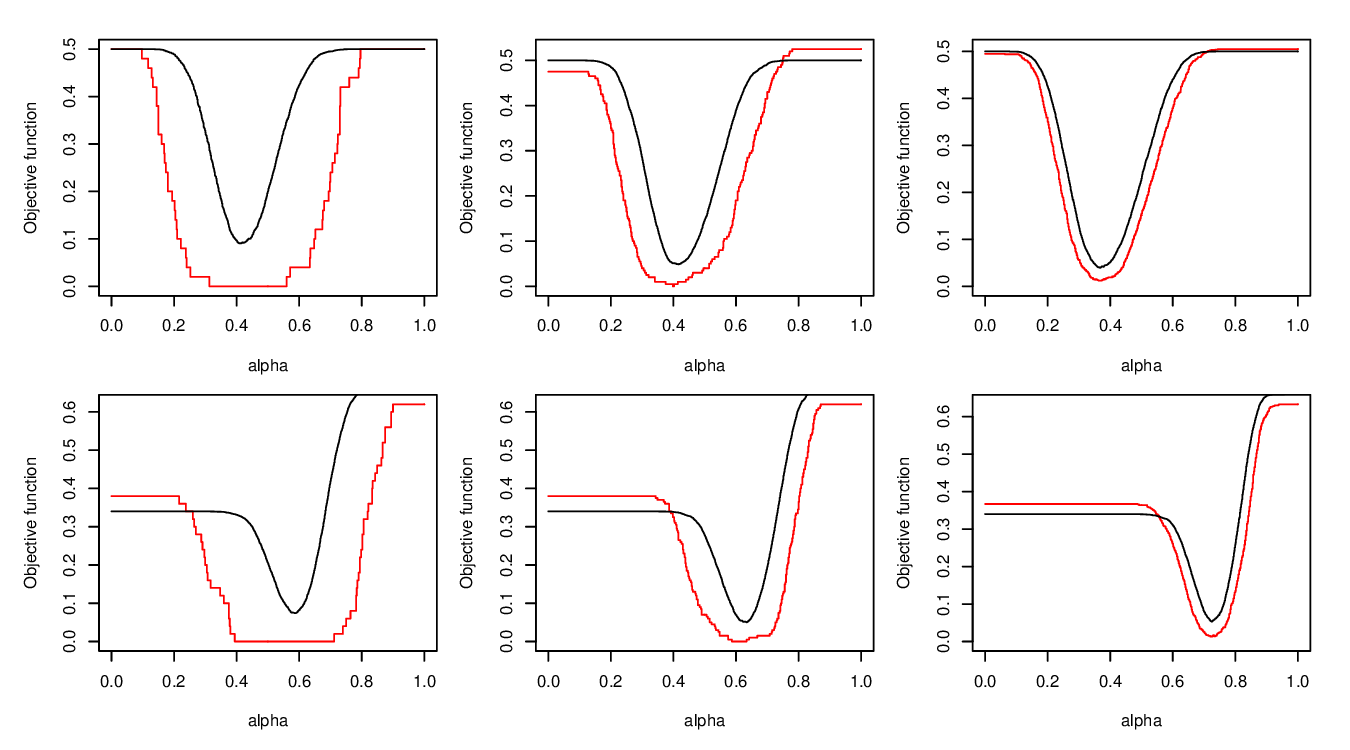}}
 \caption{ \label{fig--GPlots0.5} $\pi_1 G_{n,1}(\alpha') + \pi_0\{1 - G_{n,0}(\alpha')\}$ in~\eqref{eq--alpha} (black) and $\hat{\pi}_1 \hat{G}_{n,1}(\alpha') + \hat{\pi}_0\{1 - \hat{G}_{n,0}(\alpha')\}$ (red) for the QDA base classifier after projecting for one training data set of size $n =  50$ (left), $200$ (middle)  and $1000$ (right) from Model~3 with $\pi_1 = 0.5$ (top) and $\pi_1 = 0.66$ (bottom). Here, $p =100$ and $d=2$.} 
\end{figure}

\subsection{Choice of $B_1$ and $B_2$}
\label{sec--B1B2}
In order to minimise the Monte Carlo error as described in Theorem~\ref{thm--Bvar} and Proposition~\ref{Prop:variance}, we should choose $B_1$ to be as large as possible.  The constraint, of course, is that the computational cost of the random projection classifier scales linearly with $B_1$.  The choice of $B_2$ is more subtle; while the third term in the bound in Theorem~\ref{Thm:Main} decreases as $B_2$ increases, we saw in Section~\ref{sec--base} that upper bounds on $\mathbb{E}|\epsilon_n|$ may increase with $B_2$.   In principle, we could try to use the expressions given in Theorem~\ref{Thm:Main} and Section~\ref{sec--base} to choose $B_2$ to minimise the overall upper bound on $\mathbf{E}\{R({C}_n^{\mathrm{RP}})\} - R(C^{\mathrm{Bayes}})$.    In practice, however, we found that an involved approach such as this was unnecessary, and that the ensemble method was robust to the choice of $B_1$ and $B_2$; see Section~\ref{sec--B1B2choice} of the supplementary material for numerical evidence and further discussion.  Based on this numerical work, we recommend $B_1 = 500$ and $B_2 = 50$ as sensible default choices, and indeed these values were used in all of our experiments in Section~\ref{sec--empirical} as well as Section~\ref{Sec:FurtherSim} in the supplementary material.

\subsection{Choice of $d$}
\label{sec--d}
We want to choose $d$ as small as possible in order to obtain the best possible performance bounds as described in Section~\ref{sec--base} above.  This also reduces the computational cost.  However, the performance bounds rely on assumption~3, whose strength decreases as $d$ increases, so we want to choose $d$ large enough that this condition holds (at least approximately).

In Section \ref{sec--empirical} we see that the random projection ensemble method is quite robust to the choice of~$d$.  Nevertheless, in some circumstances it may be desirable to have an automatic choice, and cross-validation provides one possible approach when computational cost at \textit{training time} is not too constrained.  Thus, if we wish to choose $d$ from a set $\mathcal{D} \subseteq \{1,\dots, p\}$, then for each $d \in \mathcal{D}$, we train the random projection ensemble classifier, and set
\[
\hat{d} := \sargmin_{d \in \mathcal{D}} \bigl[\hat{\pi}_1 \hat{G}_{n,1}(\hat{\alpha}) + \hat{\pi}_0\{1 - \hat{G}_{n,0}(\hat{\alpha})\}\bigr],
\]
where $\hat{\alpha} = \hat{\alpha}_d$ is given in \eqref{Eq:alphahat}.  Such a proceedure does not add to the computational cost at \textit{test time}.  This strategy is most appropriate when $\max\{d:d \in \mathcal{D}\}$ is not too large (which is the setting we have in mind); otherwise a penalised risk approach may be more suitable.

\section{Empirical analysis}
\label{sec--empirical}
In this section, we assess the empirical performance of the random projection ensemble classifier in simulated and real data experiments.  We will write RP-LDA$_d$, RP-QDA$_d$ and RP-$k$nn$_d$ to denote the random projection classifier with LDA, QDA, and $k$nn base classifiers, respectively; the subscript $d$ refers to the dimension of the image space of the projections.  

For comparison we present the corresponding results of applying, where possible, the three base classifiers (LDA, QDA, $k$nn) in the original $p$-dimensional space alongside 11 other classification methods chosen to represent the state of the art.  These include Random Forests (RF) \citep{Breiman2001}; Support Vector Machines (SVM) \citep{Cortes:95}; Gaussian Process (GP) classifiers \citep{Williams:98}; and three methods designed for high-dimensional classification problems, namely Penalized LDA (PenLDA) \citep{Witten:11}, Nearest Shrunken Centroids (NSC) \citep{Tibshirani:03}, and $\ell_{1}$-penalised logistic regression (PenLog) \citep{Goeman:15}.  

A further comparison is with LDA and $k$nn applied after a single projection chosen based on the sufficient dimension reduction assumption (SDR5).  For this method, we project the data into 5 dimensions using the proposal of \citet{Shin:2014}.  This method requires $n > p$.  Finally, we compare with two related ensemble methods: optimal tree ensembles (OTE) \citep{Khan:15} and ensemble of subset of $k$-nearest neighbour classifiers (ES$k$nn) \citep{Gul:16}.  

Many of these methods require tuning parameter selection, and the parameters were chosen as follows: for the standard $k$nn classifier, we chose $k$ via leave-one-out cross validation from the set $\{3,5,7,9,11\}$.  The Random Forest was implemented using the \texttt{randomForest} package \citep{Liaw:2014}; we used an ensemble of 1000 trees, with $\lfloor \sqrt{p} \rfloor$ (the default setting in the \texttt{randomForest} package) components randomly selected when training each tree.    For the Radial SVM, we used the reproducing basis kernel $K(u,v) := \exp( -\frac{1}{p}\|u-v\|^2)$.  Both SVM classifiers were implemented using the \texttt{svm} function in the \texttt{e1071} package \citep{Meyer:2015}.  The GP classifier uses a radial basis function, with the hyperparameter chosen via the automatic method in the \texttt{gausspr} function in the \texttt{kernlab} package \citep{Karatzoglou:2015}.  The tuning parameters for the other methods were chosen using the default settings in the corresponding \texttt{R} packages \texttt{PenLDA} \citep{Witten:2011}, \texttt{NSC} \citep{Hastie:2015} and \texttt{penalized} \citep{Goeman:15} namely 6-fold, 10-fold and 5-fold cross validation, respectively.  For the OTE and ES$k$nn methods we used the default settings in the \texttt{R} packages \texttt{OTE} \citep{Khan:15a} and \texttt{ESKNN} \citep{Gul:15}.

\subsection{Simulated examples}
\label{sec:simulated}
We present four different simulation settings chosen to investigate the performance of the random projection ensemble classifier in a wide variety of scenarios.  
In each of the examples below, we take $n \in \{50, 200, 1000\}$,  $p \in \{100, 1000\}$ and investigate two different values of the prior probability.  We use Gaussian projections (cf.\ Section~\ref{sec--compcomp}) and set $B_1 = 500$ and $B_2 = 50$ (cf.\ Section~\ref{sec--B1B2}).  

The risk estimates and standard errors for the $p=100$ and $\pi_{1} = 0.5$ case are shown in Tables~\ref{tab--rotateddist100.5}~and~\ref{tab--indep100.5} (the remaining results are given in the supplementary material).  These were calculated as follows:  We set $n_{\mathrm{test}} = 1000$, $N_{\mathrm{reps}} = 100$, and for $l = 1, \ldots, N_{\mathrm{reps}}$ we generate a training set of size $n$ and a test set of size $n_{\mathrm{test}}$ from the same distribution.   Let $\hat{R}_l$ be the proportion of the test set that is classified incorrectly in the $l$th repeat of the experiment.  The overall risk estimate presented is $\widehat{\mathrm{Risk}} :=\frac{1}{N_{\mathrm{reps}}} \sum_{l=1}^{N_{\mathrm{reps}} } \hat{R}_l$.  Note that
\[
\mathbb{E}\{\widehat{\mathrm{Risk}}\} = \mathbb{E}\{R({C}_n^{\mathrm{RP}})\}
\]
and
\begin{align*}
\mathrm{Var} (\widehat{\mathrm{Risk}}) &= \frac{1}{N_{\mathrm{reps}}} \mathrm{Var}(\hat{R}_1)
\\ & = \frac{1}{N_{\mathrm{reps}}} \biggl[ \mathbb{E}\biggl\{\frac{\mathbf{E}\{R({C}_n^{\mathrm{RP}})\}[1-\mathbf{E}\{R({C}_n^{\mathrm{RP}})\}]}{n_{\mathrm{test}}} \biggr\}  + \mathrm{Var}\bigl[\mathbf{E}\{R({C}_n^{\mathrm{RP}})\}\bigr]\biggr].
\end{align*}
We therefore estimate the standard error in the tables below by
\[ 
\hat{\sigma} := \frac{1}{N_{\mathrm{reps}}^{1/2}} \biggl\{\frac{\widehat{\mathrm{Risk}}(1-\widehat{\mathrm{Risk}})}{n_{\mathrm{test}}}   + \frac{n_{\mathrm{test}}-1}{n_{\mathrm{test}}N_{\mathrm{reps}}} \sum_{l = 1}^{N_{\mathrm{reps}}} (\hat{R}_l - \widehat{\mathrm{Risk}})^2 \biggr\}^{1/2}.
\]
The method with the smallest risk estimate in each column of the tables below is highlighted in bold; where applicable, we also highlight any method with a risk estimate within one standard error of the minimum. 

\subsubsection{Sparse class boundaries}
\label{sec--multisims}
%101:120
Model 1: Here, $X |\{Y=0\} \sim \frac{1}{2} N_{p}(\mu_{0}, \Sigma) + \frac{1}{2} N_{p}(-\mu_{0}, \Sigma)$, and $X |\{Y=1\} \sim \frac{1}{2} N_{p}(\mu_{1}, \Sigma ) + \frac{1}{2} N_{p}(-\mu_{1}, \Sigma)$, where, for $p=100$, we set $\Sigma = I_{100 \times 100}$, $\mu_0 = (2,-2,0, \ldots,0)^T$ and $\mu_1 = (2,2,0, \ldots,0)^T$. 

In Model~1,  assumption~3 holds with $d=2$; for example, we could take the rows of $A^{*}$ to be the first two Euclidean basis vectors.   We see that the RP ensemble classifier with the QDA base classifier performs very well here, as does the OTE method.   Despite the fact that the regression function $\eta$ only depends on the first two components in this example, the comparators designed for sparse problems do not perform well; in some cases they are no better than a random guess.

\subsubsection{Rotated Sparse Normal}
\label{sec--rotated}
%%1:20
Model 2: Here, $X |\{Y=0\} \sim N_{p}(\Omega_{p}\mu_{0}, \Omega_{p}\Sigma_{0}\Omega_{p}^{T} )$, and $X |\{Y=1\} \sim N_{p}(\Omega_{p}\mu_{1}, \Omega_{p}\Sigma_{1}\Omega_{p}^{T})$, where $\Omega_{p}$ is a $p \times p$ rotation matrix that was sampled once according to Haar measure, and remained fixed thereafter, and we set $\mu_0 = (3,3,3, 0, \dots,0)^T$ and $\mu_1 = (0, \ldots,0)^T$.  Moreover, $\Sigma_0$ and $\Sigma_1$ are block diagonal, with blocks $\Sigma_{r}^{(1)}$, and $\Sigma_{r}^{(2)}$, for $r=0,1$, where $\Sigma_{0}^{(1)}$ is a $3\times3$ matrix with diagonal entries equal to 2 and off-diagonal entries equal to 1/2, and $\Sigma_{1}^{(1)} = \Sigma_{0}^{(1)} - I_{3\times3}$.  In both classes $\Sigma_{r}^{(2)}$ is a $(p-3) \times (p-3)$ matrix, with diagonal entries equal to 1 and off-diagonal entries equal to 1/2.

In Model~2, assumption~3 holds with $d=3$; for instance, $A^*$ can be taken to be the first three rows of $\Omega_{p}^{T}$.  Perhaps surprisingly, whether we use too small a value of $d$ (namely $d=2$), or one that is too large ($d=5$), the RP ensemble methods still classify very well.

\begin{table}
 %1:5new
  \caption{  \label{tab--rotateddist100.5} Misclassification rates for Models 1 and 2, with $p = 100$ and $\pi_1 = 0.5$.}
  \centering
\fbox{%
  \begin{tabular}{l|c  c c | c c c  }
       &\multicolumn{3}{c|}{Model 1, Bayes risk = 4.45   }  &\multicolumn{3}{c}{Model 2, Bayes risk = 4.09 }  \\
  \multicolumn{1}{r|}{$n$}  &$50$  &$200$  & $1000$  &$50$  &$200$  & $1000$  \\
    \hline
RP-LDA$_{2}$&$ 49.34 _{ 0.26 }$&$ 48.10 _{ 0.31 }$&$ 44.14 _{ 0.46 }$&$ 8.34 _{ 0.28 }$&$\mathbf{5.56} _{ 0.12 }$&$\mathbf{5.17} _{ 0.10 }$\\
RP-LDA$_{5}$&$ 49.81 _{ 0.24 }$&$ 48.86 _{ 0.30 }$&$ 46.91 _{ 0.40 }$&$\mathbf{8.17}_{ 0.27 }$&$\mathbf{5.64} _{ 0.13 }$&$\mathbf{5.14} _{ 0.10 }$\\
RP-QDA$_{2}$&$ 44.18 _{ 0.29 }$&$ 29.38 _{ 0.49 }$&$ 10.57 _{ 0.22 }$&$ 8.40 _{ 0.29 }$&$\mathbf{5.57} _{ 0.12 }$&$\mathbf{5.16} _{ 0.10 }$\\
RP-QDA$_{5}$&$ 39.32 _{ 0.33 }$&$\mathbf{22.32} _{ 0.32 }$&$  \mathbf{8.75} _{ 0.15 }$&$\mathbf{8.06}_{ 0.25 }$&$\mathbf{5.58} _{ 0.12 }$&$\mathbf{5.09}_{ 0.10 }$\\
RP-$k$nn$_{2}$&$ 46.10 _{ 0.30 }$&$ 36.18 _{ 0.32 }$&$ 19.42 _{ 0.20 }$&$  8.94 _{ 0.36 }$&$\mathbf{5.60} _{ 0.12 }$&$ 5.20 _{ 0.10 }$\\
RP-$k$nn$_{5}$&$ 43.65 _{ 0.30 }$&$ 25.34 _{ 0.35 }$&$ 10.21 _{ 0.16 }$&$  9.00 _{ 0.33 }$&$\mathbf{5.68}_{ 0.12 }$&$\mathbf{5.13} _{ 0.10 }$\\
LDA          &N/A          &$ 49.60 _{ 0.23 }$&$ 49.91 _{ 0.22 }$&N/A                &$14.32 _{ 0.22 }$&$ 6.34 _{ 0.11 }$\\
QDA &N/A                    &N/A            &$ 27.36 _{ 0.23 }$&N/A& N/A                                                 &$17.10 _{ 0.20 }$\\
$k$nn        &$\mathbf{34.66} _{ 0.35 }$&$ 23.71 _{ 0.31 }$&$ 15.31 _{ 0.17 }$&$ 12.81 _{ 0.28 }$&$ 8.80 _{ 0.15 }$&$ 7.28 _{ 0.13 }$\\
RF           &$ 49.72 _{ 0.23 }$&$ 48.33 _{ 0.25 }$&$ 43.28 _{ 0.43 }$&$ 11.11 _{ 0.31 }$&$ 6.80 _{ 0.12 }$&$ 6.07 _{ 0.11 }$\\
Radial SVM   &$ 49.83 _{ 0.22 }$&$ 50.16 _{ 0.22 }$&$ 48.67 _{ 0.22 }$&$ 24.04 _{ 1.47 }$&$ 6.37 _{ 0.14 }$&$ 5.46 _{ 0.10 }$\\
Linear SVM   &$ 50.02 _{ 0.23 }$&$ 49.55 _{ 0.21 }$&$ 50.04 _{ 0.22 }$&$  9.41 _{ 0.21 }$&$ 8.96 _{ 0.17 }$&$ 7.76 _{ 0.13 }$\\
Radial GP    &$ 48.18 _{ 0.30 }$&$ 42.76 _{ 0.29 }$&$ 26.60 _{ 0.24 }$&$ 14.09 _{ 0.63 }$&$ 5.84 _{ 0.13 }$&$\mathbf{5.09}_{ 0.10 }$\\
PenLDA       &$ 49.95 _{ 0.23 }$&$ 49.79 _{ 0.23 }$&$ 50.05 _{ 0.22 }$&$ 11.11 _{ 0.55 }$&$ 6.72 _{ 0.20 }$&$ 5.79 _{ 0.12 }$\\
NSC          &$ 49.74 _{ 0.23 }$&$ 49.69 _{ 0.26 }$&$ 49.55 _{ 0.24 }$&$ 12.61 _{ 0.61 }$&$ 7.27 _{ 0.28 }$&$ 5.82 _{ 0.13 }$\\
PenLog       &$ 49.66 _{ 0.35 }$&$ 49.88 _{ 0.24 }$&$ 50.12 _{ 0.21 }$&$ 11.37 _{ 0.22 }$&$ 7.67 _{ 0.14 }$&$ 6.00 _{ 0.11 }$\\
SDR5-LDA& N/A&$ 37.80 _{ 0.48 }$&$ 35.31 _{ 0.30 }$&N/A                &$15.07 _{ 0.22 }$&$ 6.47 _{ 0.11 }$\\
SDR5-$k$nn& N/A &$ 32.22 _{ 0.71 }$&$ 21.83 _{ 1.08 }$&N/A              &$18.81 _{ 0.29 }$&$ 7.75 _{ 0.12 }$\\
OTE          &$ 48.51 _{ 0.33 }$&$ 34.73 _{ 1.23 }$&$  9.57 _{ 0.66 }$&$ 18.26 _{ 0.47 }$&$12.44 _{ 0.26 }$&$ 9.24 _{ 0.15 }$\\
ES$k$nn      &$ 50.13 _{ 0.23 }$&$ 49.87 _{ 0.22 }$&$ 49.77 _{ 0.21 }$&$ 40.30 _{ 0.71 }$&$37.06 _{ 0.63 }$&$32.98 _{ 0.58 }$\\
\end{tabular}}
\end{table}

\subsubsection{Independent features}
\label{sec--independent1}
%%21:40
Model 3: Here, $P_0 = N_p(\mu,I_{p \times p})$, with $\mu = \frac{1}{\sqrt{p}}(1, \dots, 1, 0, \ldots, 0)^T$, where $\mu$ has $p/2$ non-zero components, while $P_1$ is the distribution of $p$ independent components, each with a standard Laplace distribution.

In Model~3, the class boundaries are non-linear and, in fact, assumption~3 is not satisfied for any $d < p$. Nevertheless, in Table~\ref{tab--indep100.5}, we see that where the LDA, QDA and $k$nn classifiers are tractable, they are outperformed by their random projection ensemble counterparts, and in fact the RP-QDA$_5$ classifier has the smallest misclassification rate among all methods implemented.  Unsurprisingly, the methods that are designed for a linear Bayes decision boundary are not effective.   The RP-QDA classifiers are especially accurate here; in particular, they are able to cope better with the non-linearity of the class boundaries than the RP-LDA classifiers.

\subsubsection{$t$-distributed features}
%%61:80
\label{sec--tsims}
Model 4: Here, $X |\{Y=r\} = \mu_r + \frac{Z_{r}}{\sqrt{U_{r}/\nu_r}}$, where $Z_{r} \sim N_p(0,\Sigma_r)$ independent of $U_{r}\sim \chi^2_{\nu_r}$, for $r = 0,1$. That is, $P_r$ is the multivariate $t$-distribution centred at $\mu_r$, with $\nu_r$ degrees of freedom and shape parameter $\Sigma_r$. We set $\mu_0 = (1,\dots,1,0, \dots,0 )^T$, where $\mu_0$ has 10 non-zero components, $\mu_1 = 0$, $\nu_0=2$, $\nu_1=1$, $\Sigma_0 = (\Sigma_{j,k})$, where $\Sigma_{j,j} = 1$, $\Sigma_{j,k} = 0.5$ if $\max(j,k) \leq 10$ and $j \neq k$, $\Sigma_{j,k} = 0$ otherwise, and $\Sigma_1 = I_{p \times p}$.

\begin{table}
 %21:25new
  \caption{  \label{tab--indep100.5} Misclassification rates for Models 3 and 4, with $p = 100$ and $\pi_{1} = 0.5$.}
  \centering
\fbox{%
  \begin{tabular}{l|c c c | c c c  }
       &\multicolumn{3}{c|}{Model 3, Bayes risk = 1.01  }  &\multicolumn{3}{c|}{Model 4, Bayes risk = 12.68  } \\
  \multicolumn{1}{r|}{$n$}  &$50$      &$200$           & $1000$   &$50$      &$200$           & $1000$ \\
    \hline
RP-LDA$_{2}$  &$ 45.11 _{ 1.03 }$&$ 44.05 _{ 0.98 }$&$ 39.22 _{ 0.89 }$  &$ 38.06 _{ 0.71 }$&$ 38.45 _{ 0.92 }$&$ 40.48 _{ 0.84 }$\\
RP-LDA$_{5}$  &$ 45.58 _{ 0.60 }$&$ 44.46 _{ 0.58 }$&$ 41.08 _{ 0.56 }$ &$ 34.84 _{ 0.63 }$&$ 32.43 _{ 0.75 }$&$ 35.09 _{ 0.89 }$\\
RP-QDA$_{2}$  &$ 11.41 _{ 0.62 }$&$  4.83 _{ 0.15 }$&$  3.85 _{ 0.09 }$ &$ 42.12 _{ 0.47 }$&$ 41.99 _{ 0.28 }$&$ 42.37 _{ 0.21 }$\\
RP-QDA$_{5}$  &$\mathbf{9.71}_{ 0.52 }$&$\mathbf{4.23}_{ 0.14 }$&$\mathbf{3.29}_{ 0.08 }$ &$ 42.13 _{ 0.35 }$&$ 42.04 _{ 0.27 }$&$ 42.59 _{ 0.21 }$\\
RP-$k$nn$_{2}$&$ 20.69 _{ 0.84 }$&$  6.86 _{ 0.27 }$&$  4.73 _{ 0.11 }$&$ 30.85 _{ 0.49 }$&$ 24.07 _{ 0.31 }$&$ 20.76 _{ 0.19 }$\\
RP-$k$nn$_{5}$&$ 21.30 _{ 0.54 }$&$  6.91 _{ 0.18 }$&$  3.78 _{ 0.10 }$&$\mathbf{29.85}_{ 0.46 }$&$ 24.02 _{ 0.30 }$&$ 20.81 _{ 0.21 }$\\
LDA          &N/A                                &$ 46.22 _{ 0.25 }$&$ 41.74 _{ 0.24 }$&N/A                          &$ 37.34 _{ 0.29 }$&$ 31.04 _{ 0.26 }$\\
QDA          &N/A&N/A                                                      &$ 15.30 _{ 0.21 }$&N/A&N/A                                                &$ 40.90 _{ 0.21 }$\\
$k$nn        &$ 49.92 _{ 0.24 }$&$ 49.81 _{ 0.22 }$&$ 49.67 _{ 0.22 }$  &$ 37.49 _{ 0.63 }$&$ 30.14 _{ 0.34 }$&$ 27.58 _{ 0.25 }$\\
RF           &$ 44.79 _{ 0.34 }$&$ 23.38 _{ 0.30 }$&$  7.72 _{ 0.16 }$      &$30.97_{ 0.60 }$&$\mathbf{20.46}_{ 0.21 }$&$\mathbf{18.69}_{ 0.17 }$\\
Radial SVM   &$ 39.34 _{ 1.47 }$&$  4.65 _{ 0.13 }$&$  3.43 _{ 0.09 }$&$ 47.72 _{ 0.40 }$&$ 45.46 _{ 0.51 }$&$ 43.70 _{ 0.72 }$\\
Linear SVM   &$ 46.57 _{ 0.26 }$&$ 46.17 _{ 0.24 }$&$ 41.67 _{ 0.26 }$&$ 36.79 _{ 0.57 }$&$ 34.21 _{ 0.56 }$&$ 31.87 _{ 0.71 }$\\
Radial GP    &$ 48.87 _{ 0.31 }$&$ 45.47 _{ 0.37 }$&$ 36.18 _{ 0.27 }$ &$ 38.39 _{ 0.84 }$&$ 26.63 _{ 0.44 }$&$ 22.77 _{ 0.20 }$\\
PenLDA       &$ 46.04 _{ 0.26 }$&$ 44.48 _{ 0.26 }$&$ 41.71 _{ 0.23 }$  &$ 45.64 _{ 0.44 }$&$ 45.22 _{ 0.53 }$&$ 45.39 _{ 0.47 }$\\
NSC          &$ 47.47 _{ 0.33 }$&$ 45.99 _{ 0.34 }$&$ 42.31 _{ 0.30 }$       &$ 46.34 _{ 0.58 }$&$ 44.69 _{ 0.69 }$&$ 45.72 _{ 0.65 }$\\
PenLog       &$ 48.81 _{ 0.29 }$&$ 46.36 _{ 0.28 }$&$ 42.15 _{ 0.24 }$& N/A & N/A & N/A\\
SDR5-LDA      &N/A              &$ 46.27 _{ 0.24 }$&$ 42.09 _{ 0.25 }$ &N/A               &$ 37.96 _{ 0.29 }$&$ 31.04 _{ 0.27 }$\\
SDR5-$k$nn    &N/A               &$ 46.14 _{ 0.27 }$&$ 36.28 _{ 0.24 }$  &N/A           &$ 39.70 _{ 0.32 }$&$ 29.31 _{ 0.26 }$\\
OTE           &$ 46.74 _{ 0.28 }$&$ 30.62 _{ 0.33 }$&$ 11.43 _{ 0.19 }$    &$ 32.24 _{ 0.51 }$&$ 23.37 _{ 0.28 }$&$ 19.59 _{ 0.19 }$\\
ES$k$nn       &$ 48.66 _{ 0.26 }$&$ 46.59 _{ 0.26 }$&$ 45.17 _{ 0.22 }$    &$ 46.15 _{ 0.51 }$&$ 44.03 _{ 0.54 }$&$ 43.77 _{ 0.46 }$\\
 \end{tabular}}
\end{table}

Model~4 explores the effect of heavy tails and the presence of correlation between the features. Again, assumption~3 is not satisfied for any $d < p$.  The RF, OTE and RP-$k$nn methods all perform very well here.  The RP-LDA and RP-QDA classifiers are less good.  This is partly due the fact that the class-conditional distributions do not have finite second and first moments, respectively, and, as a result, the class mean and covariance matrix estimates are poor.   

\subsection{Real data examples}
\label{sec--real}
In this section, we compare the classifiers above on eight real datasets available from the UC Irvine (UCI) Machine Learning Repository.  In each example, we first subsample the data to form a training set of size $n$, then use the remaining data (or, where available, take a subsample of size 1000 from it) to form the test set.  As with the simulated examples, we set $B_1 = 500$, $B_2 = 50$, used Gaussian distributed projections, and each experiment was repeated 100 times.  Where appropriate, the tuning parameters were chosen via the methods described at the beginning of Section~\ref{sec--empirical} for each of the 100 repeats of the experiment. 

\subsubsection{Eye state detection}
The electroencephalogram eye state dataset (\url{http://archive.ics.uci.edu/ml/datasets/EEG+Eye+State}) consists of $p=14$ EEG measurements on 14980 observations.  The task is to use the EEG reading to determine the state of the eye.  There are 8256 observations for which the eye is open (class~0), and 6723 for which the eye is closed (class 1).  

\subsubsection{Ionosphere dataset}
The Ionosphere dataset (\url{http://archive.ics.uci.edu/ml/datasets/Ionosphere}) consists of $p=32$ high-frequency antenna measurements on 351 observations.  Observations are classified as good (class 0) or bad (class 1), depending on whether there is evidence for free electrons in the Ionosphere or not.  The class sizes are 225 (good) and 126 (bad).   

\begin{table}
\caption{\label{tab--Eye} Misclassification rates for the Eye State and Ionosphere datasets}
\centering
\fbox{%
\begin{tabular}{l|c c c | c c c  }
& \multicolumn{3}{c|}{Eye State} &  \multicolumn{3}{c}{Ionosphere} \\
  \multicolumn{1}{r|}{$n$}  &$50$            &$200$       &$1000$  &$50$         &$100$       &$200$   \\
  \hline
RP-LDA$_{5}$  &$ 42.06 _{ 0.38 }$&$ 38.61 _{ 0.29 }$&$ 36.30 _{ 0.21 }$&$ 13.05 _{ 0.38 }$&$ 10.75 _{ 0.25 }$&$  9.78 _{ 0.26 }$\\
RP-QDA$_{5}$  &$\mathbf{ 38.97} _{ 0.39 }$&$ 32.44 _{ 0.42 }$&$ 30.91 _{ 0.87 }$ &$\mathbf{8.14}_{0.37}$&$\mathbf{6.15}_{0.22}$&$\mathbf{5.21}_{0.20}$\\
RP-$k$nn$_{5}$&$\mathbf{39.37}_{ 0.39 }$&$\mathbf{26.91} _{ 0.27 }$&$ \mathbf{13.54} _{ 0.19 }$&$ 13.05 _{ 0.46 }$&$  7.43 _{ 0.25 }$&$  5.43 _{ 0.19 }$\\
LDA          &$ 42.38 _{ 0.40 }$&$ 39.15 _{ 0.30 }$&$ 36.91 _{ 0.23 }$ &$ 23.72 _{ 0.40 }$&$ 18.27 _{ 0.28 }$&$ 15.58 _{ 0.31 }$\\
QDA          &$ 39.91 _{ 0.35 }$&$ 29.24 _{ 0.40 }$&$ 29.76 _{ 1.07 }$     &N/A     & N/A                    &$ 14.07 _{ 0.34 }$\\
$k$nn        &$ 41.70 _{ 0.40 }$&$ 29.18 _{ 0.27 }$&$ 14.45 _{ 0.16 }$   &$ 21.81 _{ 0.73 }$&$ 18.05 _{ 0.46 }$&$ 16.40 _{ 0.35 }$\\
RF           &$ \mathbf{39.27} _{ 0.37 }$&$ 29.04 _{ 0.25 }$&$ 17.63 _{ 0.20 }$    &$ 10.52 _{ 0.30 }$&$  7.54 _{ 0.19 }$&$  6.48 _{ 0.18 }$\\
Radial SVM   &$ 46.33 _{ 0.49 }$&$ 38.71 _{ 0.46 }$&$ 31.03 _{ 0.68 }$&$ 27.67 _{ 1.15 }$&$ 12.85 _{ 0.91 }$&$  6.67 _{ 0.22 }$\\
Linear SVM   &$ 42.38 _{ 0.42 }$&$ 39.55 _{ 0.36 }$&$ 38.58 _{ 0.38 }$ &$ 19.41 _{ 0.35 }$&$ 17.05 _{ 0.27 }$&$ 15.48 _{ 0.29 }$\\
Radial GP    &$ 40.73 _{ 0.38 }$&$ 32.22 _{ 0.25 }$&$ 21.66 _{ 0.21 }$    &$ 22.29 _{ 0.72 }$&$ 17.81 _{ 0.46 }$&$ 14.52 _{ 0.31 }$\\
PenLDA       &$ 44.37 _{ 0.43 }$&$ 42.50 _{ 0.28 }$&$ 41.86 _{ 0.23 }$   &$ 21.20 _{ 0.57 }$&$ 19.83 _{ 0.56 }$&$ 19.81 _{ 0.54 }$\\
NSC          &$ 44.73 _{ 0.48 }$&$ 42.37 _{ 0.29 }$&$ 42.27 _{ 0.28 }$  &$ 22.62 _{ 0.53 }$&$ 19.11 _{ 0.42 }$&$ 17.52 _{ 0.34 }$\\
SDR5-LDA     &$ 42.82 _{ 0.40 }$&$ 39.25 _{ 0.29 }$&$ 36.92 _{ 0.23 }$  &$ 25.78 _{ 0.52 }$&$ 18.98 _{ 0.30 }$&$ 15.63 _{ 0.30 }$\\
SDR5-$k$nn   &$ 42.43 _{ 0.38 }$&$ 34.13 _{ 0.32 }$&$ 25.31 _{ 0.25 }$ &$ 30.61 _{ 0.74 }$&$ 17.53 _{ 0.45 }$&$ 10.12 _{ 0.30 }$\\
OTE          &$ 40.10 _{ 0.38 }$&$ 29.92 _{ 0.28 }$&$ 18.73 _{ 0.20 }$  &$ 14.38 _{ 0.41 }$&$  9.80 _{ 0.27 }$&$  7.33 _{ 0.23 }$\\
ES$k$nn      &$ 45.62 _{ 0.41 }$&$ 43.06 _{ 0.35 }$&$ 39.37 _{ 0.34 }$ &$ 27.81 _{ 0.58 }$&$ 23.23 _{ 0.48 }$&$ 20.05 _{ 0.51 }$\\
\end{tabular}}
\end{table}

\subsubsection{Down's syndrome diagnoses in mice}
The Mice dataset (\url{http://archive.ics.uci.edu/ml/datasets/Mice+Protein+Expression}) consists of 570 healthy mice (class 0) and 507 mice with Down's syndrome (class 1).  The task is to diagnose Down's syndrome based on $p = 77$ protein expression measurements.   

\subsubsection{Hill-Valley identification}
The Hill-Valley dataset (\url{http://archive.ics.uci.edu/ml/datasets/Hill-Valley}) consists of 1212 observations of a terrain, each one when plotted in sequence represents either a Hill  (class 0, size 600) or a Valley (class 1, size 612).  The task is to classify the terrain based on a vector of dimension $p = 100$.   

\begin{table}
\caption{\label{tab--Mice} Misclassification rates for the Mice and Hill-Valley datasets.}
\centering
\fbox{%
\begin{tabular}{l| c c c  | c c c }
& \multicolumn{3}{c|}{Mice} &  \multicolumn{3}{c}{Hill-Valley} \\
  \multicolumn{1}{r|}{$n$} &$200$        &$500$   &$1000$    & $100$       &$200$        &$500$       \\
  \hline
RP-LDA$_{5}$  &$ 25.17 _{ 0.30 }$&$23.56 _{ 0.26 }$&$23.35 _{ 0.49 }$&$\mathbf{36.84}_{ 0.84 }$&$\mathbf{36.45} _{ 0.85 }$&$\mathbf{32.57}_{ 1.06 }$\\
RP-QDA$_{5}$  &$ 18.24 _{ 0.29 }$&$16.05 _{ 0.24 }$&$15.45 _{ 0.45 }$&$ 44.43 _{ 0.34 }$&$ 43.56 _{ 0.31 }$&$ 41.10 _{ 0.33 }$\\
RP-$k$nn$_{5}$&$ 11.24 _{ 0.29 }$&$\mathbf{2.24}_{ 0.10 }$&$\mathbf{0.55}_{ 0.09 }$&$ 49.08 _{ 0.24 }$&$ 47.27 _{ 0.26 }$&$ 36.39 _{ 0.29 }$\\
LDA                   &$  \mathbf{6.46} _{ 0.14 }$&$ 3.38 _{ 0.10 }$&$ 2.17 _{ 0.17 }$&N/A                           &$ 37.29 _{ 0.48 }$&$ 34.37 _{ 0.36 }$\\
$k$nn       &$ 19.65 _{ 0.26 }$&$ 7.02 _{ 0.17 }$&$ 0.94 _{ 0.13 }$&$ 49.35 _{ 0.24 }$&$ 48.82 _{ 0.21 }$&$ 47.49 _{ 0.24 }$\\
RF           &$  7.94 _{ 0.22 }$&$ 2.41 _{ 0.11 }$&$\mathbf{0.51}_{ 0.08 }$&$ 48.32 _{ 0.23 }$&$ 47.23 _{ 0.21 }$&$ 44.11 _{ 0.25 }$\\
Radial SVM  &$ 11.25 _{ 0.29 }$&$ 3.89 _{ 0.13 }$&$ 1.69 _{ 0.16 }$&$ 50.24 _{ 0.19 }$&$ 50.24 _{ 0.19 }$&$ 50.42 _{ 0.21 }$\\
Linear SVM   &$  \mathbf{6.36} _{ 0.14 }$&$ 3.64 _{ 0.10 }$&$ 2.51 _{ 0.17 }$&$ 48.56 _{ 0.22 }$&$ 47.03 _{ 0.23 }$&$ 44.84 _{ 0.28 }$\\
Radial GP    &$ 21.22 _{ 0.30 }$&$13.78 _{ 0.24 }$&$ 8.66 _{ 0.34 }$&$ 48.33 _{ 0.22 }$&$ 47.24 _{ 0.21 }$&$ 45.11 _{ 0.22 }$\\
PenLDA       &$ 26.10 _{ 0.36 }$&$24.07 _{ 0.26 }$&$23.91 _{ 0.46 }$&$ 49.59 _{ 0.22 }$&$ 49.73 _{ 0.21 }$&$ 49.55 _{ 0.22 }$\\
NSC         &$ 30.30 _{ 0.36 }$&$28.06 _{ 0.29 }$&$28.47 _{ 0.51 }$&$ 49.87 _{ 0.21 }$&$ 49.91 _{ 0.20 }$&$ 49.92 _{ 0.22 }$\\
OTE       &$ 11.83 _{ 0.32 }$&$ 6.26 _{ 0.18 }$&$ 3.26 _{ 0.23 }$&$ 48.33 _{ 0.23 }$&$ 47.18 _{ 0.22 }$&$ 44.20 _{ 0.24 }$\\
ES$k$nn &$ 39.03 _{ 0.59 }$&$ 34.33 _{ 0.66 }$&$ 31.65 _{ 0.78 }$&$ 49.31 _{ 0.23 }$&$ 48.90 _{ 0.23 }$&$ 48.03 _{ 0.25 }$\\
\end{tabular}}
\end{table}

\subsubsection{Musk identification}
The Musk dataset (\url{http://archive.ics.uci.edu/ml/datasets/Musk+\%28Version+2\%29}) consists of 1016 musk (class 0) and 5581 non-musk (class 1) molecules.  The task is to classify a molecule based on $p = 166$ shape measurements.   

\subsubsection{Cardiac Arrhythmia diagnoses}
The cardiac arrhythmia dataset (\url{https://archive.ics.uci.edu/ml/datasets/Arrhythmia}) has one normal class of size 245, and 13 abnormal classes, which we combined to form a second class of size 206.  We removed the nominal features and those with missing values, leaving $p = 194$ electrocardiogram (ECG) measurements.   In this example, the PenLDA classifier is N/A due to the fact that some features have within-class standard deviation equal to zero.

\begin{table}
\caption{\label{tab--Musk} Misclassification rates for the Musk and Cardiac Arrhythmia datasets.}
\centering
\fbox{%
\begin{tabular}{l|c c c | c c c }
  & \multicolumn{3}{c|}{Musk} &  \multicolumn{3}{c}{Arrhythmia} \\
   \multicolumn{1}{r|}{$n$}  &$100$       &$200$        &$500$ &  $50$         &$100$       &$200$        \\
  \hline
RP-LDA$_{5}$   &$ 14.63 _{ 0.31 }$&$ 12.18 _{ 0.23 }$&$ 10.15 _{ 0.15 }$&$ 33.24 _{ 0.42 }$&$ 30.19 _{ 0.35 }$&$ 27.49 _{ 0.30 }$\\
RP-QDA$_{5}$   &$ 12.08 _{ 0.27 }$&$ 9.92 _{ 0.18 }$&$ 8.64 _{ 0.13 }$ &$ \mathbf{30.47} _{ 0.33 }$&$ 28.28 _{ 0.26 }$&$ 26.31 _{ 0.28 }$\\
RP-$k$nn$_{5}$ &$\mathbf{11.81}_{ 0.27 }$&$\mathbf{9.65}_{ 0.21 }$&$ 8.04 _{ 0.15 }$&$ 33.49 _{ 0.40 }$&$ 30.18 _{ 0.33 }$&$ 27.09 _{ 0.31 }$\\
LDA            &N/A            &$ 24.88 _{ 0.42 }$&$ 9.09 _{ 0.15 }$&N/A & N/A & N/A\\
$k$nn         &$ 14.68 _{ 0.28 }$&$ 11.75 _{ 0.22 }$&$  8.20 _{ 0.15 }$ &$ 40.64 _{ 0.33 }$&$ 38.94 _{ 0.33 }$&$ 35.76 _{ 0.36 }$\\
RF            &$ 13.20 _{ 0.20 }$&$ 10.69 _{ 0.18 }$&$  7.55 _{ 0.13 }$  &$ 31.65 _{ 0.39 }$&$ \mathbf{26.72} _{ 0.29 }$&$\mathbf{22.40} _{ 0.31 }$\\
Radial SVM    &$ 15.25 _{ 0.15 }$&$ 15.21 _{ 0.15 }$&$ 15.00 _{ 0.17 }$   &$ 48.39 _{ 0.49 }$&$ 47.24 _{ 0.46 }$&$ 46.85 _{ 0.43 }$\\
Linear SVM    &$ 13.91 _{ 0.25 }$&$ 10.39 _{ 0.18 }$&$ \mathbf{7.41}_{ 0.12 }$ &$ 36.16 _{ 0.47 }$&$ 35.61 _{ 0.39 }$&$ 35.20 _{ 0.35 }$\\
Radial GP     &$ 14.91 _{ 0.16 }$&$ 14.07 _{ 0.20 }$&$ 11.14 _{ 0.19 }$  &$ 37.28 _{ 0.42 }$&$ 33.80 _{ 0.40 }$&$ 29.31 _{ 0.35 }$\\
PenLDA        &$ 27.74 _{ 0.58 }$&$ 27.14 _{ 0.54 }$&$ 26.98 _{ 0.31 }$ &N/A & N/A & N/A \\
NSC           &$ 15.32 _{ 0.18 }$&$ 15.22 _{ 0.15 }$&$ 15.20 _{ 0.16 }$  &$ 34.98 _{ 0.46 }$&$ 33.00 _{ 0.40 }$&$ 31.08 _{ 0.41 }$\\
PenLog           &$ 14.48 _{ 0.28 }$&$ 11.85 _{ 0.21 }$&N/A &  $34.92 _{0.42}$&$30.48_{0.34}$ &$26.12_{ 0.27 }$\\
SDR5-LDA      &N/A                  &$ 25.12 _{ 0.43 }$&$ 9.08 _{ 0.15 }$&N/A & N/A & N/A \\
SDR5-$k$nn    &N/A             &$ 24.09 _{ 0.62 }$&$ 9.81 _{ 0.16 }$&N/A & N/A & N/A \\
OTE           &$ 13.90 _{ 0.23 }$&$ 11.04 _{ 0.18 }$&$ 8.05 _{ 0.14 }$&$ 33.90 _{ 0.47 }$&$ 27.83 _{ 0.29 }$&$ 23.75 _{ 0.32 }$\\
ES$k$nn       &$ 19.55 _{ 0.42 }$&$ 18.09 _{ 0.30 }$&$ 16.07 _{ 0.24 }$    &$ 45.86 _{ 0.43 }$&$ 45.62 _{ 0.48 }$&$ 43.41 _{ 0.43 }$\\
\end{tabular}}
\end{table}

\subsubsection{Human Activity Recognition}
\label{sec--HAR}
This dataset (\url{http://archive.ics.uci.edu/ml/datasets/Human+Activity+Recognition+Using+Smartphones}) consists of $p=561$ accelerometer measurements, recorded from a smartphone whilst a subject is performing an activity.   We subsampled the data to include only the walking and laying activities.  In the resulting dataset, there are 1226 `walking' observations (class 0),  and 1407 `laying' observations (class 1).  

\subsubsection{Handwritten digits}
The Gisette dataset (\url{https://archive.ics.uci.edu/ml/datasets/Gisette}) consists of 6000 observations of handwritten digits, namely 3000 ``4''s and 3000 ``9''s.  Each observation represents the original $28 \times 28$ pixel image, with added noise variables resulting in a 5000-dimensional vector.   We first subsampled 1500 of the 6000 observations, giving 760 ``4''s and 740 ``9''s -- this dataset was then kept fixed through the subsequent 100 repeats of the experiment.   The observations are sparse with a large number of 0 entries.  

\begin{table}
\caption{\label{tab--HAR}Misclassification rates for the Activity recognition and Gisette datasets.}
\centering
\fbox{%
\begin{tabular}{l|c c c | c c c }
  & \multicolumn{3}{c|}{Activity Recognition} &  \multicolumn{3}{c}{Gisette} \\
  \multicolumn{1}{r|}{$n$}  &$50$       &$200$  &$1000$   &$50$       &$200$  &$1000$      \\
  \hline
RP-LDA$_{5}$  &$ 0.18 _{ 0.02 }$&$ 0.10 _{ 0.01 }$&$ 0.01 _{ 0.00 }$ &$ 15.75 _{ 0.41 }$&$ 10.58 _{ 0.17 }$&$ 9.39 _{ 0.15 }$\\
RP-QDA$_{5}$  &$ 0.15 _{ 0.02 }$&$ 0.09 _{ 0.01 }$&$ \mathbf{0.00} _{ 0.00 }$&$ 15.53 _{ 0.40 }$&$ 10.53 _{ 0.19 }$&$ 9.37 _{ 0.16 }$\\
RP-$k$nn$_{5}$&$ 0.21 _{ 0.02 }$&$ 0.11 _{ 0.01 }$&$ 0.01 _{ 0.00 }$&$ 15.95 _{ 0.46 }$&$ 11.09 _{ 0.17 }$&$ 9.57 _{ 0.16 }$\\
$k$nn                 &$ 0.26 _{ 0.02 }$&$ 0.13 _{ 0.02 }$&$ 0.02 _{ 0.01 }$&$ 18.41 _{ 0.42 }$&$ 10.44 _{ 0.18 }$&$  5.64 _{ 0.13 }$\\
RF                      &$ 0.25 _{ 0.02 }$&$ 0.17 _{ 0.02 }$&$ 0.08 _{ 0.01 }$  &$ 14.33 _{ 0.47 }$&$  9.37 _{ 0.15 }$&$  5.79 _{ 0.12 }$\\
Radial SVM   &$ 1.58 _{ 0.11 }$&$ 0.89 _{ 0.06 }$&$ 0.18 _{ 0.02 }$   &$ 50.03 _{ 0.19 }$&$ 50.41 _{ 0.19 }$&$ 50.79 _{ 0.25 }$\\
Linear SVM   &$ 0.19 _{ 0.02 }$&$ 0.12 _{ 0.01 }$&$ 0.05 _{ 0.01 }$&$ \mathbf{11.92} _{ 0.27 }$&$  \mathbf{6.82} _{ 0.11 }$&$  \mathbf{4.45} _{ 0.11 }$\\
Radial GP    &$ 0.25 _{ 0.02 }$&$ 0.20 _{ 0.02 }$&$ 0.13 _{ 0.01 }$ &$ 27.09 _{ 1.32 }$&$ 10.74 _{ 0.21 }$&$  6.70 _{ 0.13 }$\\
PenLDA       &$ \mathbf{0.11} _{ 0.02 }$&$ \mathbf{0.04} _{ 0.01 }$&$ \mathbf{0.00} _{ 0.00 }$ &N/A & N/A &N/A\\
NSC          &$ 0.29 _{ 0.02 }$&$ 0.24 _{ 0.03 }$&$ 0.06 _{ 0.01 }$  &$ 15.72 _{ 0.29 }$&$ 13.63 _{ 0.22 }$&$ 12.83 _{ 0.21 }$ \\
OTE          &$ 0.61 _{ 0.07 }$&$ 0.38 _{ 0.05 }$&$ 0.09 _{ 0.02 }$&$ 14.18 _{ 0.25 }$&$  9.69 _{ 0.17 }$&$  6.24 _{ 0.13 }$\\
ES$k$nn      &$ 1.74 _{ 0.18 }$&$ 0.88 _{ 0.09 }$&$ 0.41 _{ 0.05 }$   &$ 45.76 _{ 0.76 }$&$ 44.81 _{ 0.74 }$&$ 44.45 _{ 0.73 }$\\
\end{tabular}}
\end{table}

\subsection{Conclusion of numerical study} 
\label{Sec:Conc}
The numerical study above reveals the extremely encouraging finite-sample performance achieved by the random projection ensemble classifier.  An RP ensemble method attains the lowest misclassification error in 23 of the 36 simulated and real data settings investigated, and in 8 of the 13 remaining cases an RP ensemble method is in the top three of the classifiers considered.  The flexibility offered by the random projection ensemble classifier -- in particular, the fact that any base classifier may be used -- allows the practitioner to adapt the method to work well in a wide variety of problems.  

Another key observation is that our assumption~3 is not necessary for the RP method to work well: in Model~2, we achieve good results using $d=2$, while assumption~3 holds only with a 3 (or higher)-dimensional projection.  Moreover, even in situations where assumption~3 does not hold for any $d < p$, the RP method is still competitive; see in particular the results for Model~3.

One example where the RP ensemble framework is not effective is for the Gisette dataset.  Here the data are very sparse; for each observation a large proportion of the features are exactly zero.  Of course, applying a Gaussian or Haar random projection to an observation will remove the sparse structure.  In this case, the practitioner may benefit by using an alternative distribution for the projections, such as axis-aligned projections (cf.\ the discussion in Section~\ref{sec--discussion}).

\section{Discussion and extensions}
\label{sec--discussion}
We have introduced a general framework for high-dimensional classification via the combination of the results of applying a base classifier on carefully selected low-dimensional random projections of the data.  One of its attractive features is its generality: the approach can be used in conjunction with any base classifier.  Moreover, although we explored in detail one method for combining the random projections (partly because it facilitates rigorous statistical analysis), there are many other options available here.  For instance, instead of only retaining the projection within each block yielding the smallest estimate of test error, one might give weights to the different projections, where the weights decrease as the estimate of test error increases.  

Many practical classification problems involve $K > 2$ classes.  The main issue in extending our methodology to such settings is the definition of ${C}_n^{\mathrm{RP}}$ analogous to~\eqref{eq--RP}.  To outline one approach, let 
\[
\nu_{n,r}(x) := \frac{1}{B_1}\sum_{b_1=1}^{B_1} \mathds{1}_{\{{C}_n^{\mathbf{A}_{b_1}}(x) = r\}}
\]
for $r=0,1,\ldots,K-1$.  Given $\alpha_0,\ldots,\alpha_{K-1} > 0$ with $\sum_{r=0}^{K-1} \alpha_r = 1$, we can then define
\[
{C}_n^{\mathrm{RP}}(x) := \sargmax_{r=0,\ldots,K-1} \{\alpha_r \nu_{n,r}(x)\},
\]
where $\sargmax$ denotes the smallest element of the $\argmax$ in the case of a tie.  The choice of $\alpha_0,\ldots,\alpha_{K-1}$ is analogous to the choice of $\alpha$ in the case $K=2$.  It is therefore natural to seek to minimise the test error of the corresponding infinite-simulation random projection classifier as before.

In other situations, it may be advantageous to consider alternative types of projection, perhaps because of additional structure in the problem.  One particularly interesting issue concerns ultrahigh-dimensional settings, say $p$ in the thousands.  Here, it may be too time-consuming to generate enough random projections to explore adequately the space $\mathcal{A}_{d \times p}$.  As a mathematical quantification of this, the cardinality of an $\epsilon$-net in the Euclidean norm of the surface of the Euclidean ball in $\mathbb{R}^p$ increases exponentially in $p$ \citep[e.g.][]{Vershynin2012}.  In such challenging problems, one might restrict the projections $\mathbf{A}$ to be axis-aligned, so that each row of $\mathbf{A}$ consists of a single non-zero component, equal to 1, and $p-1$ zero components.  There are then only $\binom{p}{d} \leq p^d/d!$ choices for the projections, and if $d$ is small, it may be feasible even to carry out an exhaustive search.  Of course, this approach loses one of the attractive features of our original proposal, namely the fact that it is equivariant to orthogonal transformations.  Nevertheless, corresponding theory can be obtained provided that the projection $A^*$ in assumption~3 is axis-aligned.  This is a much stronger requirement, but it seems that imposing greater structure is inevitable to obtain good classification in such settings. 

Our main focus in this work has been on the classification performance of the random projection ensemble classifier, and not on the interpretability of the class assignments.  However, the selected projections provide weights that give an indication of the relative importance of the different variables in the model.  Another interesting direction, therefore, would be to understand the properties of the variable ranking induced by the random projection ensemble classifier. 

In conclusion, we believe that random projections offer many exciting possibilities for high-dimensional data analysis.  In a similar spirit to subsampling and bootstrap sampling, we can think of each random projection as a perturbation of our original data, and effects that are observed over many different perturbations are often the `stable' effects that are sought by statisticians; cf.~\citet{MeinshausenBuhlmann2010,ShahSamworth2013} in the context of variable selection.  Two of the key features that make them so attractive for classification problems are the ability to identify `good' random projections from the data, and the fact that we can aggregate results from selected projections.  We anticipate that these two properties will be important in identifying future application areas for related methodologies. 
 
\section{Appendix}
\label{sec--tech}

\begin{prooftitle}{of Theorem \ref{thm--Bvar}.}
Recall that the training data $\mathcal{T}_n = \{(x_1,y_1),\ldots, (x_n,y_n)\}$ are fixed and the projections $\mathbf{A}_{1},\mathbf{A}_2, \dots $, are independent and identically distributed in $\mathcal{A}$, independent of the pair $(X,Y)$. The test error of the random projection ensemble classifier has the following representation:
\begin{align*}
\mathbf{E}\{R({C}_n^{\mathrm{RP}})\} & =\mathbf{E}\Bigl\{\pi_0 \int_{\mathbb{R}^p} \mathds{1}_{\{C_n^{\mathrm{RP}}(x) = 1\}} \, dP_0(x) + \pi_1 \int_{\mathbb{R}^p} \mathds{1}_{\{C_n^{\mathrm{RP}}(x) = 0\}} \, dP_1(x)\Bigr\}
\\ & =\mathbf{E}\Bigl\{\pi_0 \int_{\mathbb{R}^p} \mathds{1}_{\{\nu_n(x) \geq \alpha\}} \, dP_0(x) + \pi_1 \int_{\mathbb{R}^p} \mathds{1}_{\{\nu_n(x) < \alpha\}} \, dP_1(x)\Bigr\}
\\ &  = \pi_0 \int_{\mathbb{R}^p} \mathbf{P}\{\nu_n(x) \geq \alpha\} \, dP_0(x) + \pi_1 \int_{\mathbb{R}^p} \mathbf{P}\{\nu_n(x) < \alpha\} \, dP_1(x),
\end{align*}
where $\nu_n(x)$ is defined in~\eqref{eq--nu}, and where the final equality follows by Fubini's theorem.  

Let $U_{b_1} := \mathds{1}_{\{{C}_n^{\mathbf{A}_{b_1}}(X) = 1\}}$, for $b_1 = 1, \dots B_1$. Then, conditional on $\mu_n(X) = \theta \in [0,1]$, the random variables $U_1, \dots, U_{B_1}$ are independent, each having a Bernoulli($\theta$) distribution.  Recall that $G_{n,0}$ and $G_{n,1}$ are the distribution functions of $\mu_n(X) | \{Y=0\}$ and  $\mu_n(X) | \{Y=1\}$, respectively.  We can therefore write
\begin{align*}
\int_{\mathbb{R}^p} \mathbf{P}\{\nu_n(x) < \alpha\} \, dP_1(x) &= \int_{[0,1]}  \mathbb{P}\Bigl\{\frac 1 {B_1} \sum_{b_1 =1}^{B_1} U_{b_1} < \alpha \Big| \hat{\mu}_n(X) = \theta \Bigr\} \, dG_{n,1}(\theta)
\\ & =  \int_{[0,1]}  \mathbb{P} (T < B_1\alpha) \, dG_{n,1}(\theta),
\end{align*}
where here and throughout the proof, $T$ denotes a $\mathrm{Bin}(B_1,\theta)$ random variable.  Similarly,   
\[
\int_{\mathbb{R}^p} \mathbf{P}\{\nu_n(x) \geq \alpha\} \, dP_0(x) =  1 -  \int_{[0,1]}  \mathbb{P}(T < B_1\alpha) \, dG_{n,0}(\theta).
\] 
It follows that
\[
\mathbf{E}\{R({C}_n^{\mathrm{RP}})\} = \pi_0 + \int_{[0,1]}   \mathbb{P}(T < B_1\alpha) \, dG^{\circ}_{n}(\theta),
\]
where $G_n^\circ := \pi_1G_{n,1} - \pi_0 G_{n,0}$.  Writing $g_n^\circ := \pi_1 g_{n,1}-\pi_0g_{n,0}$, we now show that 
\begin{equation}
\label{eq--rtp}
\int_{[0,1]} \bigl\{ \mathbb{P}(T < B_1\alpha) - \mathds{1}_{\{\theta < \alpha\}}\bigr\}\, d G_n^\circ(\theta) = \frac{1 - \alpha - \llbracket B_1\alpha \rrbracket}{B_1}g_n^\circ(\alpha) +  \frac{\alpha(1-\alpha)}{2B_1}\dot{g}_{n}^{\circ}(\alpha) + o\Bigl(\frac{1}{B_1}\Bigr)
\end{equation}
as $B_1 \to \infty$.  Our proof involves a one-term Edgeworth expansion to the binomial distribution function in~\eqref{eq--rtp}, where the error term is controlled uniformly in the parameter.  The expansion relies on the following version of Esseen's smoothing lemma.
\begin{theorem}\citep[Chapter 2, Theorem~2b]{Esseen:45}
Let $c_1$, $C_1$, $S > 0$, let $F:\mathbb{R} \to [0,\infty)$ be a non-decreasing function and let $G: \mathbb{R} \to \mathbb{R}$ be a function of bounded variation. Let $F^*(s) :=  \int_{-\infty}^{\infty} \exp(i s t) \, dF(t)$ and $G^*(s) := \int_{-\infty}^{\infty} \exp(is t) \, dG(t)$ be the Fourier--Stieltjes transforms of $F$ and $G$, respectively. Suppose that
\begin{itemize}
\item $\lim_{t \to -\infty} F(t) = \lim_{t \to -\infty} G(t) = 0$ and $\lim_{t \to \infty} F(t) = \lim_{t \to \infty} G(t)$;
\item $\int_{-\infty}^{\infty}|F(t)-G(t)|\, dt <\infty$;
\item The set of discontinuities of $F$ and $G$ is contained in $\{t_i:i \in \mathbb{Z}\}$, where $(t_i)$ is a strictly increasing sequence with $\inf_i\{t_{i+1}-t_i\} \geq c_1$; moreover $F$ is constant on the intervals $[t_i, t_{i+1})$ for all $i \in \mathbb{Z}$;
\item  $|\dot{G}(t)|\leq C_1$ for all $t \notin \{t_i:i \in \mathbb{Z}\}$.
\end{itemize}
Then there exist constants $c_2, C_2 > 0$ such that
\[
 \sup_{t\in \mathbb{R}} |F(t) - G(t) |\leq  \frac{1}{\pi} \int_{-S}^{S} \biggl| \frac{F^*(s) - G^*(s)}{s}\biggr| \, ds  +  \frac {C_1 C_2} S,
\]
provided that $S c_1 \geq c_2$.
\label{thm--esseen}
\end{theorem}
Let $\sigma^2 := \theta(1-\theta)$, and let $\Phi$ and $\phi$ denote the standard normal distribution and density functions, respectively.  Moreover, for $t \in \mathbb{R}$, let 
$$p(t) = p(t, \theta) := \frac{(1-t^2)(1-2\theta)}{6\sigma},$$
and
$$q(t) = q(t,B_1,\theta) :=  \frac{1/2 -  \llbracket B_1\theta + B_1^{1/2}\sigma t \rrbracket}{\sigma}.$$ 
In Proposition \ref{cor--edgeworth} below we apply Theorem \ref{thm--esseen} to the following functions: 
\begin{equation}
\label{eq--FB1}
F_{B_1}(t) = F_{B_1}(t,\theta) :=  \mathbb{P}\biggl(\frac{T - B_1\theta}{B_1^{1/2}\sigma} < t \biggr), 
\end{equation}
and 
\begin{equation}
\label{eq--GB1}
G_{B_1}(t) = G_{B_1}(t,\theta) := \Phi(t) + \phi(t)\frac{p(t,\theta) + q(t,B_1, \theta)}{B_1^{1/2}}.
\end{equation}

\begin{proposition}
\label{cor--edgeworth}
Let $F_{B_1}$ and $G_{B_1}$ be as in (\ref{eq--FB1}) and (\ref{eq--GB1}). There exists a constant $C>0$ such that, for all $B_1 \in \mathbb{N}$,
\[
\sup_{\theta \in (0,1)} \sup_{t \in \mathbb{R}} \sigma^3 |F_{B_1}(t,\theta) - G_{B_1}(t,\theta)| \leq \frac{C}{B_1}.
\]
\end{proposition}
Proposition~\ref{cor--edgeworth}, whose proof is given after the proof of Proposition~\ref{Prop:A^*}, bounds uniformly in $\theta$ the error in the one-term Edgeworth expansion $G_{B_1}$ of the distribution function $F_{B_1}$.  Returning to the proof of Theorem \ref{thm--Bvar}, we will argue that the dominant contribution to the integral in (\ref{eq--rtp}) arises from the interval $(\max\{0,\alpha - \epsilon_1\}$,$\min\{\alpha + \epsilon_1,1\})$, where $\epsilon_{1} := B_1^{-1/2}\log B_1$. For the remainder of the proof we assume $B_1$ is large enough that $[\alpha - \epsilon_1,\alpha + \epsilon_1] \subseteq (0,1)$.

For the region $|\theta - \alpha| \geq \epsilon_{1}$, by Hoeffding's inequality, we have that
\begin{align*}
\sup_{|\theta - \alpha| \geq \epsilon_1} \bigl| \mathbb{P}(T < B_1\alpha) - \mathds{1}_{\{\theta < \alpha\}}\bigr|  \leq \sup_{|\theta - \alpha| \geq \epsilon_1} \exp\bigl(-2B_1(\theta - \alpha)^2\bigr) \leq e^{-2\log^{2}{B_1}} = O(B_1^{-M}),
\end{align*}
for each $M>0$, as $B_1 \to \infty$. Writing $I := [\alpha-\epsilon_1,\alpha+\epsilon_1]$, it follows that
\begin{equation}
\label{Eq:Hoeffding}
\int_{[0,1]} \bigl\{ \mathbb{P}(T < B_1\alpha) - \mathds{1}_{\{\theta < \alpha\}}\bigr\} \, dG_{n}^\circ(\theta) = \int_I  \bigl\{\mathbb{P}(T < B_1\alpha) -  \mathds{1}_{\{\theta < \alpha\}}\bigr\} \, dG_{n}^\circ(\theta) + O(B_1^{-M}),
\end{equation}
for each $M>0$, as $B_1 \to \infty$.

For the region $|\theta - \alpha| < \epsilon_{1}$, by Proposition~\ref{cor--edgeworth}, there exists $C'>0$ such that, for all $B_1$ sufficiently large,
\begin{align*} 
\sup_{|\theta-\alpha|< \epsilon_1} \biggl|\mathbb{P}(T < B_1\alpha) - \Phi\biggl( \frac{{B_1}^{1/2}(\alpha-\theta)}{\sigma}\biggr)  - \frac{1}{B_1^{1/2}}\phi\biggl( \frac{{B_1}^{1/2}(\alpha-\theta)}{\sigma}\biggr)  r\biggl( \frac{{B_1}^{1/2}(\alpha-\theta)}{\sigma}\biggr) \biggr| \leq \frac{C'}{B_1},
\end{align*}
where $r(t) := p(t) + q(t)$.  Hence, using the fact that for large $B_1$, $\sup_{|\theta - \alpha| < \epsilon_1} |g_n^\circ(\theta)| \leq |g_n^\circ(\alpha)| + 1 < \infty$ under assumption~1, we have
\begin{align}
\label{eq--main1}
\int_I  \bigl\{\mathbb{P}(&T < B_1\alpha) - \mathds{1}_{\{\theta < \alpha\}}\bigr\} \, dG_{n}^\circ(\theta) \nonumber \\ 
&= \int_I \biggl\{ \Phi\biggl(\frac{B_1^{1/2}(\alpha - \theta)}{\sigma}\biggr) -  \mathds{1}_{\{\theta < \alpha\}} \biggr\}\, dG_{n}^\circ(\theta) \nonumber \\ 
& \hspace{30 pt} + \frac{1}{B_1^{1/2}} \int_I \phi\biggl(\frac{B_1^{1/2}(\alpha - \theta)}{\sigma}\biggr) r\biggl(\frac{B_1^{1/2}(\alpha - \theta)}{\sigma}\biggr) \,  dG_{n}^\circ(\theta) + o\Bigl(\frac{1}{B_1}\Bigr),
\end{align}
as $B_1 \rightarrow \infty$.  To aid exposition, we will henceforth concentrate on the dominant terms in our expansions, denoting the remainder terms as $R_1, R_2,\ldots$.  These remainders are then controlled at the end of the argument.  For the first term in (\ref{eq--main1}), we write 
\begin{align}
\label{eq--main2}
\int_I \biggl\{ \Phi\biggl(&\frac{B_1^{1/2}(\alpha - \theta)}{\sigma}\biggr) - \mathds{1}_{\{\theta<\alpha\}}\biggr\} \, dG_n^\circ(\theta) \nonumber \\
&= \int_I \biggl\{ \Phi\biggl(\frac{B_1^{1/2}(\alpha - \theta)}{\sqrt{\alpha(1-\alpha)}}\biggr) - \mathds{1}_{\{\theta<\alpha\}}\biggr\} \,  dG_n^\circ(\theta) \nonumber \\ 
& \hspace{50 pt}  +\frac{(1-2\alpha)B_1^{1/2}}{2\{\alpha(1-\alpha)\}^{3/2}} \int_I (\alpha-\theta)^2 \phi\biggl(\frac{B_1^{1/2}(\alpha - \theta)}{\sqrt{\alpha(1-\alpha)}}\biggr) \, dG_n^\circ(\theta) + R_1.
\end{align}
Now, for the first term in (\ref{eq--main2}),  
\begin{align}
\label{eq--main3}
\int_I \biggl\{ &\Phi\biggl(\frac{B_1^{1/2}(\alpha - \theta)}{\sqrt{\alpha(1-\alpha)}}\biggr) - \mathds{1}_{\{\theta<\alpha\}}\biggr\} \, dG_n^{\circ}(\theta) \nonumber \\ 
&= \int_{\alpha-\epsilon_1}^{\alpha + \epsilon_1} \biggl\{ \Phi\biggl(\frac{B_1^{1/2}(\alpha - \theta)}{\sqrt{\alpha(1-\alpha)}}\biggr) - \mathds{1}_{\{\theta<\alpha\}}\biggr\} \bigl\{g_n^\circ(\alpha) + (\theta-\alpha) \dot{g}_{n}^\circ(\alpha)\bigr\}\,d\theta + R_2 \nonumber \\
&= \frac {\sqrt{\alpha(1-\alpha)}} {B_1^{1/2}} \int_{-\infty}^{\infty}  \{\Phi(-u) -  \mathds{1}_{\{u<0\}}\} \biggl\{g_n^\circ(\alpha) +  \frac{\sqrt{\alpha(1-\alpha)}}{B_1^{1/2}} u \dot{g}_{n}^\circ(\alpha) \biggr\} \, du + R_2 + R_3 \nonumber \\ 
&= \frac {\alpha(1-\alpha)} {2B_1} \, \dot{g}_{n}^\circ(\alpha) + R_2 + R_3.
\end{align}
For the second term in (\ref{eq--main2}), write
\begin{align}
\label{eq--dom3}
& \frac{(1-2\alpha)B_1^{1/2}}{2\{\alpha(1-\alpha)\}^{3/2}} \int_I (\alpha-\theta)^2 \phi\biggl(\frac{B_1^{1/2}(\alpha - \theta)}{\sqrt{\alpha(1-\alpha)}}\biggr) \, dG_n^\circ(\theta)  \nonumber
\\ & \hspace{30 pt} =  \frac{(1-2\alpha)B_1^{1/2}}{2\{\alpha(1-\alpha)\}^{3/2}}g_n^{\circ}(\alpha)\int_{\alpha - \epsilon_1}^{\alpha + \epsilon_1}(\alpha-\theta)^2  \phi\biggl(\frac{B_1^{1/2}(\alpha - \theta)}{\sqrt{\alpha(1-\alpha)}}\biggr) \, d\theta + R_4 \nonumber
\\ & \hspace{30 pt} = \frac{1/2 - \alpha}{B_1}  g_n^\circ(\alpha) \int_{-\infty}^{\infty} u^2 \phi(-u)  \, du + R_4 + R_5 =  \frac{1/2-\alpha}{B_1}  g_n^\circ(\alpha) + R_4 + R_5.
\end{align}
Returning to the second term in (\ref{eq--main1}), observe that
\begin{align}
\label{eq--main4}
\frac{1}{B_1^{1/2}}&\int_I \phi\biggl(\frac{B_1^{1/2}(\alpha - \theta)}{\sigma}\biggr) r\biggl(\frac{B_1^{1/2}(\alpha - \theta)}{\sigma}\biggr) \, dG_n^{\circ}(\theta) \nonumber \\ 
&=  \frac{1/2 -  \llbracket B_1\alpha \rrbracket}{B_1^{1/2}}\int_I   \frac{1}{\sigma} \phi\biggl(\frac{B_1^{1/2}(\alpha - \theta)}{\sigma}\biggr) \,dG_n^{\circ}(\theta) \nonumber \\
&\hspace{30 pt} + \frac{1}{6B_1^{1/2}}  \int_I \frac{(1-2\theta)}{\sigma} \biggl\{1 - \frac{B_1(\alpha - \theta)^2}{\sigma^2}\biggr\} \phi\biggl(\frac{B_1^{1/2}(\alpha - \theta)}{\sigma}\biggr) \,dG_n^{\circ}(\theta) \nonumber \\ 
&= \frac{1/2 -  \llbracket B_1\alpha \rrbracket}{B_1^{1/2}}\int_I   \frac{1}{\sigma} \phi\biggl(\frac{B_1^{1/2}(\alpha - \theta)}{\sigma}\biggr) \,dG_n^{\circ}(\theta) + R_6 \nonumber \\
&= \frac{1/2 -  \llbracket B_1\alpha \rrbracket}{B_1^{1/2}\sqrt{\alpha(1-\alpha)}} g_n^{\circ}(\alpha) \int_{\alpha - \epsilon_1}^{\alpha+\epsilon_1} \phi\biggl(\frac{B_1^{1/2}(\alpha - \theta)}{\sqrt{\alpha(1-\alpha)}}\biggr)  \, d\theta + R_6 + R_7 \nonumber \\
&= \frac{1/2 -  \llbracket B_1\alpha \rrbracket }{B_1} g_n^\circ(\alpha) + R_6 + R_7 + R_8.
\end{align}
The claim~\eqref{eq--rtp} will now follow from~(\ref{Eq:Hoeffding}), (\ref{eq--main1}), (\ref{eq--main2}), (\ref{eq--main3}), (\ref{eq--dom3}) and (\ref{eq--main4}), once we have shown that 
\begin{equation}
\label{Eq:AllErrorTerms}
\sum_{j=1}^8 |R_j| = o(B_1^{-1})
\end{equation}
as $B_1 \rightarrow \infty$.

\emph{To bound $R_1$}: For $\zeta \in (0,1)$, let $h_{\theta}(\zeta) := \Phi\bigl(\frac{B_1^{1/2}(\alpha - \theta)}{\sqrt{\zeta(1-\zeta)}}\bigr)$. Observe that, by a Taylor expansion about $\zeta = \alpha$, there exists $B_0 \in \mathbb{N}$, such that, for all $B_1 > B_0$ and all $\theta, \zeta \in (\alpha - \epsilon_1, \alpha + \epsilon_1)$,
\begin{align*}
& \biggl| \Phi\biggl(\frac{B_1^{1/2}(\alpha - \theta)}{\sqrt{\zeta(1-\zeta)}}\biggr) - \Phi\biggl(\frac{B_1^{1/2}(\alpha - \theta)}{\sqrt{\alpha(1-\alpha)}}\biggr)  + (\zeta-\alpha) \frac{(1-2\alpha)B_1^{1/2}(\alpha-\theta)}{2\{\alpha(1-\alpha)\}^{3/2}} \phi\biggl(\frac{B_1^{1/2}(\alpha - \theta)}{\sqrt{\alpha(1-\alpha)}}\biggr)\biggr| 
\\ & \hspace{60 pt} = |h_{\theta}(\zeta) - h_{\theta}(\alpha) - (\zeta-\alpha)\dot{h}_{\theta}(\alpha)|
\\ & \hspace{60 pt} \leq \frac{(\zeta-\alpha)^2}{2}\sup_{\zeta' \in [\alpha-\zeta, \alpha + \zeta]} |\ddot{h}_\theta(\zeta')|  \leq (\zeta-\alpha)^2\frac{\log^{3}{B_1}}{2\sqrt{2\pi} \{\alpha(1-\alpha)\}^{7/2}}.
\end{align*}
Using this bound with $\zeta = \theta$, we deduce that, for all $B_1$ sufficiently large,
\begin{align*}
|R_1| &= \biggl|\int_I  \biggl\{ \Phi\biggl(\frac{B_1^{1/2}(\alpha - \theta)}{\sigma}\biggr) - \Phi\biggl(\frac{B_1^{1/2}(\alpha - \theta)}{\sqrt{\alpha(1-\alpha)}}\biggr) \nonumber \\ & \hspace{4cm} -  \frac{(1-2\alpha)B_1^{1/2}(\alpha-\theta)^2}{2\{\alpha(1-\alpha)\}^{3/2}}\phi\biggl(\frac{B_1^{1/2}(\alpha - \theta)}{\sqrt{\alpha(1-\alpha)}}\biggr) \biggr\} \, dG_n^\circ(\theta)\biggr| \nonumber \\
& \leq \frac{\log^{3}{B_1}}{2\sqrt{2\pi} \{\alpha(1-\alpha)\}^{7/2}} \int_{\alpha - \epsilon_1}^{\alpha + \epsilon_1}  (\theta-\alpha)^2 |g_n^{\circ}(\theta)|  \, d\theta \nonumber
\\ & \leq \frac{\log^6{B_1}}{3\sqrt{2\pi}B_1^{3/2}\{\alpha(1-\alpha)\}^{7/2}} \sup_{|\theta - \alpha| \leq \epsilon_1} |g_n^{\circ}(\theta)| = o\Bigl(\frac{1}{B_1}\Bigr)
\end{align*}
as $B_1 \to \infty$.

\emph{To bound $R_2$}:  Since $g_n^{\circ}$ is differentiable at $\alpha$, given $\epsilon > 0$, there exists $\delta_\epsilon > 0$ such that
\[
|g_n^\circ(\theta) - g_n^\circ(\alpha) - (\theta-\alpha) \dot{g}_{n}^\circ(\alpha)| < \epsilon |\theta-\alpha|,
\]
for all $|\theta-\alpha| < \delta_\epsilon$.  It follows that, for all $B_1$ sufficiently large,
\begin{align*}
|R_2| &= \biggl|\int_I \biggl\{ \Phi\biggl(\frac{B_1^{1/2}(\alpha - \theta)}{\sqrt{\alpha(1-\alpha)}}\biggr) - \mathds{1}_{\{\theta<\alpha\}}\biggr\} \, dG_n^{\circ}(\theta) \\ 
&\hspace{2.5cm}- \int_{\alpha - \epsilon_1}^{\alpha + \epsilon_1} \biggl\{ \Phi\biggl(\frac{B_1^{1/2}(\alpha - \theta)}{\sqrt{\alpha(1-\alpha)}}\biggr) - \mathds{1}_{\{\theta<\alpha\}}\biggr\} \bigl\{g_n^\circ(\alpha) + (\theta-\alpha) \dot{g}_{n}^\circ(\alpha)\bigr\}\,d\theta\biggr| \\
&\leq \epsilon \int_{\alpha - \epsilon_1}^{\alpha + \epsilon_1} \biggl| \Phi\biggl(\frac{B_1^{1/2}(\alpha - \theta)}{\sqrt{\alpha(1-\alpha)}}\biggr) - \mathds{1}_{\{\theta<\alpha\}}\biggr| |\theta-\alpha| \, d\theta \\
&\leq \frac {\epsilon \alpha(1-\alpha)} {B_1} \int_{-\log{B_1}/\sqrt{\alpha(1-\alpha)}}^{\log{B_1}/\sqrt{\alpha(1-\alpha)}} \bigl| \Phi(-u) - \mathds{1}_{\{u<0\}}\bigr| |u| \, du \nonumber
\\ & \leq  \frac {2\epsilon \alpha(1-\alpha)} {B_1} \int_{0}^{\infty} u \Phi(-u) \, du = \frac {\epsilon \alpha(1-\alpha)} {2B_1}.
\end{align*}
We deduce that $|R_2| = o(B_1^{-1})$ as $B_1 \to \infty$.

\emph{To bound $R_3$}: For large $B_1$, we have
\begin{align*}
%\label{eq--R4}
|R_3| &= \biggl|\int_{\alpha - \epsilon_1}^{\alpha + \epsilon_1} \biggl\{ \Phi\biggl(\frac{B_1^{1/2}(\alpha - \theta)}{\sqrt{\alpha(1-\alpha)}}\biggr) - \mathds{1}_{\{\theta<\alpha\}}\biggr\} \bigl\{g_n^\circ(\alpha) + (\theta-\alpha) \dot{g}_{n}^\circ(\alpha)\bigr\}\,d\theta \nonumber \\
&\hspace{2cm}- \frac {\sqrt{\alpha(1-\alpha)}} {B_1^{1/2}} \int_{-\infty}^{\infty}  \{\Phi(-u) -  \mathds{1}_{\{u<0\}}\} \biggl\{g_n^\circ(\alpha) +  \frac{\sqrt{\alpha(1-\alpha)}}{B_1^{1/2}} u \dot{g}_{n}^\circ(\alpha) \biggr\} \, du\biggr| \nonumber \\ 
&= \frac {2\alpha(1-\alpha)} {B_1} |\dot{g}_{n}^\circ(\alpha)| \int_{\epsilon_1B_1^{1/2}/\{\alpha(1-\alpha)\}^{1/2}}^{\infty}  u\Phi(-u) \, du \nonumber
\\ & \leq \frac {2\{\alpha(1-\alpha)\}^{3/2}} {B_1 \log{B_1}} |\dot{g}_{n}^\circ(\alpha)| \int_{0}^{\infty}  u^2\Phi(-u) \, du = \frac {2\sqrt{2}\{\alpha(1-\alpha)\}^{3/2}} {3\sqrt{\pi}B_1 \log{B_1}} |\dot{g}_{n}^\circ(\alpha)| =  o(B_1^{-1})
\end{align*}
as $B_1 \to \infty$.

\emph{To bound $R_4$}: Since $g_n^{\circ}$ is continuous at $\alpha$, given $\epsilon > 0$, there exists $B_0'\in \mathbb{N}$ such that, for all $B_1 > B_0'$,
\begin{equation}
\label{eq--cont}
\sup_{|\theta-\alpha | \leq \epsilon_1}|g_{n}^\circ(\theta) - g_n^\circ(\alpha)| < \epsilon.
\end{equation}
Hence, given $\epsilon >0$, for all $B_1 > B_0'$,
\begin{align*}
|R_4| &= \biggl|\frac{(1-2\alpha)B_1^{1/2}}{2\{\alpha(1-\alpha)\}^{3/2}} \int_{\alpha - \epsilon_1}^{\alpha + \epsilon_1}  (\alpha-\theta)^2 \phi\biggl(\frac{B_1^{1/2}(\alpha - \theta)}{\sqrt{\alpha(1-\alpha)}}\biggr)\{g_n^\circ(\theta) - g_n^\circ(\alpha)\} \, d\theta \biggr| \nonumber \\
&\leq  \frac{\epsilon |1-2\alpha|}{2B_1} \int_{-\infty}^{\infty} u^2 \phi(-u) \, du = \frac{\epsilon |1-2\alpha|}{2B_1}.
\end{align*}

\emph{To bound $R_5$}: For all $B_1$ sufficiently large, 
\begin{align*}
|R_5| &= \frac{|1-2\alpha|}{B_1}  |g_n^\circ(\alpha)| \int_{\log{B_1}/\sqrt{\alpha(1-\alpha)}}^{\infty} u^2 \phi(-u)  \, du \nonumber
\\ & \leq  \frac{\sqrt{\alpha(1-\alpha)}}{B_1\log{B_1}}  |g_n^\circ(\alpha)| \int_{0}^{\infty} u^3 \phi(-u)  \, du  =  \frac{\sqrt{2\alpha(1-\alpha)}}{\sqrt{\pi}B_1\log{B_1}}  |g_n^\circ(\alpha)| = o\Bigl(\frac{1}{B_1}\Bigr)
\end{align*}
as $B_1 \to \infty$.

\emph{To bound $R_6$}: We write $R_6  = R_{61} + R_{62},$ where
\begin{align*}
R_{61} :=  \frac{(1-2\alpha)}{6B_1^{1/2}\sqrt{\alpha(1-\alpha)}}\int_I \biggl\{1 - \frac{B_1(\alpha - \theta)^2}{\alpha(1-\alpha)}\biggr\} \phi\biggl(\frac{B_1^{1/2}(\alpha - \theta)}{\sqrt{\alpha(1-\alpha)}}\biggr) \,dG_n^{\circ}(\theta)
\end{align*}
and 
\begin{align*}
R_{62} &:= \frac{1}{6B_1^{1/2}}\int_I \frac{(1-2\theta)}{\sigma} \biggl\{1 - \frac{B_1(\alpha - \theta)^2}{\sigma^2}\biggr\} \phi\biggl(\frac{B_1^{1/2}(\alpha - \theta)}{\sigma}\biggr) \,dG_n^{\circ}(\theta) 
\\ & \hspace{30 pt}-  \frac{(1-2\alpha)}{6B_1^{1/2}\sqrt{\alpha(1-\alpha)}}\int_I \biggl\{1 - \frac{B_1(\alpha - \theta)^2}{\alpha(1-\alpha)}\biggr\} \phi\biggl(\frac{B_1^{1/2}(\alpha - \theta)}{\sqrt{\alpha(1-\alpha)}}\biggr) \,dG_n^{\circ}(\theta).
\end{align*}
By~\eqref{eq--cont}, it follows that, for $B_1 > B_0'$ sufficiently large,
\begin{align*}
|R_{61}| &\leq \frac{|1-2\alpha|}{6B_1^{1/2}\sqrt{\alpha(1-\alpha)}} |g_n^{\circ}(\alpha)| \biggl|\int_{\alpha - \epsilon_1}^{\alpha+\epsilon_1} \biggl\{1 - \frac{B_1(\alpha - \theta)^2}{\alpha(1-\alpha)}\biggr\} \phi\biggl(\frac{B_1^{1/2}(\alpha - \theta)}{\sqrt{\alpha(1-\alpha)}}\biggr) \,  d\theta \biggr| 
\\ & \hspace{30 pt} + \epsilon  \frac{|1-2\alpha|}{6B_1^{1/2}\sqrt{\alpha(1-\alpha)}}\int_{\alpha - \epsilon_1}^{\alpha+\epsilon_1} \biggl|1 - \frac{B_1(\alpha - \theta)^2}{\alpha(1-\alpha)}\biggr| \phi\biggl(\frac{B_1^{1/2}(\alpha - \theta)}{\sqrt{\alpha(1-\alpha)}}\biggr) \,d\theta.
\\ & \leq \frac{|1-2\alpha|}{6B_1} |g_n^{\circ}(\alpha)| \biggl|\int_{-\log{B_1}/\sqrt{\alpha(1-\alpha)}}^{\log{B_1}/\sqrt{\alpha(1-\alpha)}} (1 - u^2) \phi(-u) \, du \biggr| 
\\ & \hspace{30 pt} + \epsilon \frac{|1-2\alpha|}{6B_1} \int_{-\infty}^{\infty} (1 + u^2) \phi(-u) \, du \leq \frac{\epsilon}{B_1}.
\end{align*}
We deduce that $R_{61} = o(B_1^{-1})$ as $B_1 \to \infty$.

To control $R_{62}$, by the mean value theorem, we have that for all $B_1$ sufficiently large and all $\zeta \in [\alpha-\epsilon_1,\alpha+\epsilon_1]$,
\begin{align*}
 \sup_{|\theta -\alpha | <\epsilon_1} \Biggl|\frac{(1-2\zeta)}{\sqrt{\zeta(1-\zeta)}} \biggl\{1 - &\frac{B_1(\alpha - \theta)^2}{\zeta(1-\zeta)}\biggr\} \phi\biggl(\frac{B_1^{1/2}(\alpha - \theta)}{\sqrt{\zeta(1-\zeta)}}\biggr)
\\ &  \hspace{30 pt} - \frac{(1-2\alpha)}{\sqrt{\alpha(1-\alpha)}} \biggl\{1 - \frac{B_1(\alpha - \theta)^2}{\alpha(1-\alpha)}\biggr\} \phi\biggl(\frac{B_1^{1/2}(\alpha - \theta)}{\sqrt{\alpha(1-\alpha)}}\biggr)\Biggr| 
\\ & \hspace{6cm} \leq \frac{\log^4{B_1}}{\sqrt{2\pi}\{\alpha(1-\alpha)\}^{7/2}} |\zeta - \alpha|.
\end{align*}  
Thus, for large $B_1$,
\begin{align*}
|R_{62}| &\leq  \frac{\log^4{B_1}}{6 \sqrt{2\pi} B_1^{1/2}\{\alpha(1-\alpha)\}^{7/2}} \sup_{|\theta-\alpha| \leq \epsilon_1} |g_n^{\circ}(\theta)| \int_{\alpha-\epsilon_1}^{\alpha + \epsilon_1}|\theta - \alpha| \, d\theta 
\\ & \leq \frac{ \log^6{B_1} \{1+|g_n^{\circ}(\alpha)|\}}{6\sqrt{2\pi} B_1^{3/2}\{\alpha(1-\alpha)\}^{7/2}} = o\Bigl(\frac{1}{B_1}\Bigr).
\end{align*}
We deduce that $|R_6| = o(B_1^{-1})$ as $B_1 \to \infty$.

\emph{To bound $R_7$}: write $R_7 = R_{71} + R_{72}$, where
\[
R_{71} :=  \frac{1/2 -  \llbracket B_1\alpha \rrbracket}{B_1^{1/2}\sqrt{\alpha(1-\alpha)}} \int_{\alpha - \epsilon_1}^{\alpha+\epsilon_1}   \phi\biggl(\frac{B_1^{1/2}(\alpha - \theta)}{\sqrt{\alpha(1-\alpha)}}\biggr)\{ g_n^{\circ}(\theta) - g_n^{\circ}(\alpha)\} \, d\theta, 
\]
and
\begin{align*}
R_{72} &:= \frac{1/2 -  \llbracket B_1\alpha \rrbracket}{B_1^{1/2}} \int_I \biggl\{ \frac{1}{\sigma} \phi\biggl(\frac{B_1^{1/2}(\alpha - \theta)}{\sigma}\biggr) - \frac{1}{\sqrt{\alpha(1-\alpha)}} \phi\biggl(\frac{B_1^{1/2}(\alpha - \theta)}{\sqrt{\alpha(1-\alpha)}}\biggr)\biggr\} \,dG_n^{\circ}(\theta).
\end{align*}
By the bound in (\ref{eq--cont}), given $\epsilon>0$, for all $B_1$ sufficiently large,
\begin{align*}
|R_{71}| &\leq  \frac{\epsilon}{2 B_1^{1/2}\sqrt{\alpha(1-\alpha)}} \int_{-\infty}^{\infty}   \phi\biggl(\frac{B_1^{1/2}(\alpha - \theta)}{\sqrt{\alpha(1-\alpha)}}\biggr) \, d\theta = \frac{\epsilon}{2 B_1}.
\end{align*}
Moreover, by the mean value theorem, for all $B_1$ sufficiently large and all $|\zeta - \alpha | \leq \epsilon_1$,
\begin{align*}
 \sup_{|\theta -\alpha | <\epsilon_1} \Bigg|\frac{1}{\sqrt{\zeta(1-\zeta)}} \phi\biggl(\frac{B_1^{1/2}(\alpha - \theta)}{\sqrt{\zeta(1-\zeta)}}\biggr) - \frac{1}{\sqrt{\alpha(1-\alpha)}}&\phi\biggl(\frac{B_1^{1/2}(\alpha - \theta)}{\sqrt{\alpha(1-\alpha)}}\biggr)\Bigg| \\ 
&\leq \frac{\log^2{B_1}}{\sqrt{2\pi} \{\alpha(1-\alpha)\}^{5/2}} |\zeta - \alpha|.
\end{align*} 
It follows that, for all $B_1$ sufficiently large,
\begin{align*}
|R_{72}| &\leq \frac{\log^2{B_1}}{2\sqrt{2\pi}B_1^{1/2} \{\alpha(1-\alpha)\}^{5/2}} \sup_{|\theta-\alpha|\leq \epsilon_1} |g_n^{\circ}(\theta)| \int_{\alpha-\epsilon_1}^{\alpha+\epsilon_1} |\theta-\alpha| \, d\theta 
\\ & \leq  \frac{\log^4{B_1} \{1+|g_n^{\circ}(\alpha)|\}}{2\sqrt{2\pi}B_1^{3/2} \{\alpha(1-\alpha)\}^{5/2}}.
\end{align*}
We deduce that $|R_7| = o(B_1^{-1})$ as $B_1 \to \infty$.

\emph{To bound $R_8$}: We have
\[
|R_8| = \frac{2(1/2 -  \llbracket B_1\alpha \rrbracket)}{B_1} |g_n^\circ(\alpha)|\int_{\epsilon_1B_1^{1/2}/\{\alpha(1-\alpha)\}^{1/2}}^\infty \phi(-u) \, du = o\Bigl(\frac{1}{B_1}\Bigr)
\]
as $B_1 \to \infty$.

We have now established the claim at~\eqref{Eq:AllErrorTerms}, and the result follows.
\end{prooftitle}

\begin{prooftitle}{of Theorem~\ref{Thm:BaggingBound}.}
In the case where $B_1 < \infty$, we have
\begin{align*}
R({C}_n^{\mathrm{RP}}) &- R(C^{\mathrm{Bayes}})
\\ &=  \int_{\mathbb{R}^p}\bigl[\eta(x) (\mathds{1}_{\{{C}_n^{\mathrm{RP}}(x) = 0\}} - \mathds{1}_{\{C^{\mathrm{Bayes}}(x) = 0\}})  + \{1-\eta(x)\} (\mathds{1}_{\{{C}_n^{\mathrm{RP}}(x) =1\}} - \mathds{1}_{\{C^{\mathrm{Bayes}}(x) =1\}})\bigl] \, dP_{X}(x) \\
&= \int_{\mathbb{R}^p}\bigl\{|2\eta(x)-1| |\mathds{1}_{\{\nu_n(x) < \alpha\}} - \mathds{1}_{\{\eta(x) < 1/2\}}|\bigr\} \, dP_{X}(x) \\
&=  \int_{\mathbb{R}^p}\bigl\{|2\eta(x)-1| \mathds{1}_{\{\nu_n(x) \geq \alpha\}}\mathds{1}_{\{\eta(x) < 1/2\}} + |2\eta(x)-1| \mathds{1}_{\{\nu_n(x) < \alpha\}}\mathds{1}_{\{\eta(x) \geq 1/2\}}\bigr\}\, dP_{X}(x) \\
&\leq \int_{\mathbb{R}^p}\Bigl[\frac{1}{\alpha}|2\eta(x)-1| \nu_n(x)\mathds{1}_{\{\eta(x) < 1/2\}} + \frac{1}{1-\alpha}|2\eta(x)-1| \{1-\nu_n(x)\}\mathds{1}_{\{\eta(x) \geq 1/2\}} \Bigr]\, dP_{X}(x).
\end{align*}
It follows that
\begin{align*}
\mathbf{E}\{R({C}_n^{\mathrm{RP}})\} - R(C^{\mathrm{Bayes}}) & \leq \mathbf{E}\biggl\{\int_{\mathbb{R}^p} \frac{1}{\alpha}|2\eta(x)-1|\mathds{1}_{\{C_n^{\mathbf{A}_{1}} (x) = 1\}} \mathds{1}_{\{\eta(x) < 1/2\}} \\
&\hspace{3cm}+ \frac{1}{1-\alpha}|2\eta(x)-1|\mathds{1}_{\{C_n^{\mathbf{A}_{1}} (x) = 0\}} \mathds{1}_{\{\eta(x) \geq 1/2\}} \, dP_{X}(x)\biggr\} \\
&\leq \frac{1}{\min(\alpha,1-\alpha)} \mathbf{E}\biggl\{\int_{\mathbb{R}^p} |2\eta(x)-1|\bigl|\mathds{1}_{\{C_n^{\mathbf{A}_{1}} (x) = 0\}}  - \mathds{1}_{\{\eta(x) < 1/2\}}\bigr| \, dP_{X}(x)\biggr\} \\
& =\frac{1}{\min(\alpha,1-\alpha)} \bigl[ \mathbf{E}\{R(C_n^{\mathbf{A}_1})\} - R(C^{\mathrm{Bayes}})\bigr],
\end{align*}
as required.  When $B_1 = \infty$, we replace both occurrences of $R({C}_n^{\mathrm{RP}})$ with $R({C}_n^{\mathrm{RP}^*})$ and the argument goes through in almost identical fashion after changing $\nu_n$ to $\mu_n$.
\end{prooftitle} 

\begin{prooftitle}{of Theorem~\ref{Thm:Main}.}
First write
\[
\mathbf{E}\{R(C_n^{\mathbf{A}_1})\} - R(C^{\mathrm{Bayes}}) = \mathbf{E}({R}_n^{\mathbf{A}_1}) - R(C^{\mathrm{Bayes}}) + \epsilon_n.
\]
Using assumption~2, we have that
\begin{align*}
\mathbf{E}({R}_n^{\mathbf{A}_1}) &= \mathbf{E}\bigl(R_n^{\mathbf{A}_1}\mathds{1}_{\{R_n^{\mathbf{A}_1} \leq R_n^* + |\epsilon_n|\}}\bigr) + \mathbf{E}\bigl(R_n^{\mathbf{A}_1}\mathds{1}_{\{R_n^{\mathbf{A}_1} > R_n^* + |\epsilon_n|\}}\bigr) \\
&\leq R_n^* + |\epsilon_n| + \mathbf{P}(R_n^{\mathbf{A}_1} > R_n^* + |\epsilon_n|) \\
&= R_n^* + |\epsilon_n| + \mathbf{P}(R_n^{\mathbf{A}_{1,1}} > R_n^* + |\epsilon_n|)^{B_2} \\
&\leq R_n^* + |\epsilon_n| + (1-\beta)^{B_2}.
\end{align*}
But, for any $A \in \mathcal{A}$ and by definition of ${R}_n^*$ and $\epsilon_n^A$, we have ${R}_n^* \leq {R}_n^A = R({C}_n^{A}) - \epsilon_n^{A}$.  It therefore follows by Theorem~\ref{Thm:BaggingBound} that
\begin{align*}
\mathbf{E}\{R({C}_n^{\mathrm{RP}})\} - R(C^{\mathrm{Bayes}}) &\leq \frac{1}{\min(\alpha,1-\alpha)}\bigl[\mathbf{E}\{R({C}_n^{\mathbf{A}_1})\} - R(C^{\mathrm{Bayes}})\bigr] \\
&\leq \frac{R({C}_n^{A}) - R(C^{\mathrm{Bayes}})}{\min(\alpha,1-\alpha)} + \frac{2|\epsilon_n| - \epsilon_n^{A}}{\min(\alpha,1-\alpha)} + \frac{(1-\beta)^{B_{2}}}{\min(\alpha,1-\alpha)},
\end{align*}
as required.
\end{prooftitle}

\begin{prooftitle}{of Proposition~\ref{Prop:Cond}.}
For a Borel set $C \subseteq \mathbb{R}^d$, let $P_{A^*X}(C) := \int_{\{x:A^*x \in C\}} \, dP_X(x)$, so that $P_{A^*X}$ is the marginal distribution of $A^*X$.  Further, for $z \in \mathbb{R}^d$, write $P_{X|A^*X = z}$ for the conditional distribution of $X$ given $A^*X =z$.  If $Y$ is independent of $X$ given $A^*X$, and if $B$ is a Borel subset of $\mathbb{R}^p$, then
\begin{align*}
\label{Eq:CondProb}
\int_B \eta^{A^*}(A^*x) \, dP_X(x) &= \int_{\mathbb{R}^d} \int_{B \cap \{w:A^*w=z\}} \eta^{A^*}(A^*w) \, dP_{X|A^*X=z}(w) \, dP_{A^*X}(z) \\
&= \int_{\mathbb{R}^d} \eta^{A^*}(z)\mathbb{P}(X \in B|A^*X=z) \, dP_{A^*X}(z) \\
&= \int_{\mathbb{R}^d} \mathbb{P}(Y=1,X \in B|A^*X=z) \, dP_{A^*X}(z) \\
&= \mathbb{P}(Y=1,X \in B) = \int_B \eta(x) \, dP_X(x).
\end{align*}
We deduce that $P_X(\{x \in \mathbb{R}^p:\eta(x) \neq \eta^{A^*}(A^*x)\}) = 0$; in particular, assumption~3 holds, as required.
\end{prooftitle}

\begin{prooftitle}{of Proposition~\ref{Prop:A^*}.}
We have
\begin{align*}
R(C^{A^*-\mathrm{Bayes}}) &= \int_{\mathbb{R}^p \times \{0,1\}} \mathds{1}_{\{C^{A^*-\mathrm{Bayes}}(A^*x) \neq y\}} \, dP(x,y) \\
&= \int_{\mathbb{R}^p} \eta(x)\mathds{1}_{\{\eta^{A^*}(A^*x) < 1/2\}} \, dP_X(x) + \int_{\mathbb{R}^p} \{1-\eta(x)\}\mathds{1}_{\{\eta^{A^*}(A^*x) \geq 1/2\}} \, dP_X(x) \\
&= \int_{\mathbb{R}^p} \eta(x)\mathds{1}_{\{\eta(x) < 1/2\}} \, dP_X(x) + \int_{\mathbb{R}^p} \{1-\eta(x)\}\mathds{1}_{\{\eta(x) \geq 1/2\}} \, dP_X(x) 
\\ & = R(C^{\mathrm{Bayes}}),
\end{align*}
where we have used assumption~3 to obtain the penultimate equality.  
\end{prooftitle}

\begin{prooftitle}{of Proposition~\ref{cor--edgeworth}.}
Recall that $\sigma^2 := \theta(1-\theta)$.  Let 
\begin{align*}
F_{B_1}^*(s) = F_{B_1}^*(s,\theta) &:= \int_{-\infty}^{\infty} e^{is t} \, dF_{B_1}(t) = \biggl\{(1-\theta)\exp\biggl(-\frac{is\theta}{B_1^{1/2}\sigma}\biggr) + \theta \exp\biggl(\frac{is(1-\theta)}{B_1^{1/2}\sigma}\biggr)\biggr\}^{B_1}.
\end{align*}
Moreover, let  $P(t) := \frac{\phi(t)p(t)}{B_1^{1/2}}$ and $Q(t) := \frac{\phi(t)q(t)}{B_1^{1/2}}$.  By, for example, \citet[Chapter 8, Section 43]{Gnedenko:54}, we have
\[
\Phi^*(s) := \int_{\mathbb{R}} \exp(i s t) \, d\Phi(t) = \exp(-s^2/2), 
\]
\[
P^*(s) := \int_{\mathbb{R}} \exp(ist) \, dP(t) = -\frac{1-2\theta}{6B_1^{1/2}\sigma} i s^3 \exp(-s^2/2) 
\]
and
\[
Q^*(s) := \int_{\mathbb{R}} \exp(ist) \, dQ(t) = - \frac{s}{2\pi B_1^{1/2}\sigma} \sum_{l \in \mathbb{Z} \setminus \{0\}} \frac{\exp(i 2\pi B_1 l \theta)}{l} \exp\Bigl\{-\frac{1}{2}\bigl(s + 2\pi B_1^{1/2} \sigma l\bigr)^2\Bigr\}.
\]
Thus
\begin{align*}
G_{B_1}^*(s) &= G_{B_1}^*(s,\theta) := \int_{\mathbb{R}} \exp(i  s t) \, dG_{B_1}(t) = \Phi^*(s) + P^*(s) + Q^*(s)
\\  &= \exp(-s^2/2) - \frac{1-2\theta}{6B_1^{1/2}\sigma} is^3 \exp(-s^2/2) 
\\ &\hspace{3.5 cm} - \frac{s}{2\pi B_1^{1/2}\sigma} \sum_{l \in \mathbb{Z} \setminus \{0\}} \frac{\exp(i2\pi B_1l \theta)}{l} \exp\Bigl\{-\frac{1}{2}\bigl(s+2\pi B_1^{1/2} \sigma l\bigr)^2\Bigr\}.
\end{align*}
Letting $c_2 > 0$ be the constant given in the statement of Theorem~\ref{thm--esseen} (in fact we assume without loss of generality that $c_2 > \pi$), we show that there exists a constant $C' > 0$ such that, for all $B_1 \in \mathbb{N}$,
\begin{equation}
\label{eq--fourierbound} 
\sup_{\theta \in (0,1)} \sigma^3 \int_{-c_2 B_1^{1/2}\sigma}^{c_2 B_1^{1/2}\sigma} \biggl|\frac{F_{B_1}^*(s,\theta) - G_{B_1}^*(s,\theta)}{s}\biggr|\, ds \leq \frac {C'}{B_1}.
\end{equation}
To show (\ref{eq--fourierbound}), write
\begin{align}
\label{eq--interval}
 \int_{-c_2 B_1^{1/2}\sigma}^{c_2 B_1^{1/2}\sigma} &\biggl|\frac{F_{B_1}^*(s) - G_{B_1}^*(s)}{s}\biggr|\, ds  = \int_{-S_1}^{S_1} \biggl|\frac{F_{B_1}^*(s) - G_{B_1}^*(s)}{s}\biggr| \, ds  \nonumber
\\ & + \int_{S_1 \leq |s| \leq S_2} \biggl|\frac{F_{B_1}^*(s) - G_{B_1}^*(s)}{s}\biggr|\, ds  + \int_{S_2 \leq |s| \leq c_2 B_1^{1/2}\sigma} \biggl|\frac{F_{B_1}^*(s) - G_{B_1}^*(s)}{s}\biggr|\, ds,
\end{align}
where $S_1 := \frac{B_1^{1/2}\sigma^{3/2}}{32 (3\theta^2 - 3\theta +1)^{3/4}}$ and $S_2 := \pi B_1^{1/2} \sigma$. Note that $S_1 \leq S_2/2$ for all $\theta \in (0,1)$.

We bound each term in (\ref{eq--interval}) in turn.  By \citet[Theorem ~1, Section ~41]{Gnedenko:54}, there exists a universal constant $C_3 >0$, such that, for all $|s| \leq S_1$,
\begin{align*}
&|F_{B_1}^*(s,\theta) - \Phi^*(s) - P^*(s)|  \leq \frac{C_3}{B_1 \sigma^3}(s^4 + s^6)\exp(-s^2/4).
\end{align*}
Thus
\begin{equation}
\label{eq--S12}
\int_{-S_1}^{S_1} \biggl|\frac{F_{B_1}^*(s) - \Phi^*(s) - P^*(s)}{s}\biggr| \, ds \leq \frac{C_3}{B_1 \sigma^3} \int_{-\infty}^{\infty} (|s|^3 + |s|^5) \exp(-s^2/4) \, ds = \frac{144 C_3}{B_1\sigma^{3}}.
\end{equation}
Moreover, observe that $\bigl(s+2\pi B_1^{1/2} \sigma l\bigr)^2 \geq s^2 + 2\pi^2 B_1 \sigma^2 l^2$ for all $|s|\leq S_1$.  Thus, for $|s|\leq S_1$,
\begin{align*}
\biggl|\frac{Q^*(s)}{s}\biggr| & \leq  \frac{1}{2\pi B_1^{1/2}\sigma} \biggl|\sum_{l \in \mathbb{Z} \setminus \{0\}} \frac{\exp(i 2\pi B_1 l \theta)}{l} \exp\Bigl\{-\frac{1}{2}\bigl(s+2\pi B_1^{1/2} \sigma l\bigr)^2\Bigr\}\biggr|
\\ & \leq \frac{\phi(s)}{\sqrt{2\pi} B_1^{1/2}\sigma} \int_{-\infty}^{\infty} \exp\bigl(-\pi^2 B_1 \sigma^2 u^2\bigr) \, du = \frac{\phi(s)}{\sqrt{2} \pi B_1\sigma^2}.
\end{align*}
It follows that
\begin{equation}
\label{Eq:Q*}
\int_{-S_1}^{S_1} \biggl|\frac{Q^*(s)}{s}\biggr| \, ds \leq \frac{1}{\sqrt{2} \pi B_1\sigma^2}.
\end{equation}
For $|s| \in [S_1,S_2]$, observe that
\[
|F_{B_1}^*(s)| = \biggl[1-2\sigma^2\biggl\{1-\cos\Bigl(\frac{s}{B_1^{1/2}\sigma}\Bigr)\biggr\}\biggr]^{B_1/2} \leq \exp(-s^2/8).
\] 
Thus
\begin{equation}
\label{Eq:F*middle}
\int_{S_1 \leq |s| \leq S_2} \biggl|\frac{F_{B_1}^*(s)}{s}\biggr| \, ds \leq  \frac{2}{S_1^2} \int_{S_1}^{S_2} s \exp(-s^2/8) \, ds \leq \frac{2^{13}}{B_1 \sigma^3}.
\end{equation}
Now,
\begin{equation}
\label{Eq:Phi*middle}
\int_{S_1 \leq |s| \leq S_2} \biggl|\frac{\Phi^*(s)}{s}\biggr| \, ds \leq \frac{2}{S_1^2} \int_0^\infty s\exp(-s^2/2)\, ds \leq \frac{2^{11}}{B_1 \sigma^3},
\end{equation}
and
\begin{equation}
\label{Eq:P*middle}
\int_{S_1 \leq |s| \leq S_2} \biggl|\frac{P^*(s)}{s}\biggr| \, ds \leq \frac{1}{3 S_1 B_1^{1/2} \sigma} \int_0^\infty s^3\exp(-s^2/2)\, ds \leq \frac{2^6}{3\sqrt{2} B_1\sigma^3}.
\end{equation}
To bound the final term, observe that, for all $|s| \in [S_1, S_2]$, since $(a+b)^2 \geq (a^2+b^2)/5$ for all $|a| \leq |b|/2$, we have
\begin{equation}
\label{eq--S1S2}
\int_{S_1 \leq |s| \leq S_2} \biggl|\frac{Q^*(s)}{s}\biggr| \, ds  \leq \frac{1}{2\pi B_1^{1/2}\sigma} \int_{S_1 \leq |s| \leq S_2} e^{-s^2/10} \int_{-\infty}^\infty e^{-2\pi^2B_1\sigma^2u^2/5} \, du \, ds \leq \frac{5}{4\pi B_1\sigma^3}.
\end{equation}

Finally, for $|s| \in  [S_2, c_2 B_1^{1/2}\sigma]$, note that
\begin{align}
\label{Eq:PhiPend}
\int_{S_2 \leq |s| \leq c_2 B_1^{1/2}\sigma} \biggl|\frac{\Phi^*(s) + P^*(s)}{s}\biggr| \, ds &\leq \frac{2}{S_2^2} \int_{0}^{\infty} s e^{-s^2/2} \, ds + \frac{1}{3 S_2 B_1^{1/2}\sigma} \int_{0}^{\infty} s^3 e^{-s^2/2} \, ds \nonumber
\\ & \leq \frac{1}{\pi^2 B_1\sigma^3} \biggl(1 + \frac{\pi}{3}\biggr). 
\end{align}
To bound the remaining terms, by substituting $s = B_1^{1/2}\sigma u$, we see that
\begin{align}
\label{eq--Ij1}
\int_{S_2}^{c_2 B_1^{1/2}\sigma} \biggl|\frac{F_{B_1}^*(s) - Q_{B_1}^*(s)}{s}\biggr| \, ds &=  \int_{\pi}^{c_2} \biggl|\frac{F_{B_1}^*(B_1^{1/2} \sigma u) - Q_{B_1}^*(B_1^{1/2} \sigma u)}{u}\biggr|\, du \nonumber
\\ & = \sum_{j=1}^J \int_{\pi(2j-1)}^{\pi(2j+1)} \biggl|\frac{F_{B_1}^*(B_1^{1/2}\sigma u) - Q_{B_1}^*(B_1^{1/2}\sigma u)}{u}\biggr|\, du \nonumber
\\ & \hspace{30 pt}  +  \int_{\pi(2J+1)}^{c_2} \biggl|\frac{F_{B_1}^*(B_1^{1/2}\sigma u) - Q_{B_1}^*(B_1^{1/2}\sigma u)}{u}\biggr|\, du,
\end{align} 
where $J := \lfloor \frac{c_2-\pi}{2\pi}\rfloor$.
Let 
\begin{align}
\label{Eq:Ij}
I_j :=& \int_{\pi(2j-1)}^{\pi(2j+1)} \biggl|\frac{F_{B_1}^*(B_1^{1/2}\sigma u) - Q_{B_1}^*(B_1^{1/2}\sigma u)}{u}\biggr|\, du \nonumber
\\  =& \int_{-\pi}^{\pi} \biggl|\frac{F_{B_1}^*\bigl(B_1^{1/2}\sigma (v+2\pi j)\bigr) - Q_{B_1}^*\bigl(B_1^{1/2}\sigma (v+2\pi j)\bigr)}{v+2\pi j}\biggr|\, dv.
\end{align}
Observe that 
\begin{align*}
F_{B_1}^*\bigl(B_1^{1/2}\sigma (v+2\pi j)\bigr) &=  \Bigl[(1-\theta)\exp\bigl\{-{i(v+2\pi j)\theta}\bigr\} + \theta \exp\bigl\{{i(v+2\pi j)(1-\theta)}\bigr\}\Bigr]^{B_1}
\\ & = \exp(-i 2\pi B_1 j\theta) \bigl[(1-\theta)\exp(-iv\theta) + \theta \exp\{iv(1-\theta)\}\bigr]^{B_1}
\\ & =  \exp(-i 2\pi  B_1j\theta) F_{B_1}^*(B_1^{1/2}\sigma v).
\end{align*}
Similarly,
\begin{align*}
Q_{B_1}^*\bigl(B_1^{1/2}\sigma (v+2\pi j)\bigr) &= - \frac{(v+2\pi j)}{2\pi} \sum_{l \in \mathbb{Z}\setminus\{0\}} \frac{\exp(i 2\pi  B_1 l \theta)}{l} \exp\biggl\{-\frac{B_1\sigma^2}{2}\bigl(v+2\pi j + 2\pi l\bigr)^2\biggr\}
\\ & =  \frac{(v+2\pi j)\exp(-i 2 \pi  B_1 j \theta)}{2\pi j}  \exp\biggl(-\frac{B_1\sigma^2 v^2}{2}\biggr) 
\\ & \hspace{-10 pt} - \frac{(v+2\pi j)}{2\pi} \sum_{l \in \mathbb{Z}\setminus\{0,-j\}} \frac{\exp(i 2\pi  B_1 l \theta)}{l} \exp\biggl\{-\frac{B_1\sigma^2}{2}\bigl(v+2\pi j + 2\pi l\bigr)^2\biggr\}.
\end{align*}
But, for $v \in [-\pi,\pi]$,
\begin{align*}
 \biggl| \frac{1}{2\pi} \sum_{l \in \mathbb{Z}\setminus\{0,-j\}} &\frac{e^{i 2\pi  B_1 l \theta}}{l} \exp\biggl\{-\frac{B_1\sigma^2}{2}\bigl(v+2\pi j + 2\pi l\bigr)^2\biggr\}\biggr| \leq \frac{1}{2\pi} \sum_{m \in \mathbb{Z}\setminus\{0\}} e^{-\frac{B_1\sigma^2}{2}(v+2\pi m)^2} \\
&\leq \frac{e^{-B_1\sigma^2v^2/10}}{2\pi} \sum_{m \in \mathbb{Z}\setminus\{0\}} e^{-2\pi^2B_1\sigma^2m^2/5} \leq \frac{e^{-B_1\sigma^2v^2/10}}{\pi(e^{2\pi^2B_1\sigma^2/5}-1)} \leq \frac{5e^{-B_1\sigma^2v^2/10}}{2\pi^3B_1\sigma^2}. 
\end{align*} 
It follows that 
\begin{align}
\label{eq--Ij2}
I_j & \leq \int_{-\pi}^{\pi} \biggl|\frac{F_{B_1}^*(B_1^{1/2}\sigma v) - \bigl( \frac{v}{2\pi j} + 1 \bigr) \exp\bigl(-\frac{B_1\sigma^2 v^2}{2}\bigr)}{v+2\pi j}\biggr|\, dv + \frac{5\sqrt{5}}{\sqrt{2}\pi^{5/2}B_1^{3/2}\sigma^3}.
\end{align}
Now
\begin{align}
\label{eq--Ij3}
\int_{-\pi}^{\pi} &\biggl|\frac{F_{B_1}^*(B_1^{1/2}\sigma v) -   \exp\bigl(-\frac{B_1\sigma^2 v^2}{2}\bigr)}{v+2\pi j}\biggr|\, dv \leq \frac{1}{\pi j B_1^{1/2}\sigma}\int_{-\pi B_1^{1/2} \sigma}^{\pi  B_1^{1/2}\sigma} \big|F_{B_1}^*(u) -  e^{-u^2/2}\big|\, du \nonumber \\
 = & \frac{1}{\pi j B_1^{1/2}\sigma}\int_{-S_3}^{S_3} \big|F_{B_1}^*(u) -  e^{-u^2/2}\big|\, du + \frac{1}{\pi j B_1^{1/2}\sigma}\int_{S_3 \leq |u| \leq \pi B_1^{1/2} \sigma} \big|F_{B_1}^*(u) -  e^{-u^2/2}\big| \, du,
\end{align}
where $S_3 := \frac{ B_1^{1/2}\sigma}{5(2\theta^2-2\theta + 1)} \geq S_1$. By \citet[Theorem 2, Section 40]{Gnedenko:54}, we have that
\begin{align}
\label{eq--Ij4}
 \frac{1}{\pi j B_1^{1/2}\sigma}\int_{-S_3}^{S_3} \big|F_{B_1}^*(u) -  e^{-u^2/2}\big|\, du & \leq \frac{7}{6\pi j B_1\sigma^2}\int_{-S_3}^{S_3} |u|^3 e^{-u^2/4} \, du \leq \frac{56}{3 \pi jB_1\sigma^2}.
\end{align}
Moreover,
\begin{align}
\label{eq--Ij5}
\frac{1}{\pi j B_1^{1/2}\sigma}\int_{S_3 \leq |u| \leq \pi B_1^{1/2} \sigma} \big|F_{B_1}^*(u) -  e^{-u^2/2}\big| \, du \leq   \frac{2}{\pi j S_3 B_1^{1/2}\sigma}  \int_{0}^{\infty} u (e^{-u^2/8} + e^{-u^2/2}) \, du \leq   \frac{50}{\pi j B_1 \sigma^2}.
\end{align}
Finally,
\begin{align}
\label{eq--Ij6}
 \frac 1 {2\pi j} \int_{-\pi}^{\pi} \frac{|v|}{|v|+2\pi j} \exp\biggl(-\frac{B_1\sigma^2 v^2}{2}\biggr)\, dv & \leq \frac{1}{2\pi^{2}j^{2}} \int_{0}^{\pi} v \exp\biggl(-\frac{B_{1}\sigma^{2}v^{2}}{2}\biggr) \, dv \leq \frac{1}{2\pi^{2} j^{2}B_1\sigma^2}.
\end{align}     
By (\ref{eq--Ij1}), (\ref{Eq:Ij}), (\ref{eq--Ij2}), (\ref{eq--Ij3}), (\ref{eq--Ij4}), (\ref{eq--Ij5}) and (\ref{eq--Ij6}), it follows that 
\begin{align}
\label{eq--S2B1}
\int_{S_2 \leq |s| \leq c_2 B_1^{1/2}\sigma} \biggl|\frac{F_{B_1}^*(s) - Q_{B_1}^*(s)}{s}\biggr| \, ds & \leq \frac{10\sqrt{5}(J+1)}{\sqrt{2}\pi^{5/2}B_1^{3/2}\sigma^3} + \frac{140}{\pi B_1 \sigma^2} \sum_{j=1}^{J+1} \frac{1}{j} \nonumber
\\ &  \leq \frac{10\sqrt{5}(J+1)}{\sqrt{2}\pi^{5/2}B_1^{3/2}\sigma^3} + \frac{140}{\pi B_1 \sigma^2}\{1 + \log{(J+1)}\}.
\end{align}
By (\ref{eq--interval}), (\ref{eq--S12}), (\ref{Eq:Q*}), (\ref{Eq:F*middle}), (\ref{Eq:Phi*middle}), (\ref{Eq:P*middle}), (\ref{eq--S1S2}), (\ref{Eq:PhiPend}) and (\ref{eq--S2B1}), we conclude that~(\ref{eq--fourierbound}) holds.  The result now follows from Theorem~\ref{thm--esseen}, by taking $c_1 = \frac{1}{B_1^{1/2}\sigma}$, $C_1 = \frac{1}{3B_1^{1/2} \sigma}$ and $S =  c_2 B_1^{1/2}\sigma$ in that result.
\end{prooftitle}

\section*{Acknowledgements}

Both authors are supported by an Engineering and Physical Sciences Research Council Fellowship EP/J017213/1; the second author is also supported by a Philip Leverhulme prize.  The authors thank Rajen Shah, Ming Yuan and the anonymous reviewers for helpful comments.

\clearpage

{\Large \textbf{Random-projection ensemble classification: supplementary material}}

\medskip

This is the supplementary material for \citet{CanningsSamworth:16Main}, hereafter referred to as the main text.

\section{A bound on the Monte Carlo variance of $R(C_n^{\mathrm{RP}})$}

The following bound on the asymptotic Monte Carlo variance of $R(C_n^{\mathrm{RP}})$ complements the result on its Monte Carlo expectation presented in Theorem~\ref{thm--Bvar}:
\begin{proposition}
\label{Prop:variance}
Assume assumption~1.  Then
\[
\limsup_{B_1 \to \infty} B_1 \mathbf{Var}\{R({C}_n^{\mathrm{RP}})\} \leq  \alpha(1 - \alpha)\bar{g}_n^2(\alpha),
\]
where $\bar{g}_n(\alpha) := \pi_0 g_{n,0}(\alpha) + \pi_1 g_{n,1}(\alpha)$.
\end{proposition}
\begin{proof}
Recall that the training data are considered fixed.  First write
\[
R({C}_n^{\mathrm{RP}}) = \pi_0 \int_{\mathbb{R}^p} \mathds{1}_{\{\nu_n(x) \geq \alpha\}} \, dP_0(x) + \pi_1 \int_{\mathbb{R}^p} \mathds{1}_{\{\nu_n(x) < \alpha\}} \, dP_1(x).
\]
Now, for $r=0,1$,
\begin{align*}
\mathbf{Var}&\Bigl(\int_{\mathbb{R}^p} \mathds{1}_{\{\nu_n(x) < \alpha\}} \, dP_r(x)\Bigr) = \mathbf{E}\Bigl\{\Bigl(\int_{\mathbb{R}^p} \mathds{1}_{\{\nu_n(x) < \alpha\}} - \mathbf{P}\{\nu_n(x) < \alpha\} \, dP_r(x)\Bigr)^2\Bigr\} \\
& = \mathbf{E} \Bigl(\int_{\mathbb{R}^p}\int_{\mathbb{R}^p}[\mathds{1}_{\{\nu_n(x) < \alpha\}} - \mathbf{P}\{\nu_n(x) < \alpha\}][\mathds{1}_{\{\nu_n(x') < \alpha\}} - \mathbf{P}\{\nu_n(x') < \alpha\} ]\, dP_r(x)\,dP_r(x')\Bigr) \\
& = \int_{\mathbb{R}^p}\int_{\mathbb{R}^p} \mathbf{P}\{\nu_n(x) < \alpha, \nu_n(x') < \alpha\} - \mathbf{P}\{\nu_n(x) < \alpha\} \mathbf{P}\{\nu_n(x') < \alpha\} \, dP_r(x)\,dP_r(x') \\
& \leq \int_{\mathbb{R}^p}\int_{\mathbb{R}^p} \min\bigl[\mathbf{P}\{\nu_n(x) < \alpha\},\mathbf{P}\{\nu_n(x') < \alpha\}\bigr] - \mathbf{P}\{\nu_n(x) < \alpha\} \mathbf{P}\{\nu_n(x') < \alpha\} \, dP_r(x)\,dP_r(x'),
\end{align*} 
where we used Fubini's theorem for the final equality.  Similarly to the proof of Theorem~\ref{thm--Bvar}, and letting $T \sim \mathrm{Bin}(B_1,\theta)$ and $T' \sim \mathrm{Bin}(B_1,\theta')$,
\begin{align*}
\int_{\mathbb{R}^p} & \int_{\mathbb{R}^p} \min\bigl[\mathbf{P}\{\nu_n(x) < \alpha\},\mathbf{P}\{\nu_n(x') < \alpha\}\bigr] - \mathbf{P}\{\nu_n(x) < \alpha\} \mathbf{P}\{\nu_n(x') < \alpha\} \, dP_r(x)\,dP_r(x')
\\ & = \int_{[0,1]} \int_{[0,1]}  \min\bigl\{\mathbb{P} (T < B_1\alpha),\mathbb{P} (T' < B_1\alpha)\bigr\} - \mathbb{P} (T' < B_1\alpha)\mathbb{P} (T < B_1\alpha)  \, d{G}_{n,r}(\theta) \, d{G}_{n,r}(\theta')
\\ &  \leq 2\int_{[0,1]} \int_{[0,\theta']}  \mathbb{P} (T' < B_1\alpha)\mathbb{P}(T \geq B_1\alpha) \, d{G}_{n,r}(\theta) \, d{G}_{n,r}(\theta').
\end{align*}
Now, again similarly to the proof of Theorem~\ref{thm--Bvar}, and assuming $B_1$ is large enough that $[\alpha - \epsilon_1,\alpha+\epsilon_1] \subseteq (0,1)$, where $\epsilon_1 := B_1^{-1/2}\log{B_1}$, we have 
\begin{align}
\label{Eq:Step1}
\int_{[0,1]} \int_{[0,\theta']}  \mathbb{P} &(T' < B_1\alpha)\mathbb{P}(T \geq B_1\alpha) \, d{G}_{n,r}(\theta) \, d{G}_{n,r}(\theta') \nonumber \\
&= \int_{[\alpha - \epsilon_1,\alpha + \epsilon_1]} \int_{[\alpha-\epsilon_1,\theta']} \mathbb{P}(T' < B_1\alpha)\mathbb{P}(T \geq B_1\alpha) \, d{G}_{n,r}(\theta) \, d{G}_{n,r}(\theta') + O(B_1^{-M}),
\end{align}
for all $M > 0$, as $B_1 \to \infty$.  By the Berry--Esseen theorem \citep[e.g.][]{Gnedenko:54}, there exists $C>0$ such that
\begin{align*}
& \sup_{|\theta-\alpha|< \epsilon_1} \biggl|\mathbb{P}(T < B_1\alpha) - \Phi\biggl( \frac{{B_1}^{1/2}(\alpha-\theta)}{\sigma}\biggr) \biggr|   \leq \frac{C}{B_1^{1/2}}.
\end{align*}
Hence, for $\theta' \in [\alpha - \epsilon_1,\alpha + \epsilon_1]$, and for large $B_1$,
\begin{equation}
\label{Eq:Step2}
\biggl|\int_{[\alpha-\epsilon_1,\theta']} \mathbb{P}(T \geq B_1\alpha) - \Phi\biggl( - \frac{{B_1}^{1/2}(\alpha-\theta)}{\sigma}\biggr)  \, d{G}_{n,r}(\theta)\biggr| \leq \frac{2C\log B_1}{B_1} \{g_{n,r}(\alpha) + 1\}.
\end{equation}
Now, by similar arguments to those bounding $R_1$ in the proof of Theorem~1, we have that, by a Taylor expansion about $\zeta = \alpha$, there exists $B_0 \in \mathbb{N}$, such that, for all $B_1 > B_0$ and all $\theta, \zeta \in [\alpha - \epsilon_1, \alpha + \epsilon_1]$,
\begin{align*}
& \biggl| \Phi\biggl(\frac{B_1^{1/2}(\alpha - \theta)}{\sqrt{\zeta(1-\zeta)}}\biggr) - \Phi\biggl(\frac{B_1^{1/2}(\alpha - \theta)}{\sqrt{\alpha(1-\alpha)}}\biggr) \biggr|   \leq |\zeta-\alpha| \frac{\log{B_1}}{2\sqrt{2\pi} \{\alpha(1-\alpha)\}^{3/2}}.
\end{align*}
Using this bound with $\zeta = \theta$, we deduce that, for all $B_1$ sufficiently large and for all  $\theta' \in [\alpha - \epsilon_1, \alpha + \epsilon_1]$,
\begin{align}
\label{Eq:Step3}
\biggl|\int_{[\alpha-\epsilon_1,\theta']} \Phi\biggl( - \frac{{B_1}^{1/2}(\alpha-\theta)}{\sigma}\biggr)  - & \Phi\biggl( - \frac{{B_1}^{1/2}(\alpha-\theta)}{\sqrt{\alpha(1-\alpha)}}\biggr)  \, d{G}_{n,r}(\theta) \biggr| 
\nonumber \\ & \leq \frac{\log{B_1}}{2\sqrt{2\pi} \{\alpha(1-\alpha)\}^{3/2}}\int_{[\alpha-\epsilon_1,\alpha + \epsilon_1]}|\theta-\alpha| \, dG_{n,r}(\theta)
\nonumber \\ & \leq \frac{\log^3{B_1}}{2\sqrt{2\pi} \{\alpha(1-\alpha)\}^{3/2} B_1} \{g_{n,r}(\alpha)+1\}.
\end{align} 
Moreover, using the fact that $G_{n,r}$ is continuously differentiable at $\alpha$, we have for large $B_1$ and uniformly for $\theta' \in [\alpha - \epsilon_1, \alpha + \epsilon_1]$ that
\begin{align}
\label{Eq:Step4}
\int_{\alpha-\epsilon_1}^{\theta'} &\Phi\biggl( - \frac{{B_1}^{1/2}(\alpha-\theta)}{\sqrt{\alpha(1-\alpha)}}\biggr)g_{n,r}(\theta) \, d\theta  
\nonumber \\ & = g_{n,r}(\alpha) \int_{\alpha-\epsilon_1}^{\theta'} \Phi\biggl( - \frac{{B_1}^{1/2}(\alpha-\theta)}{\sqrt{\alpha(1-\alpha)}}\biggr) \, d\theta + O(B_1^{-1} \log^2{B_1}) \nonumber \\ 
& = g_{n,r}(\alpha) \biggl\{(\theta'-\alpha) \Phi\biggl(\frac{-B_1^{1/2} (\alpha-\theta')}{\sqrt{\alpha(1-\alpha)}}\biggr)  + \frac{\sqrt{\alpha(1-\alpha)}}{B_1^{1/2}}\phi\biggl(\frac{-B_1^{1/2} (\alpha-\theta')}{\sqrt{\alpha(1-\alpha)}}\biggr) \biggr\} + O(B_1^{-1} \log^2{B_1}).
\end{align}
We deduce from~\eqref{Eq:Step1},~\eqref{Eq:Step2},~\eqref{Eq:Step3} and~\eqref{Eq:Step4} that
\begin{align*}
\int_{[0,1]} & \int_{[0,\theta']} \mathbb{P} (T' < B_1\alpha) \mathbb{P}(T \geq B_1\alpha) \, d{G}_{n,r}(\theta) \, d{G}_{n,r}(\theta') 
\\ & =  {g}_{n,r}^2(\alpha) \int_{\alpha-\epsilon_1}^{\alpha+\epsilon_1} \mathbb{P} (T' < B_1\alpha) \biggl\{(\theta' - \alpha) \Phi\biggl(\frac{-B_1^{1/2} (\alpha-\theta')}{\sqrt{\alpha(1-\alpha)}}\biggr) 
\\ & \hspace{180 pt} +  \frac{\sqrt{\alpha(1-\alpha)}}{B_1^{1/2}}\phi\biggl(\frac{-B_1^{1/2} (\alpha-\theta')}{\sqrt{\alpha(1-\alpha)}}\biggr) \biggr\}  \, d\theta' +  o(B_1^{-1})
\\ & = {g}_{n,r}^2(\alpha) \int_{\alpha-\epsilon_1}^{\alpha+\epsilon_1} \Phi\biggl(\frac{B_1^{1/2} (\alpha-\theta')}{\sqrt{\alpha(1-\alpha)}}\biggr)\biggl\{(\theta' - \alpha) \Phi\biggl(\frac{-B_1^{1/2} (\alpha-\theta')}{\sqrt{\alpha(1-\alpha)}}\biggr) 
\\ & \hspace{180 pt} +  \frac{\sqrt{\alpha(1-\alpha)}}{B_1^{1/2}}\phi\biggl(\frac{-B_1^{1/2} (\alpha-\theta')}{\sqrt{\alpha(1-\alpha)}}\biggr) \biggr\}  \, d\theta'  +  o(B_1^{-1}) 
\\ & = \frac{\alpha(1-\alpha)}{B_1} g_{n,r}^2(\alpha) \int_{-\infty}^{\infty} u \Phi(u)\Phi(-u) + \Phi(-u) \phi(u) \, du + o(B_1^{-1}) = \frac{\alpha(1-\alpha)}{2B_1} {g}_{n,r}^2(\alpha) + o(B_1^{-1}).
\end{align*}
We conclude that 
\begin{align*} 
\mathbf{Var}\{R({C}_n^{\mathrm{RP}})\} & \leq \biggl\{\pi_0\mathbf{Var}^{1/2}\Bigl(\int_{\mathbb{R}^p} \mathds{1}_{\{\nu_n(x) < \alpha\}} \, dP_0(x)\Bigr) + \pi_1 \mathbf{Var}^{1/2}\Bigl(\int_{\mathbb{R}^p} \mathds{1}_{\{\nu_n(x) < \alpha\}} \, dP_1(x)\Bigr)\biggr\}^2
\\ & = \Bigl\{ \pi_0\Bigl(\frac{\alpha(1-\alpha)}{B_1} {g}_{n,0}^2(\alpha)\Bigr)^{1/2}  + \pi_1\Bigl(\frac{\alpha(1-\alpha)}{B_1} {g}_{n,1}^2(\alpha)\Bigr)^{1/2} \Bigr\}^2 + o(B_1^{-1}) 
\\ & = \frac{\alpha(1-\alpha)}{B_1} \bar{g}_{n}^2(\alpha) + o(B_1^{-1}),
\end{align*} 
as required. 
\end{proof}

\section{Further discussion of assumptions}

In this section we investigate empirically assumptions~1~and~2 which, unlike assumption~3, depend on the configuration of the training data pairs.  

\subsection{Assumption~1}
\label{sec--A1} 
Assumption~1 asks that the distribution functions $G_{n,0}$ and $G_{n,1}$ are twice differentiable at $\alpha$.  Our interest here is in showing that the quantity $\gamma_n(\alpha)$, which appears in the conclusion of Theorem~1, is not too large, so that we can regard $\gamma_n(\alpha)/B_1$ as negligible for our recommended choice of $B_1 = 500$.

To this end, we approximate $\gamma_n(\hat{\alpha})$ over 100 training datasets generated from Model 2 as follows:  given a training dataset and a set of $B_1B_2$ projections with $B_1 \in \{50,51,\ldots,500\}$, $B_2 = 50$, we can estimate $R(C_{n}^{\mathrm{RP}})$ using an independent test set of size 1000.  Taking the average estimate over 10 independent sets of $B_1B_2$ projections (with the same training data), yields an estimate $\hat{R}_{B_1}$ of $\mathbf{E}\{R(C_{n}^{\mathrm{RP}})\}$ for the different values of $B_1$.  We then find the least squares estimator $(\hat{a},\hat{b})$ of $(a,b)$ in the model $\hat{R}_{B_1} \sim a + b/B_1$, so that $\hat{b}$ can be regarded as an approximation to $\gamma_n(\hat{\alpha})$.   We took $d=5$, and the cutoff $\hat{\alpha}$ was chosen via the method discussed in Section~5.1 for $B_1= 500$, and then kept fixed as we vary $B_1$. 

In Figure~\ref{Fig:A1}, we present histograms of these approximations to $\gamma_n(\hat{\alpha})$ over 100 realisations training data sets of size 50 for $p \in \{100, 1000\}$.  In this case, we see that $\gamma_n(\hat{\alpha})$ is bounded with high probability by $1$ when $p = 100$.  Therefore, the expansion in Theorem~1 gives us a test error approximation within $0.002$ of the test error of the infinite-simulation version of the RP ensemble classifier.  When $p = 1000$, we find that $\gamma_n(\hat{\alpha}) \leq 8$ with high probability; note that $8/500 < 0.02$.  For a few realisations of the training data, we have $\gamma_n(\hat{\alpha}) < 0$; indeed, while for the particular choice $\alpha = \alpha^*$, we have $\gamma_n(\alpha^*) \geq 0$, there is no reason for $\gamma_n(\alpha)$ to be non-negative for all values of $\alpha$.
 
\begin{figure}[h]
  \centering
\makebox{\includegraphics[width=\textwidth]{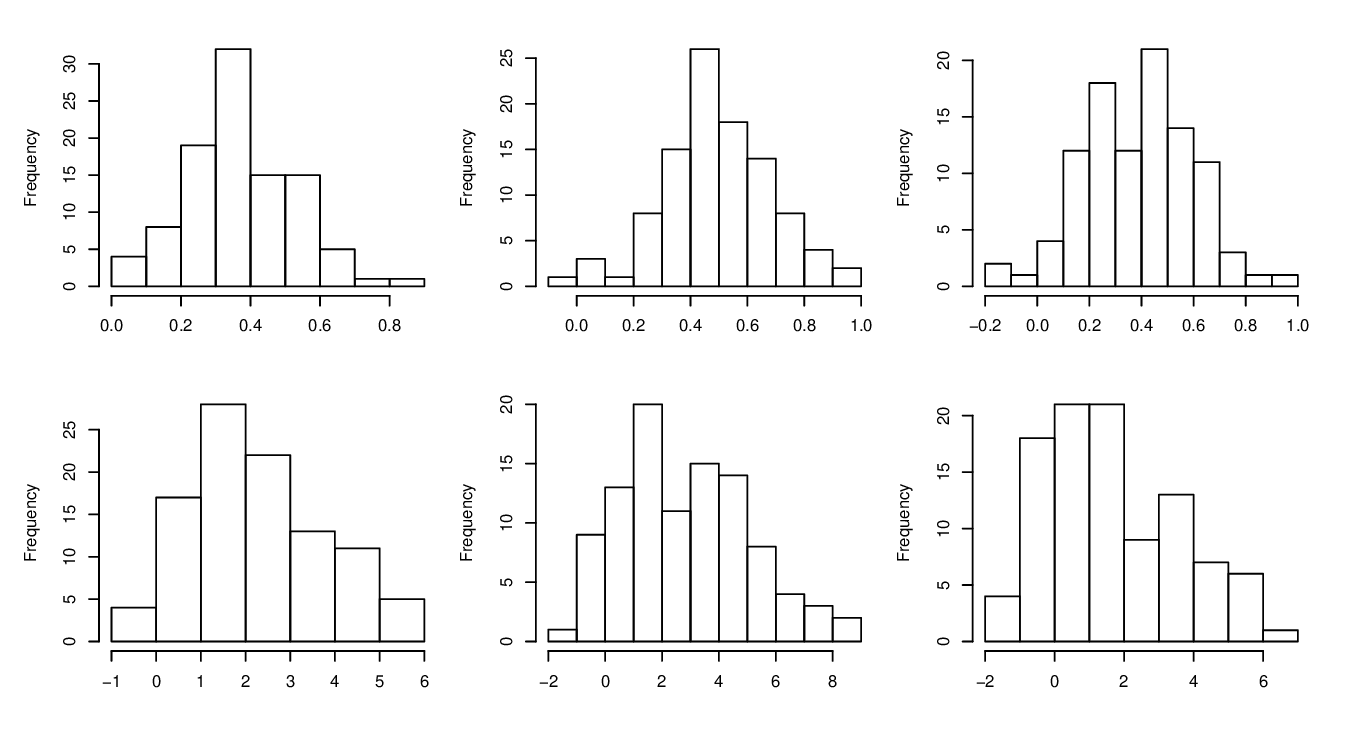}}
 \caption{\label{Fig:A1} Histograms of approximations to $\gamma_n(\hat{\alpha})$ constructed as described in Section~\ref{sec--A1} with $n = 50$, and $p = 100$ (top), $1000$ (bottom).  The base classifiers used were LDA (left), QDA (middle) and $k$nn (right).}
\end{figure}

\subsection{Assumption~2}
\label{sec--A2} 
Assumption~2 asks that there exists $\beta \in (0,1]$ such that $\mathbf{P}\bigl({R}_n^{\mathbf{A}_{1,1}} \leq {R}_n^* + |\epsilon_n| \bigr) \geq \beta$.  Note that if this probability were zero, then it would require both that the set of projections on which the minimum test error estimate is attained had zero Haar measure, and that the method of estimating test error is exactly unbiased as an estimate of the expected test error.

We can also estimate $\mathbf{P}\bigl({R}_n^{\mathbf{A}_{1,1}} \leq {R}_n^* + |\epsilon_n| \bigr)$ empirically as follows: 
\begin{enumerate}
\item First, compute the random projection ensemble classifier using $B_1B_2$ projections, but after selecting each projection, $\{\mathbf{A}_{{b_{1}}}:b_{1} = 1, \ldots, B_{1}\}$, we approximate $R({C}_{n}^{\mathbf{A}_{b_{1}}})$ using an independent validation set of size 1000.  Recalling that we also calculate ${R}_{n}^{\mathbf{A}_{b_{1}}}$ when selecting the projections, let 
\[
\hat{\epsilon}_{n} := \frac{1}{B_{1}} \sum_{b_{1} = 1}^{B_{1}}  \bigl\{R({C}_{n}^{\mathbf{A}_{b_{1}}}) - {R}_{n}^{\mathbf{A}_{b_{1}}}\bigr\}.
\]
\item Let $\hat{R}_n^* := \min_{b_1 ,b_2} R_{n}^{\mathbf{A}_{b_{1}, b_2}}$.
\item Generate a new set of $B_{3}$ independent projections $\mathbf{A}_{1}', \ldots, \mathbf{A}'_{B_{3}}$ having the same distribution as $\mathbf{A}_{1,1}$, and estimate $\mathbf{P}\bigl({R}_n^{\mathbf{A}_{1,1}} \leq {R}_n^* + |\epsilon_n| \bigr)$ by $B_3^{-1}\sum_{b_3=1}^{B_3} \mathds{1}_{\bigl\{{R}_n^{\mathbf{A}'_{b_{3}}} \leq \hat{R}_n^* + |\hat{\epsilon}_n|\bigr\}}$.
\end{enumerate} 

In Figure \ref{fig--A2histogram}, we present histograms of our estimates of $\mathbf{P}\bigl({R}_n^{\mathbf{A}_{1,1}} \leq {R}_n^* + |\epsilon_n| \bigr)$, constructed via the procedure above for data simulated from Model~1. The other parameters are $n = 50$, $d=5$, $B_1 = 500$, $B_2 = 50$ and $B_3 = 1000$. For the LDA base classifier, assumption~2 holds in each of the 100 repeats of the experiment with $\beta > 3/10$.  For the QDA and $k$nn base classifier, we find that $\beta > 0.05$ with high probability.  The results are stable for the different choices of dimension.

Noting that $(1-1/20)^{50} \approx 0.077$ while $(1-1/20)^{100} \approx 0.006$, the bound in Theorem~\ref{Thm:Main} suggests that we should choose $B_2$ to be slightly larger than 50 in this case.  Moreover, $|\epsilon_n|$ typically increases with $B_2$ as well, which gives larger values of $\beta$.  However, as discussed in Section~\ref{sec--B1B2choice}, we found that in practice choosing $B_2 = 50$ was sufficient, and little was gained by increasing it further.  
 
\begin{figure}
  \centering
\makebox{\includegraphics[width=\textwidth]{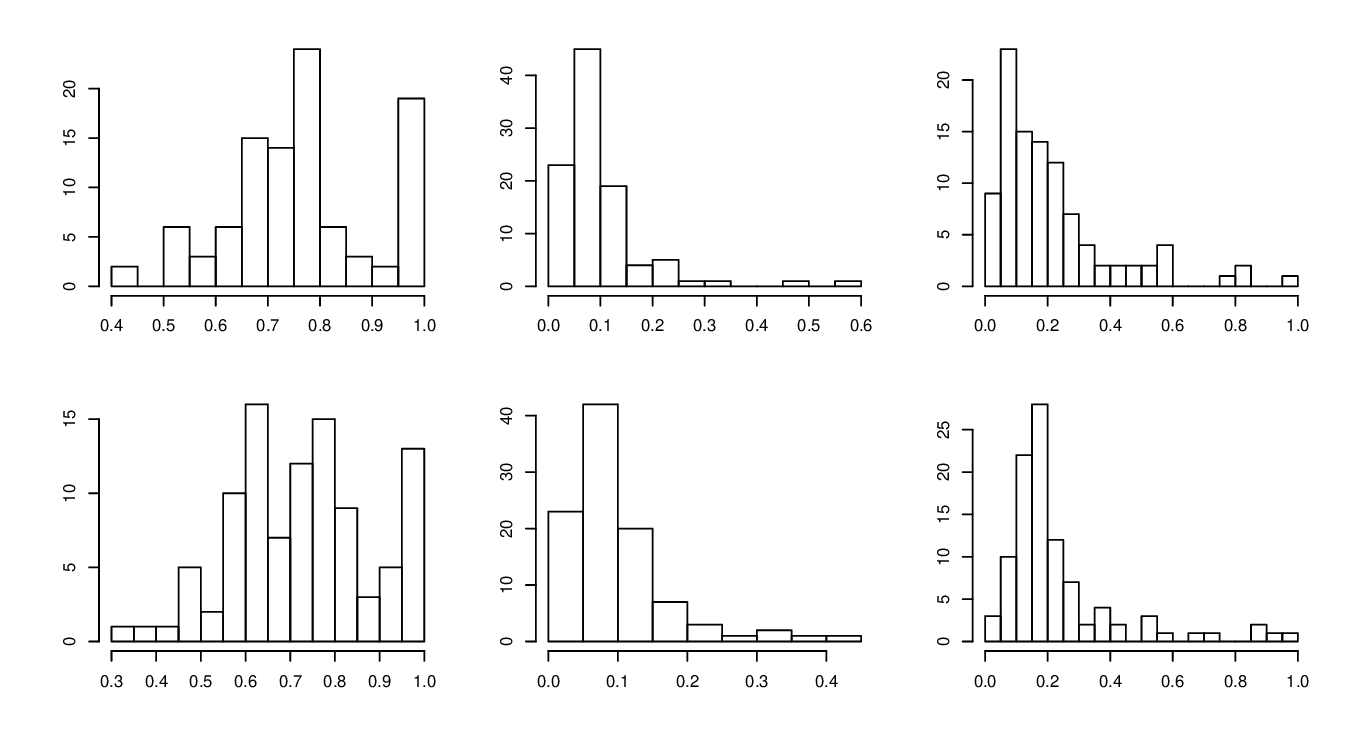}}
 \caption{\label{fig--A2histogram}Histograms of estimates of $\mathbf{P}\bigl({R}_n^{\mathbf{A}_{1,1}} \leq {R}_n^* + |\epsilon_n| \bigr)$ constructed as described in Section~\ref{sec--A2} over 100 simulated training datasets from Model~1, with $n = 50$, and $p = 100$ (top), $1000$ (bottom).  The base classifiers used were LDA (left), QDA (middle) and $k$nn (right).}
\end{figure} 

\section{Choice of $B_1$ and $B_2$} 
\label{sec--B1B2choice}
Here we elaborate on the discussion in Section~\ref{sec--B1B2} of the main text.   In Figures~\ref{fig--Model1.100.1000B1}~and~\ref{fig--Model1.100.1000B2}, we present the risk of the random projection ensemble classifier with LDA, QDA and $k$nn base classifiers for different values of $B_1$ and $B_2$ for data simulated from Model 2 (cf.\ Section~\ref{sec--rotated} of the main text).  We see that, as we expect from our theoretical results, the risk decreases as we increase $B_{1}$ and $B_{2}$.  In fact, increasing $B_{1}$ has a greater effect than increasing $B_{2}$, and it was this observation (which was also observed in other settings) that informed our recommendation of $B_1=500$ and $B_2=50$ as sensible default values.

\begin{figure}
  \centering
  \makebox{\includegraphics[width=\textwidth]{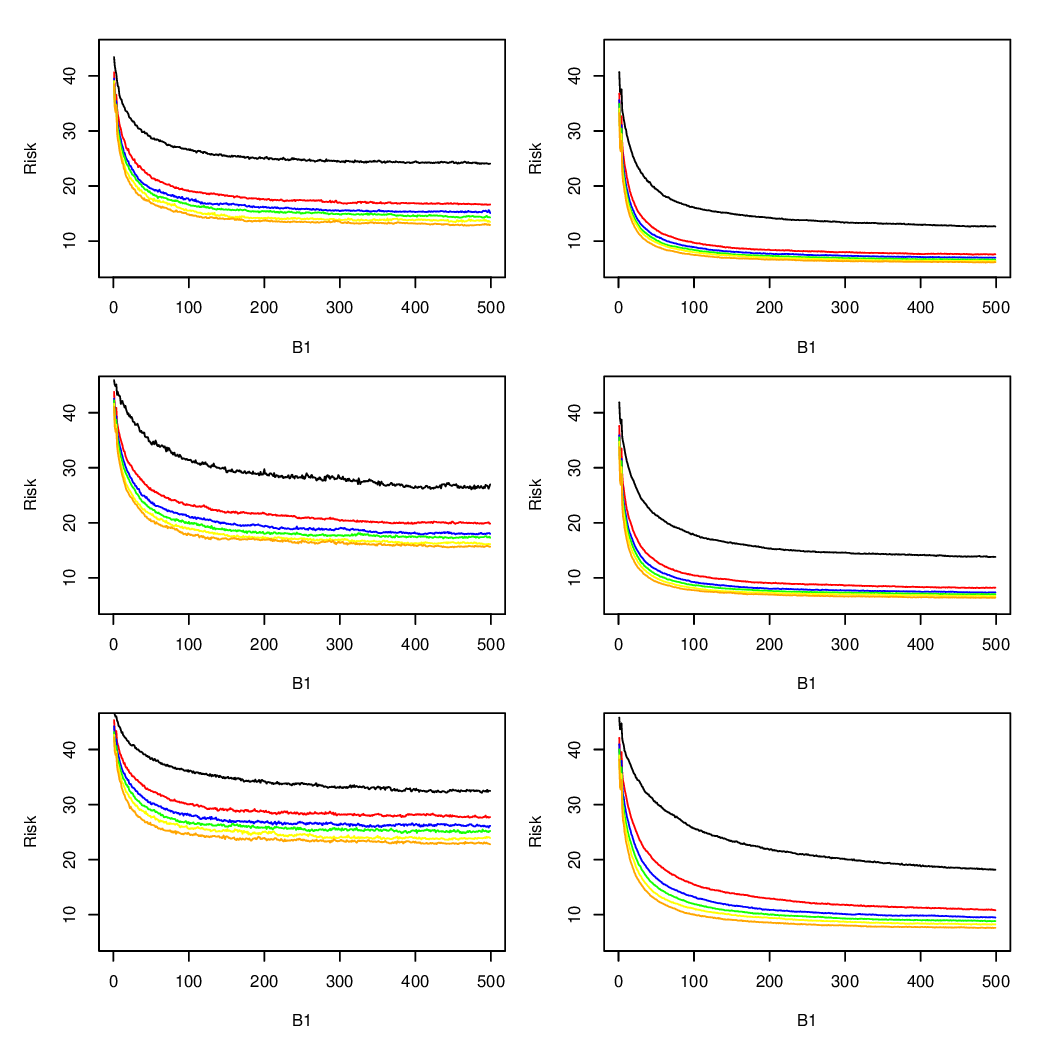}}
 \caption{\label{fig--Model1.100.1000B1}Risk estimates of the random projection ensemble classifier for Model~2 with $n = 100$ (left), $n=1000$ (right), $p=1000$, $\pi_1 = 0.5$.  The base classifier is LDA (top), QDA (middle) and $k$nn (bottom)  and we set $ d = 5$.  Here, $B_1$ varies from $2$ to $500$ on the $x$-axis, with $B_2=1$ (black), $10$ (red),  $25$ (blue), $50$ (green),  $100$ (yellow), $250$ (orange).  The risk was estimated as the average of 100 repeats of the experiment, with an independent test set of size $1000$ in each run.}
\end{figure} 

\begin{figure}
  \centering
  \makebox{\includegraphics[width=\textwidth]{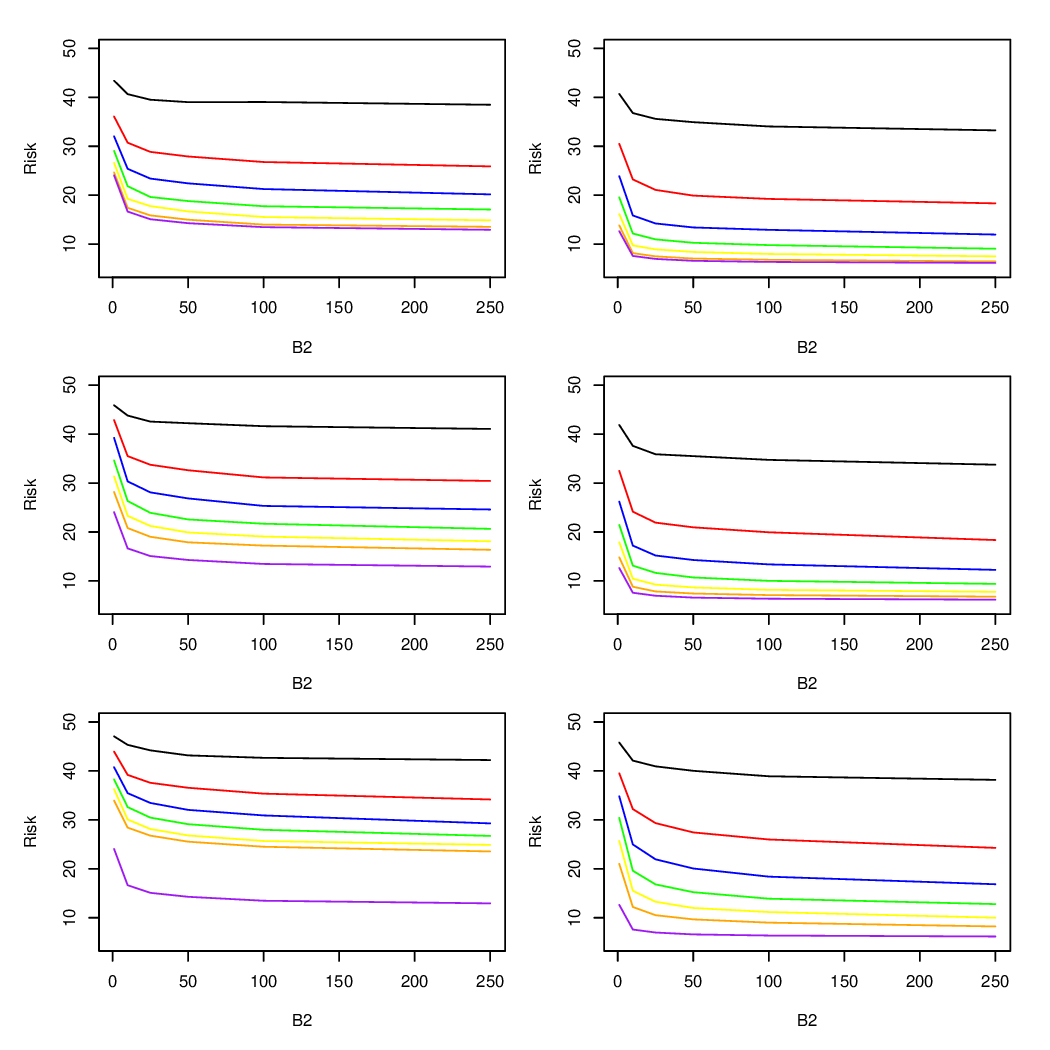}}
 \caption{\label{fig--Model1.100.1000B2}Risk estimates of the random projection ensemble classifier for Model~2 with $n = 100$ (left), $n=1000$ (right), $p=1000$, $\pi_1 = 0.5$.  The base classifier is LDA (top), QDA (middle) and $k$nn (bottom)  and we set $ d = 5$.  Here, $B_2$ varies from $1$ to $250$ on the $x$-axis, with $B_1=1$ (black), $10$ (red),  $25$ (blue), $50$ (green),  $100$ (yellow), $250$ (orange) and $500$ (purple).  The risk was estimated as the average of 100 repeats of the experiment, with an independent test set of size $1000$ in each run.}
\end{figure} 

\section{Further simulation results} 
\label{Sec:FurtherSim}

Here we present the remaining simulation results: the model numbers refer to those in Section~6.1 of the main text, with the covariance matrix $\Sigma$ in Model~1 when $p=1000$ given by $\mathrm{diag}(1,1,1/16,\ldots,1/16)$.

Tables~\ref{tab--rotateddist100.66} and~\ref{tab--indep100.66} examine cases of unbalanced priors, where $ \pi_{1} = 0.66$.  In Tables~\ref{tab--rotateddist1000.5} and \ref{tab--indep1000.5}, we present the results when the dimension of the feature space is $p = 1000$ and where we project into either $d=5$ or $d=10$ dimensions.   The results are consistent with the conclusions drawn in Section~\ref{Sec:Conc} of the main text.

\begin{table}
%6:10new
  \caption{  \label{tab--rotateddist100.66}Misclassification rates for Models 1 and 2, with  $p =100$ and $\pi_{1} = 0.66$.}
  \centering
\fbox{%
  \begin{tabular}{l |c c c | c c c }
       &\multicolumn{3}{c|}{Model 1, Bayes risk = 4.15} &\multicolumn{3}{c}{ Model 2,  Bayes risk = 3.67 }  \\
  \multicolumn{1}{r|}{$n$}  &$50$            &$200$      & $1000$ &$50$            &$200$      & $1000$\\
    \hline
    RP-LDA$_{2}$&$ 46.91 _{ 0.71 }$&$ 34.61 _{ 0.31 }$&$ 34.05 _{ 0.22 }$&$ 7.49 _{ 0.29 }$&$\mathbf{5.03} _{ 0.11 }$&$\mathbf{4.69} _{ 0.10 }$\\
RP-LDA$_{5}$&$ 47.03 _{ 0.65 }$&$ 40.05 _{ 0.61 }$&$ 34.05 _{ 0.22 }$&$\mathbf{6.92}_{ 0.20 }$&$\mathbf{5.03} _{ 0.11 }$&$ \mathbf{4.67} _{ 0.10 }$\\
RP-QDA$_{2}$&$ 40.49 _{ 0.67 }$&$ 33.54 _{ 0.31 }$&$ 32.00 _{ 0.58 }$&$ 7.16 _{ 0.25 }$&$ \mathbf{4.96} _{ 0.11 }$&$ 4.73 _{ 0.10 }$\\
RP-QDA$_{5}$&$ 36.14 _{ 0.55 }$&$\mathbf{22.52} _{ 0.40 }$&$\mathbf{12.17}_{ 0.17 }$&$ 7.57 _{ 0.25 }$&$ 5.18 _{ 0.13 }$&$\mathbf{ 4.61} _{ 0.10 }$\\
RP-$k$nn$_{2}$&$ 39.53 _{ 0.65 }$&$ 33.02 _{ 0.42 }$&$ 32.03 _{ 0.52 }$&$ 7.86 _{ 0.31 }$&$ 5.18 _{ 0.11 }$&$ 4.76 _{ 0.10 }$\\
RP-$k$nn$_{5}$&$ 40.79 _{ 0.62 }$&$ 28.21 _{ 0.56 }$&$\mathbf{12.30} _{ 0.33 }$&$ 7.76 _{ 0.26 }$&$ 5.37 _{ 0.16 }$&$\mathbf{ 4.69} _{ 0.10 }$\\
LDA          &N/A                    &$ 45.74 _{ 0.26 }$&$ 39.26 _{ 0.24 }$&N/A                     &$12.71 _{ 0.21 }$&$ 5.50 _{ 0.11 }$\\
QDA          &N/A&N/A                                   &$ 26.88 _{ 0.23 }$&N/A             &N/A            &$18.08 _{ 0.25 }$\\
$k$nn        &$\mathbf{30.57} _{ 0.32 }$&$ 23.23 _{ 0.26 }$&$ 17.10 _{ 0.23 }$&$12.95 _{ 0.29 }$&$ 8.84 _{ 0.16 }$&$ 6.91 _{ 0.13 }$\\
RF           &$ 36.06 _{ 0.43 }$&$ 34.37 _{ 0.21 }$&$ 34.03 _{ 0.22 }$&$12.55 _{ 0.42 }$&$ 7.61 _{ 0.16 }$&$ 6.04 _{ 0.11 }$\\
Radial SVM   &$ 34.86 _{ 0.43 }$&$ 34.11 _{ 0.20 }$&$ 34.05 _{ 0.22 }$&$30.52 _{ 0.76 }$&$12.24 _{ 0.42 }$&$ 6.01 _{ 0.12 }$\\
Linear SVM   &$ 44.66 _{ 0.39 }$&$ 44.48 _{ 0.25 }$&$ 34.05 _{ 0.22 }$&$ 8.33 _{ 0.20 }$&$ 7.89 _{ 0.16 }$&$ 7.07 _{ 0.13 }$\\
Radial GP    &$ 35.10 _{ 0.38 }$&$ 33.99 _{ 0.21 }$&$ 31.44 _{ 0.22 }$&$19.76 _{ 0.74 }$&$ 8.29 _{ 0.19 }$&$ 5.03 _{ 0.10 }$\\
PenLDA       &$ 42.66 _{ 0.41 }$&$ 41.83 _{ 0.28 }$&$ 38.02 _{ 0.23 }$&$10.07 _{ 0.54 }$&$ 5.88 _{ 0.17 }$&$ 5.17 _{ 0.12 }$\\
NSC          &$ 36.56 _{ 0.42 }$&$ 34.42 _{ 0.21 }$&$ 34.05 _{ 0.22 }$&$12.01 _{ 0.64 }$&$ 6.44 _{ 0.22 }$&$ 5.25 _{ 0.12 }$\\
PenLog       &$ 36.28 _{ 0.53 }$&$ 34.77 _{ 0.26 }$&$ 35.05 _{ 0.23 }$&$10.60 _{ 0.22 }$&$ 7.05 _{ 0.14 }$&$ 5.42 _{ 0.10 }$\\
SDR5-LDA&N/A     &$ 37.80 _{ 0.48 }$&$ 35.31 _{ 0.30 }$&N/A                &$13.37 _{ 0.23 }$&$ 5.66 _{ 0.10 }$\\
SDR5-$k$nn&N/A      &$ 32.22 _{ 0.71 }$&$ 21.83 _{ 1.08 }$&N/A              &$16.57 _{ 0.26 }$&$ 7.12 _{ 0.12 }$\\
OTE          &$ 39.60 _{ 0.51 }$&$ 31.92 _{ 0.70 }$&$ 15.52 _{ 1.03 }$&$17.88 _{ 0.44 }$&$11.62 _{ 0.24 }$&$ 8.56 _{ 0.15 }$\\
ES$k$nn      &$ 43.62 _{ 0.47 }$&$ 43.51 _{ 0.47 }$&$ 44.04 _{ 0.46 }$&$35.52 _{ 0.59 }$&$33.31 _{ 0.62 }$&$32.30 _{ 0.55 }$\\
 \end{tabular}}
\end{table}

\begin{table}
  \caption{  \label{tab--indep100.66} Misclassification rates for Model 3 and 4, with $p = 100$ and $\pi_1 = 0.66$.}
  \centering
\fbox{%
  \begin{tabular}{l|c c c | c c c  }
\hline %26:30new
  &\multicolumn{3}{c|}{Model 3, Bayes risk = 1.05}   &\multicolumn{3}{c}{Model 4, Bayes risk = 10.93}  \\ 
  \multicolumn{1}{r|}{$n$}  &$50$       &$200$       & $1000$ &$50$       &$200$       & $1000$ \\
    \hline
RP-LDA$_{2}$  &$ 53.48 _{ 1.59 }$&$ 34.10 _{ 0.21 }$&$ 34.02 _{ 0.22 }$ &$ 40.05 _{ 0.64 }$&$ 34.72 _{ 0.27 }$&$ 34.08 _{ 0.23 }$\\
RP-LDA$_{5}$  &$ 59.14 _{ 0.81 }$&$ 37.27 _{ 0.89 }$&$ 34.02 _{ 0.22 }$&$ 44.28 _{ 0.52 }$&$ 37.75 _{ 0.38 }$&$ 34.20 _{ 0.23 }$\\
RP-QDA$_{2}$  &$ 15.76 _{ 0.91 }$&$ 13.83 _{ 1.16 }$&$ 23.96 _{ 1.38 }$&$ 38.80 _{ 0.80 }$&$ 44.10 _{ 0.65 }$&$ 46.66 _{ 0.49 }$\\
RP-QDA$_{5}$&$ \mathbf{13.75} _{ 0.88 }$&$  \mathbf{4.65} _{ 0.13 }$&$ \mathbf{3.68} _{ 0.09 }$&$ 39.74 _{ 0.60 }$&$ 47.02 _{ 0.49 }$&$ 50.00 _{ 0.31 }$\\
RP-$k$nn$_{2}$&$ 21.80 _{ 1.14 }$&$  9.01 _{ 0.33 }$&$  7.35 _{ 0.41 }$&$ 29.55 _{ 0.81 }$&$ 20.84 _{ 0.37 }$&$ 16.61 _{ 0.17 }$\\
RP-$k$nn$_{5}$&$ 19.31 _{ 0.43 }$&$  6.18 _{ 0.14 }$&$  4.10 _{ 0.10 }$&$ 27.60 _{ 0.71 }$&$ 18.67 _{ 0.29 }$&$ 15.80_{ 0.16 }$\\
LDA          &N/A                          &$ 46.29 _{ 0.24 }$&$ 39.62 _{ 0.22 }$  &N/A                  &$ 34.28 _{ 0.30 }$&$ 30.82 _{ 0.26 }$\\
QDA          &N/A&N/A                 &$ 21.42 _{ 0.24 }$&N/A&N/A                                  &$ 47.47 _{ 0.22 }$\\
$k$nn        &$ 65.23 _{ 0.20 }$&$ 65.53 _{ 0.22 }$&$ 65.70 _{ 0.21 }$   &$ 30.11 _{ 0.53 }$&$ 24.48 _{ 0.32 }$&$ 21.91 _{ 0.21 }$\\
RF           &$ 36.29 _{ 0.39 }$&$ 32.51 _{ 0.27 }$&$ 16.78 _{ 0.33 }$          &$\mathbf{26.06}_{ 0.57 }$&$\mathbf{17.29} _{ 0.22 }$&$15.82_{ 0.17 }$\\
Radial SVM   &$ 32.92 _{ 0.56 }$&$ 22.83 _{ 1.02 }$&$\mathbf{3.70} _{ 0.09 }$&$ 34.62 _{ 0.46 }$&$ 33.91 _{ 0.22 }$&$ 34.05 _{ 0.23 }$\\
Linear SVM   &$ 47.79 _{ 0.28 }$&$ 45.37 _{ 0.25 }$&$ 34.02 _{ 0.22 }$   &$ 32.95 _{ 0.43 }$&$ 31.33 _{ 0.38 }$&$ 34.56 _{ 0.25 }$\\
Radial GP    &$ 39.72 _{ 0.90 }$&$ 37.37 _{ 0.32 }$&N/A   &$ 31.20 _{ 0.62 }$&$ 23.00 _{ 0.30 }$&$ 18.71 _{ 0.17 }$\\
PenLDA       &$ 46.21 _{ 0.32 }$&$ 42.98 _{ 0.25 }$&$ 38.84 _{ 0.21 }$  &$ 50.03 _{ 1.19 }$&$ 35.92 _{ 0.75 }$&$ 34.96 _{ 0.23 }$\\
NSC          &$ 37.58 _{ 0.53 }$&$ 34.47 _{ 0.23 }$&$ 34.05 _{ 0.22 }$ &$ 35.84 _{ 0.47 }$&$ 34.28 _{ 0.23 }$&$ 34.15 _{ 0.23 }$\\
PenLog       &$ 37.52 _{ 0.60 }$&$ 38.14 _{ 0.33 }$&$ 38.03 _{ 0.23 }$&N/A&N/A&N/A\\
SDR5-LDA     &N/A&$ 46.23 _{ 0.26 }$&$ 39.56 _{ 0.21 }$&N/A          &$ 34.89 _{ 0.30 }$&$ 30.79 _{ 0.26 }$\\
SDR5-$k$nn   &N/A&$ 46.23 _{ 0.29 }$&$ 34.15 _{ 0.23 }$&N/A                    &$ 37.98 _{ 0.34 }$&$ 23.91 _{ 0.24 }$\\
OTE          &$ 41.42 _{ 0.46 }$&$ 34.63 _{ 0.26 }$&$ 13.43 _{ 0.20 }$ &$ 28.31 _{ 0.61 }$&$ 18.00 _{ 0.25 }$&$\mathbf{15.43}_{ 0.17 }$\\
ES$k$nn      &$ 44.79 _{ 0.45 }$&$ 42.74 _{ 0.27 }$&$ 41.12 _{ 0.24 }$ &$ 42.26 _{ 0.57 }$&$ 41.22 _{ 0.44 }$&$ 39.54 _{ 0.46 }$\\
 \end{tabular}}
\end{table}

\begin{table}
 %11:15new
  \caption{  \label{tab--rotateddist1000.5} Misclassification rates for Models 1 and 2, with $ p =1000$ and $\pi_{1}  = 0.5$.}
  \centering
\fbox{%
  \begin{tabular}{l|c c c | c c c   }
       &\multicolumn{3}{c|}{Model 1, Bayes risk = 4.45 }     &\multicolumn{3}{c|}{Model 2, Bayes risk = 4.09 }  \\
  \multicolumn{1}{r|}{$n$}  &$50$    &$200$          & $1000$  &$50$    &$200$          & $1000$\\
    \hline
RP-LDA$_{5}$&$ 51.84 _{ 0.52 }$&$ 51.23 _{ 0.65 }$&$ 40.17 _{ 0.68 }$&$ 21.34 _{ 0.77 }$&$ 10.20 _{ 0.41 }$&$  6.66 _{ 0.17 }$\\
RP-LDA$_{10}$&$ 51.21 _{ 0.52 }$&$ 52.02 _{ 0.40 }$&$ 44.81 _{ 0.59 }$&$ 20.38 _{ 0.76 }$&$\mathbf{8.31} _{ 0.27 }$&$\mathbf{5.82}_{0.11} $\\
RP-QDA$_{5}$&$ 26.02 _{ 0.67 }$&$ 9.81 _{ 0.25 }$&$ 5.94 _{ 0.12 }$&$ 27.51 _{ 0.79 }$&$ 11.59 _{ 0.42 }$&$  7.04 _{ 0.17 }$\\
RP-QDA$_{10}$&$ 23.64 _{ 0.81 }$&$ 8.00 _{ 0.17 }$&$ \mathbf{5.44} _{ 0.10 }$&$ 26.23 _{ 0.78 }$&$9.62 _{ 0.30 }$&$ 6.17 _{0.12}$\\
RP-$k$nn$_{5}$&$ 33.91 _{ 0.81 }$&$ 11.10 _{ 0.33 }$&$ 6.34 _{ 0.12 }$&$ 32.47 _{ 0.75 }$&$ 17.29 _{ 0.63 }$&$  8.72 _{ 0.24 }$\\
RP-$k$nn$_{10}$&$ 27.01 _{ 0.86 }$&$ 8.45 _{ 0.24 }$&$\mathbf{ 5.52} _{ 0.10 }$&$31.68_{ 0.75 }$&$16.95 _{0.96}$&$8.34_{ 0.19 } $\\
$k$nn        &$ \mathbf{9.97 }_{ 0.26 }$&$ \mathbf{6.61} _{ 0.14 }$&$ 5.75 _{ 0.11 }$&$ 31.04 _{ 0.37 }$&$ 21.19 _{ 0.24 }$&$ 14.66 _{ 0.17 }$\\
RF           &$ 49.99 _{ 0.22 }$&$ 49.98 _{ 0.23 }$&$ 49.70 _{ 0.23 }$&$ 34.84 _{ 0.87 }$&$ 16.79 _{ 0.56 }$&$  9.17 _{ 0.16 }$\\
Radial SVM   &$ 49.94 _{ 0.22 }$&$ 50.01 _{ 0.23 }$&$ 50.08 _{ 0.22 }$&$ 48.51 _{ 0.58 }$&$ 36.03 _{ 1.00 }$&$  6.75 _{ 0.13 }$\\
Linear SVM   &$ 50.04 _{ 0.23 }$&$ 50.19 _{ 0.23 }$&$ 50.08 _{ 0.21 }$&$ \mathbf{13.16} _{ 0.39 }$&$  8.67 _{ 0.13 }$&$  8.64 _{ 0.13 }$\\
Radial GP    &$ 45.69 _{ 0.63 }$&$ 31.79 _{ 0.68 }$&$ 9.62 _{ 0.20 }$&$ 41.81 _{ 0.80 }$&$ 17.43 _{ 0.63 }$&$  6.13 _{ 0.12 }$\\
PenLDA       &$ 49.85 _{ 0.22 }$&$ 49.94 _{ 0.23 }$&$ 49.94 _{ 0.22 }$&$ 33.09 _{ 1.14 }$&$ 21.21 _{ 0.96 }$&$ 14.00 _{ 0.62 }$\\
NSC          &$ 49.78 _{ 0.23 }$&$ 49.97 _{ 0.24 }$&$ 49.55 _{ 0.26 }$&$ 37.58 _{ 0.96 }$&$ 25.58 _{ 0.92 }$&$ 16.64 _{ 0.67 }$\\
PenLog       &$ 49.32 _{ 0.39 }$&$ 48.64 _{ 0.56 }$&$ 49.96 _{ 0.33 }$&$ 29.96 _{ 0.61 }$&$ 13.70 _{ 0.19 }$&$  8.54 _{ 0.12 }$\\
OTE          &$ 49.88 _{ 0.22 }$&$ 48.77 _{ 0.35 }$&$ 42.37 _{ 0.92 }$ &$ 35.02 _{ 0.50 }$&$ 22.23 _{ 0.32 }$&$ 17.96 _{ 0.26 }$\\
ES$k$nn      &$ 49.93 _{ 0.21 }$&$ 50.20 _{ 0.23 }$&$ 49.78 _{ 0.23 }$&$ 48.59 _{ 0.29 }$&$ 48.27 _{ 0.28 }$&$ 47.99 _{ 0.28 }$\\

  \end{tabular}}
\end{table}

\begin{table}
 %31:35new
  \caption{  \label{tab--indep1000.5}Misclassification rates for Models 3 and 4, with $p = 1000$ and $\pi_1= 0.5$.}
  \centering
\fbox{%
  \begin{tabular}{l|c c c | c c c }
       &\multicolumn{3}{c|}{Model 3, Bayes risk = 0.00}        &\multicolumn{3}{c}{Model 4, Bayes risk = 6.96 }  \\
         \multicolumn{1}{r|}{$n$}  &$50$        &$200$       & $1000$ &$50$        &$200$       & $1000$ \\
    \hline
RP-LDA$_{5}$  &$ 44.36 _{ 1.30 }$&$ 45.52 _{ 1.55 }$&$ 42.12 _{ 1.64 }$&$ 44.74 _{ 0.47 }$&$ 44.33 _{ 0.67 }$&$ 44.09 _{ 0.75 }$\\
RP-LDA$_{10}$&$ 41.55 _{ 1.05 }$&$45.20 _{1.00}$&$44.61 _{1.25}$&$44.12 _{0.42}$&$40.87_{0.55}$ &$42.00_{0.79}$ \\
RP-QDA$_{5}$  &$\mathbf{1.04}_{ 0.35 }$&$  0.02 _{ 0.01 }$&$\mathbf{0.00}_{ 0.00 }$ &$ 43.60 _{ 0.35 }$&$ 42.21 _{ 0.23 }$&$ 42.50 _{ 0.24 }$\\
RP-QDA$_{10}$&$10.07_ {1.43 }$& $\mathbf{0.01}_{0.004}$&$ \mathbf{0.00}_{0.00} $&$44.17_{0.37}$&$42.09_{ 0.21} $&$42.44 _{ 0.24 }$\\
RP-$k$nn$_{5}$&$ 20.63 _{ 1.81 }$&$  0.73 _{ 0.14 }$&$\mathbf{0.00} _{ 0.00 }$&$ 42.83 _{ 0.53 }$&$ 38.49 _{ 0.43 }$&$ 34.54 _{ 0.32 }$\\
RP-$k$nn$_{10}$&$31.23_{ 1.40 }$& $0.81 _{0.15} $ &$0.05_{0.03}$ &$44.08 _{0.58}$ & $39.20 _{0.61} $& $36.84_{ 0.41}$\\
$k$nn         &$ 50.15 _{ 0.23 }$&$ 50.09 _{ 0.22 }$&$ 49.85 _{ 0.23 }$&$ 48.36 _{ 0.37 }$&$ 47.25 _{ 0.52 }$&$ 48.68 _{ 0.47 }$\\
RF            &$ 47.55 _{ 0.32 }$&$ 18.11 _{ 0.77 }$&$  0.10 _{ 0.02 }$      &$ 40.74 _{ 0.56 }$&$ 27.02 _{ 0.44 }$&$ 21.42 _{ 0.20 }$\\
Radial SVM    &$ 42.03 _{ 1.61 }$&$\mathbf{0.00} _{ 0.00 }$&$\mathbf{0.00} _{ 0.00 }$ &$ 48.37 _{ 0.37 }$&$ 47.22 _{ 0.37 }$&$ 46.74 _{ 0.33 }$\\
Linear SVM    &$ 46.84 _{ 0.26 }$&$ 45.51 _{ 0.23 }$&$ 46.11 _{ 0.23 }$  &$ 41.82 _{ 0.42 }$&$ 35.93 _{ 0.35 }$&$ 34.09 _{ 0.38 }$\\
Radial GP     &$ 49.88 _{ 0.24 }$&$ 49.91 _{ 0.23 }$&$ 49.05 _{ 0.24 }$&$ 45.75 _{ 0.45 }$&$ 40.33 _{ 0.71 }$&$ 30.05 _{ 0.46 }$\\
PenLDA        &$ 47.95 _{ 0.29 }$&$ 47.93 _{ 0.31 }$&$ 46.79 _{ 0.29 }$  &$ 49.72 _{ 0.32 }$&$ 49.45 _{ 0.29 }$&$ 49.23 _{ 0.26 }$\\
NSC           &$ 47.74 _{ 0.30 }$&$ 46.57 _{ 0.28 }$&$ 45.55 _{ 0.29 }$  &$ 49.35 _{ 0.32 }$&$ 48.51 _{ 0.41 }$&$ 48.82 _{ 0.32 }$\\
PenLog        &$ 48.84 _{ 0.29 }$&$ 47.39 _{ 0.25 }$&$ 45.91 _{ 0.24 }$&N/A&N/A&N/A\\
OTE            &$ 47.85 _{ 0.25 }$&$ 30.79 _{ 0.33 }$&$ 4.46 _{ 0.16 }$&$\mathbf{37.66}_{ 0.63 }$&$\mathbf{23.85}_{ 0.27 }$&$\mathbf{20.86}_{ 0.20 }$\\
ES$k$nn        &$ 49.00 _{ 0.28 }$&$ 46.62 _{ 0.25 }$&$ 45.12 _{ 0.23 }$&$ 48.63 _{ 0.32 }$&$ 48.12 _{ 0.30 }$&$ 47.52 _{ 0.30 }$\\
  \end{tabular}}
\end{table}

\section{Computational timings}
In Section~\ref{sec--compcomp} of the main text, we discuss the computational complexity of the random projection ensemble classifier and the scope for incorporating parallel computing.  We plot the average run times for our classifier as we vary the number of processors used in Figure~\ref{fig--Times}, which reveals that it is possible to significantly speed up the procedure using multiple CPUs. 

\begin{figure}
  \centering
  \makebox{\includegraphics[width=\textwidth]{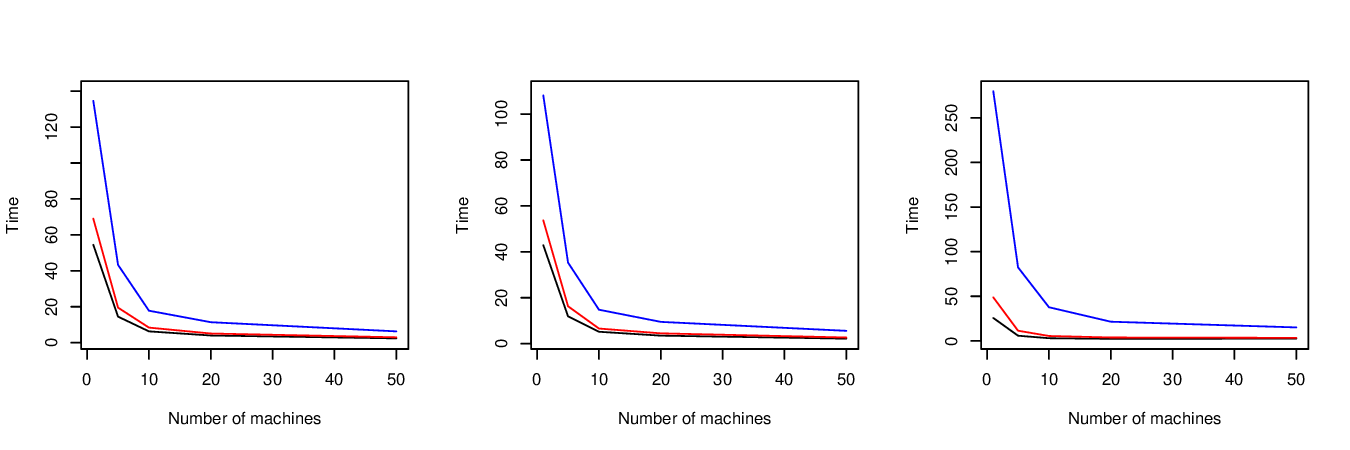}}
 \caption{\label{fig--Times} Average elapsed times (in seconds) for the random projection ensemble classifier with LDA (left), QDA (middle) and $k$nn (right) for Model~2, with $p = 100$, $n = 50$ (black), $200$ (red) and $1000$ (blue).  The other parameters are $d = 5$, $B_1 = 500$ and $B_2 = 50$.}
\end{figure} 

In Table~\ref{tab:computation}, we compare the running time of the random projection ensemble classifier with the alternative methods considered with data from Model~2.  For the random projection ensemble classifiers, we used $B_1 = 500$ and $B_2 = 50$, and present the elapsed time running parallelised code on 50 Intel$^\textsuperscript{\textregistered}$ Core\texttrademark \ i5-4690 CPU @ 3.50GHz [4 CPU] machines with $n \in \{50,200,1000\}$ and $p \in \{100,1000\}$.  For the comparators, we present the elapsed time on one such machine, though in some cases, these times could also be reduced through parallel computing.  The code is written in \texttt{R} and makes use of the \texttt{parallel} package\footnote{see \url{https://stat.ethz.ch/R-manual/R-devel/library/parallel/doc/parallel.pdf}}.  Our \texttt{R} package \texttt{RPEnsemble} \citep{CanningsSamworth:16b} makes it straightforward for the user to set up a computing cluster on which to run the RP procedure in parallel.

Our method is typically slower for the small sample size and dimension, but it scales well with both $n$ and $p$.  For instance, when $n$ increases from 50 to 1000 with $p=100$ for the LDA and QDA classifiers, the computational cost only increases by a factor of about 3.  Similarly, when $n = 200$, and we increase $p$ from 100 to 1000, the run time increases by a factor of less than 5. 

\begin{table}
  \caption{  \label{tab:computation} Average elapsed time (in seconds) for Model 2.}
  \centering
\fbox{%
  \begin{tabular}{l|c c c | c c c }
 &\multicolumn{3}{c|}{$p = 100$}&\multicolumn{3}{c}{$p = 1000$}\\
  \multicolumn{1}{r|}{$n$}    &$50$&$200$ & $1000$ &$50$&$200$ & $1000$\\
    \hline
RP-LDA$_{5}$   & 2.4 & 3.0 & 6.3 & 9.7 & 12 & 23 \\
RP-QDA$_{5}$    & 2.2 & 2.7 & 5.6 & 9.6 & 11 & 22 \\
RP-$k$nn$_{5}$ & 2.6 & 3.4 &15 &10 & 13 & 43 \\ 
LDA                   & N/A & 0.01 & 0.02 & N/A & N/A & N/A  \\
QDA                   & N/A & N/A & 0.03 & N/A & N/A & N/A  \\
$k$nn                 & 0.01 & 0.04 & 0.51 & 0.09 & 0.55 &10 \\
RF                    & 0.16 & 0.60 &3.8 & 0.91 &4.9 &37 \\
Radial SVM            & 0.02 & 0.04 & 0.23 & 0.15 & 0.42 & 3.8 \\
Linear SVM            & 0.01 & 0.02 & 0.10 & 0.14 & 0.33 & 1.7 \\
Radial GP             & 0.04 & 0.10 &4.8 & 0.07 & 0.15 & 8.0 \\
PenLDA                & 0.06 & 0.09 & 0.23 & 0.27 & 0.44 & 1.8 \\
NSC                   & 0.13 & 0.15 & 0.26 & 0.20 & 0.33 & 1.0 \\
PenLog                & 0.78 &1.4 &3.2 &3.9 &21&130\\ 
SDR5-LDA              & N/A & 0.13 &1.1 & N/A & N/A & N/A  \\
SDR5-$k$nn            & N/A & 0.14 &1.1 & N/A & N/A & N/A  \\
OTE                   &1.4 & 3.0&15&4.0 &17&140\\
ES$k$nn               & 0.11 &  0.11&  0.13& 0.14 & 0.15 & 0.21 \\
  \end{tabular}}
\end{table}

\clearpage

{\Large \textbf{Rejoinder to `Random-projection ensemble classification'}}

\medskip

We are very grateful to the discussants for their insightful comments on our work, and are glad to find a broad consensus that methods based on random projections offer considerable promise for high-dimensional data analysis.  The comments are extremely wide-ranging, and we apologise in advance for the fact that, due to space limitations, we are unable to address all of them.  It is clear, however, that there is considerable scope for future research in this area, and we look forward to witnessing and contributing to its development.

\medskip

\emph{Correlation between features}
 
Kent presents an interesting toy example, which focuses on the effect of the correlation between the features.  As we discuss in Figure~1 of the main text, it is usually only sensible to aggregate over carefully selected (rather than all) projections.  Even in Kent's high correlation case ($\rho = 0.99$), where only $5\%$ of projections result in a base classifier with at least half the discriminatory power, we still expect with $B_2=50$ to find such a projection in most blocks.  We carried out a small simulation study on Gaussian class-conditional distributions with $\pi_0=\pi_1=1/2$:
\begin{itemize} 
\item Case 1a: $p =2$, $\rho = 0, \mu_{1} = a_{1}(1,0)^{T}, \mu_{0} = a_{1}(-1,0)^{T}$, where $a_{1}$ is such that the Bayes risk is $14.44\%$;
\item Case 1b: $p =2$, $\rho = 0.99, \mu_{1} = a_{2}(1,-1)^{T}, \mu_{0} = a_{2}(-1,1)^{T}$, where $a_{2}$ is such that the Bayes risk is $14.44\%$.
\end{itemize} 
In Table~\ref{tab--toy} we present the misclassification errors for LDA applied to the original data and the random projection ensemble classifier with $d = 1$, $B_{1} = 500$, $B_{2} = 50$, $n=200$, and both Gaussian and axis-aligned projections.  We also present the average test error of the LDA classifier applied on the chosen projections.   LDA is tailored to these setups, and indeed it performs very well;  the RP-LDA$_1$ classifier has similar performance in both cases.  The extreme correlation ($\rho  = 0.99$) does not greatly affect the performance of the RP-LDA$_{1}$ (Gaussian) classifier; in particular, while the high correlation does have a small effect on the average error base classifier applied on the chosen projections, this is overcome in the ensemble step.  This illustrates what we believe to be the advantage of aggregation over the choice of a single projection (discussed by Carvalho, Page and Barney). 

We now repeat the experiment with $p = 100$, $d = 5$, and all other parameters kept as before.  The class-conditional covariance matrices have ones on the diagonal and $\rho$ on the off-diagonal.
\begin{itemize} 
\item Case 2a: $p =100$, $\rho = 0, \mu_{1} = a_{3}(1,0, \ldots, 0)^{T}, \mu_{0} = a_{3}(-1,0, \ldots, 0)^{T}$, where $a_{3}$ is such that Bayes risk $= 14.44\%$ 
\item Case 2b: $p =100$, $\rho = 0.99, \mu_{1} = a_{4}(1,-1,0, \ldots,0)^{T}, \mu_{0} = a_{4}(-1,1,0, \ldots, 0)^{T}$, where $a_{4}$ is such that Bayes risk $= 14.44\%$ 
\end{itemize} 

Here, the sample covariance matrix is ill-conditioned, so LDA performs poorly, and the random projection ensemble classifier offers considerable improvement.   Except in Case 1b, Assumption~3 holds with an axis-aligned projection.  The axis-aligned version performs better here since we restrict the set of projections, so we have a greater chance of finding good ones.  However, in Case 1b there is no axis-aligned projection that results in a classifier significantly better than a random guess, and the resulting random projection ensemble classifier is also close to a random guess. 

\medskip

\emph{Methodological variations}

Many discussants suggested alternatives to our basic methodological proposal.  These included the assignment of weights to the selected projections, based on their empirical performance (Chen and Shah; Feng; Zhang; Josh, Fan and James), choosing projections via projection pursuit (Janson), consideration of the underlying algebraic and topological structure (Stehl\'ik and St\v{r}elec), decoupling rotation and dimension reduction (Blaser and Fryzlewicz) or averaging over class probability estimates rather than classifiers (Gneiting and Lerch).  These are attractive and sensible ideas, though similarly to Chen and Shah, we found in our experiments that more sophisticated weighting schemes led to only relatively minor (if any) improvements.  One advantage of our proposal is that it is amenable to theoretical understanding, through the independence of the selected projections, conditional on the training data.  Meanwhile, Tomal, Welch and Zamar highlight their ensemble of phalanxes method, where features are clustered hierarchically into subsets, Casarin, Frattorolo and Rossini and Stander and Dalla Valle suggest copula-based discriminant analysis and Tong discusses neural network approaches, which are also attractive but currently seem less amenable to theoretical understanding.

\begin{table}
\caption{\label{tab--toy}Misclassification rates for the Gaussian toy example.}
\centering
\fbox{%
\begin{tabular}{l | c | c c | c  c }
& &\multicolumn{2}{c|}{Gaussian} &\multicolumn{2}{c}{axis-aligned}\\
& LDA       &RP-LDA$_{d}$ &  $\frac{1}{B_{1}} \sum_{b_{1} =1}^{B_{1}} R(C_{n}^{\mathbf{A}_{b_{1}}})$ & RP-LDA$_{d}$ &  $\frac{1}{B_{1}} \sum_{b_{1} =1}^{B_{1}} R(C_{n}^{\mathbf{A}_{b_{1}}})$ \\
\hline
Case 1a &$14.2_{0.2}$&$15.1_{0.4}$&$15.3_{0.5}$&$14.2_{0.3}$&$14.2_{0.3}$\\
Case 1b &$14.8_{0.3}$&$15.3_{0.3}$&$17.6_{0.3}$&$47.1_{0.4}$&$47.1_{0.4}$\\
\hline
Case 2a & $27.1_{0.8}$&$19.7_{0.6}$&$38.4_{0.3}$&$14.9_{0.6}$&$18.0_{0.4}$\\
Case 2b & $27.7_{0.9}$&$21.6_{0.9}$&$38.8_{0.3}$&$19.4_{0.8}$&$25.1_{0.3}$\\
\end{tabular}}
\end{table}

Some contributors discussed the axis-aligned version of our proposal in more detail (Janson; Ling, Yang and Xue).   Another popular alternative was to generate the projections from different distributions with the aim of finding \emph{good} projections more efficiently (Blaser and Fryzlewicz; Zhang; Derenski, Fan and James).  Other ideas included choosing new projections to be dissimilar to those already chosen; either orthogonal (Feng) or by adding some similarity penalty (Lu and Xue).  We remark that, in our experience and in high dimensions, the selected projections tend to be nearly orthogonal anyway.  Thulin suggests including a random rescaling when generating the projections; on the other hand, both Critchley and Durrant discuss deterministic rescaling or standardising of the variables.  While one could construct examples where such renormalisation would lead to poor performance, these ideas are certainly worth investigating further.

Our paper focuses on zero-one error loss, where the two types of misclassification are assumed equally serious. As pointed out by both Hand and Tong and Li, in practice it is often the case that one type of error is more serious than the other.  Suppose now that for some $m > 0$,
\[
R(C) = \pi_1 \int_{\mathbb{R}^p} \mathds{1}_{\{C(x) = 0\}} \, dP_1(x) + m\pi_0 \int_{\mathbb{R}^p} \mathds{1}_{\{C(x) = 1\}} \, dP_0(x),
\]
so that assigning a class zero observation to class one is $m$ times more serious than the other error.  Three modifications should be made to the methodology.  First, the base classifier should target the misclassification imbalance; for example, for LDA the projected data base classifier would be
\[
C_{n}^{A - \mathrm{LDA}}(x) := \left\{ \begin{array}{l l} 
1 & \text{if} \ \log\bigl(\frac{\hat{\pi}_{1}}{m\hat{\pi}_{0}} \bigr) + \bigl(A x - \frac{\hat{\mu}_{1}^{A} + \hat{\mu}_{0}^{A}}{2}\bigr)^{T} \hat{\Omega}^{A}(\hat{\mu}_{1}^{A} - \hat{\mu}_{0}^{A}) \geq 0, \\
0 & \text{otherwise.} \end{array} \right.
\]
Second, the projections should be selected based on the corresponding weighted estimate (cf. (7) in the main text), 
for example using the training error 
\[
R_n^A := \frac{1}{n_{1} + m n_{0}} \biggl\{ \sum_{\{i: Y_i = 1\}} \mathds{1}_{\{C_n^{A}(X_i) = 0\}} +   m \sum_{\{i: Y_i = 0\} } \mathds{1}_{\{C_n^{A}(X_i) = 1\}} \biggr\}.   
\]
Finally, $\alpha$ should be chosen to mimic the weighted version of equation (5), i.e.
\[
\alpha^* = \argmin_{\alpha' \in [0,1]} \Bigl[  \pi_1 G_{n,1} (\alpha') + m \pi_0 \{1 - G_{n,0}(\alpha') \} \Bigr].
\]

\medskip

\emph{Theoretical extensions}

Several discussants (Critchley; Fan and Zhu; Feng; Kong; Shi, Song and Lu; Tong and Li; Wang and Leng) comment on our theoretical assumptions, and in particular the quantity $\beta$ in our Assumption~2.  Since the training data are considered fixed in the corresponding section of the paper, $\beta$ can depend on the training data (and therefore $n$ and $p$).  In the online supplement, we show that in practice we can typically expect Assumption~2 to hold with $\beta$ not too small.  We see in particular that increasing $p$ does not necessarily lead to $\beta \searrow 0$ (recall that the Johnson--Lindenstrauss Lemma guarantees that, regardless of the magnitude of $p$, we can reduce dimension from $p$ to $O(\log n)$ while nearly preserving pairwise distances).  

Assumption~3 is at the population level.  A natural relaxation is to assume that the oracle projection $A^*$ does not perfectly preserve the class information, but instead allow for a region where the projected classifier disagrees with the Bayes classifier.  This can be formalised through the existence of a projection $A^* \in \mathcal{A}$ and $\tau \geq 0$ such that 
\[
P_X(\{x \in \mathbb{R}^p: \eta(x) \geq 1/2\} \triangle \{x \in \mathbb{R}^p: \eta^{A^*}(A^*x) \geq 1/2\}) = \tau.
\]
Then, by a straightforward extension to Proposition 2, we have that $R(C^{\mathrm{Bayes}}) \leq R(C^{A^* - \mathrm{Bayes}}) \leq R(C^{\mathrm{Bayes}}) + \tau$.

Bing and Wegkamp suggest a possible alternative approach to our theoretical analysis, which involves regarding the RP classifier as a \emph{plug-in} rule with $\nu_{n}(x) + 1/2 - \alpha$ acting as an estimate of $\eta(x)$.  We have found that $\nu_{n}(x)$ is not a good estimate of $\eta(x)$ (even with the suggested bias correction), though it would be interesting to find conditions under which we can hope to estimate $\eta$ using our RP methodology (cf.~Gneiting and Lerch).  

\medskip

\emph{Numerical comparisons}

We welcome the contributions which added to our numerical work, aiding the understanding of the practical properties of the random projection ensemble classifier.  For instance, Gallaugher and McNicholas compare with mixture discriminant analysis, while Stander and Dalla Valle apply the random projection ensemble classifier to a trip advisor dataset.  

Hennig and Viroli found that our proposal performed poorly compared with their quantile-based classifier in two of their setups.  In their Setup 2, Class 1 has $p$ independent, log-normal components, whereas (in the $100 q\%$ signal variables case) class 0 has $p$ independent components, $qp$ log-normal components shifted by 0.2, and $(1-q)p$ log-normal components.  A key characteristic of the data in this setup is that all variables are skewed and positive.  In this example, our Assumption~3 does not hold for $d = 5$, and in fact the best low-dimensional projection has high test error (compared with a Bayes risk of almost zero when $q=1$).  Nevertheless, we can check for skewness and include a marginal logarithmic transformation as a preprocessing step in this instance.  In Table~\ref{tab--Quantile}, we present error rates when data are generated from Hennig and Viroli's Setup~2 with $p =100$, $n =50$, and we take componentwise logarithms of the data before applying the RP methodology.  For reference we also present the performance of the Quantile based methods (QCG, QCS) from Hennig and Viroli's discussion.   Our transformation works very well when $q =1$ (it should be noted that many of the other methods discussed by Hennig and Viroli may also benefit from this preprocessing).  In the case $q=0.1$ and when $n$ is this small, the problem is very challenging and all methods struggle; in particular, we are unable to retain many of the signal projections because our overfitting term $\epsilon_n$ is large.  

Bergsma and Jamil only use $B_{1} = 30$, $B_{2} =  5$ when using the RP methodology in conjunction with Gaussian process regression with fractional Brownian motion for reasons of computational cost.  We have found that larger values of $B_1$ and $B_2$ give considerably better results, but fortunately simple (and quick to compute) base classifiers usually suffice.  Hand suggests a comparison with a weighted $k$-nearest neighbour classifier.  One option is the bagged nearest neighbour classifier, which is essentially a weighted nearest neighbour classifier with geometrically decaying weights \citep{HallSamworth2005,BiauDevroye2010}.  An alternative is to use the optimal weighting scheme, which produces an asymptotic improvement of $5-10\%$ in excess risk over the unweighted $k$-nearest neighbour classifier when $d \leq 15$ \citep{Samworth2012}.  It would be interesting to see if similar improvements are obtained when used in conjunction with the RP methodology.

\begin{table}
\caption{\label{tab--Quantile}Misclassification rates for the random projection ensemble classifier for Setup 2 with log preprocessing ($B_{1} = 500$, $B_{2} = 50$, Gaussian projections).}
\centering
\fbox{%
\begin{tabular}{c | c c c | c c}
& RP-LDA$_{5}$ & RP-QDA$_{5}$ & RP-$k$nn$_{5}$ & QCG & QCS \\
\hline
$q =1$&$18.4_{0.9}$&$12.6_{1.6}$&$16.4_{2.2}$&$25.7$&$21.3$ \\
                                 $q = 0.1$&$46.7_{0.7}$&$46.4_{0.6}$&$46.1_{0.5}$&$44.3$&$41.5$\\
                                 \hline
\end{tabular}}
\end{table}

\medskip

\emph{Other statistical problems}

It was particularly pleasing to see many contributions that discuss using the random projection ensemble framework to tackle other high-dimensional statistical problems.  Several contributors suggested ways in which the information in the chosen projections can be aggregated to provide measures of variable importance (Fortunato, Anderlucci and Montanari; Derenski, Fan and James; Gataric).  Li and Yu, Critchley and Murtagh and Contreas considered clustering (unsupervised learning) problems, where the labels of the training data are unknown.  Here we require both a (sample) measure of the performance of the base method in order to select the projections analogously to~(7) in the main text, and a suitable method for aggregating the chosen projections.  Fan and Zhu discuss the use of random projections for the estimation of the top $k$ left singular space of a data matrix; the result they state together with an appropriate version of Wedin's theorem \citep{Wedin1972,YWS2015,Wang2016} may allow the control of the sine angle distance they seek.  Other interesting new directions discussed include interaction network learning (Demirkaya and Lv), regression (Kong; Shin, Zheng and Wu), feature detection (Mateu) and estimation of central subspaces in the context of sufficient dimension reduction (Sabolov\'a and Marriott).

\medskip

\emph{Which random projection ensemble classifier?} 

We are grateful to Switzer for pointing out two early references to the use of random projections for classification.  As noted by some discussants (Hand; Hennig and Viroli; Critchley), the flexibility offered by our random projection ensemble classification framework naturally poses the question of when a particular base classifier should be used (of course, analogous questions arise regardless of whether methods are used in conjunction with random projections).  If no natural choice is suggested from understanding of the data generating mechanism, one possible approach is to randomise the choice of base classifier for each projection, say choosing between LDA, QDA and $k$nn, each with probability $1/3$.  Alternatively, we can try all three base methods on each projection and retain the projection, base method pair that minimises the leave-one-out error estimate.  If one of these three original classifiers is clearly best, then it should emerge as the `winner' within most blocks of $B_2$ projections.  This strategy therefore provides additional robustness, and the theory goes through unchanged for these versions of the random projection ensemble classifiers.  Post-pruning, as suggested by Fortunato, is another option, but we do not pursue that here.  We implement both methods proposed above (denoted RP-Random$_{d}$ and RP-Min$_{d}$, respectively) in a small simulation study, summarised in Table~3, where the model numbers refer to the settings described in Section 6.1.  For Models 2, 3 and 4, the risks of both variants of the classifier are comparable to (or better than) that of the best performing choice of base method.  For Model 1 there is only a slight deterioration in performance.  Taking these ideas further, and addressing comments from Bing and Wegkamp, Critchley and Liu and Cheng, one could even add randomisation over $d$ and/or Gaussian/axis-aligned projections.

\begin{table}
\caption{\label{tab--base}Misclassification rates for the randomised and selected base classifier variants, with $p = 100$, $n = 200$, $B_{1} = 500$, $B_{2} = 50$, $d = 5$ and Gaussian projections.  For comparison, in the bottom row we present the risk of the best performing version of the random projection ensemble classifier as seen in Section 6.1 of the main text.}
\centering
\fbox{%
\begin{tabular}{l | l l l l }
& Model 1 & Model 2 & Model 3 & Model 4 \\
\hline
RP-Random$_{5}$&$26.2_{0.7}$&$6.0_{0.2}$&$3.6_{0.1}$&$23.6_{0.5}$\\
RP-Min$_{5}$        &$23.6_{0.7}$&$6.1_{0.3}$&$3.7_{0.2}$&$23.9_{0.6}$\\
\hline
Best RP               &$22.32_{0.32}$&$5.58_{0.12}$&$4.23_{0.14}$&$24.02_{0.30}$\\
\end{tabular}}
\end{table}

\medskip

\emph{Ultrahigh dimensional problems}

Tong discusses the applicability of our random projection methodology in contemporary machine learning problems.  He correctly points out that some modern datasets have potentially millions of features and observations, far larger than the problem sizes we investigate in our numerical studies in Section~6.   Of course, the fact that such large datasets exist does not mean that we should neglect the (still relevant) smaller problems.  Moreover, in ultrahigh-dimensional problems it is often reasonable to assume that only a subset of the features are relevant. Indeed, many studies of such problems focus on reducing the data dimension by attempting to screen out the noise variables \citep[e.g.][]{FanLv2008,FSW2009,MeinshausenBuhlmann2010,ShahSamworth2013}.  If high dimension is still a problem, another common technique is to use a single random projection \citep[e.g.][]{Achlioptas2003} into a lower dimensional space.  Either or both these techniques can be used as a preprocessing step to give thousands, say, rather than millions of features, and then the RP methodology can be applied.  In fact, in the \citet{Dahl2013} paper cited by Tong, in order to make the problem more manageable, the authors apply feature screening and a sparse random projection to reduce dimension to 4000, before applying a neural net classifier.  

\medskip

\emph{Responses to direct questions}

Gallaugher and McNicholas seek clarification about our real data settings -- we used the Hill-Valley data set without noise, pooled the training and test sets, then subsampled at random our own training and test sets as described in Section~6.2.  The missing values in the Mice dataset were imputed as the sample average value for that feature for the non-missing entries.  Kong asks why the performance improves as $p$ increases for Model 1.  One reason  is that, while the signal is the same (the Bayes risk is 4.45\% in both cases), the variance of the noise components is reduced in the higher-dimensional setting; see also the explanation of Yatracos.  In answer to Zhang, penalised logistic regression does not perform well in Setting 1 because, despite the fact the model is highly sparse (only two features are relevant for classification), the class boundaries are non-linear.  Stander and Dalla Valle ask whether it is possible to quantify classification uncertainty using $C_n^{\mathbf{A}_1},\ldots,C_n^{\mathbf{A}_{B_1}}$.  Regarding the training data as fixed and having observed $\nu_n(x)=t < \alpha$, say, one can indeed obtain a simple bound on the probability of observing $\nu_n(x)$ at least as small as $t$ when $C_n^{\mathrm{RP}^*}(x) = 1$ (a kind of `$p$-value'), via the fact that $\nu_n(x) \sim B_1^{-1}\mathrm{Bin}\bigl(B_1,\mu_n(x)\bigr)$.


\begin{thebibliography}{}
\bibitem[{Ailon and Chazelle(2006)}]{Ailon:06} Ailon, N. and Chazelle, B. (2006). Approximate nearest neighbours and the fast Johnson--Lindenstrauss transform. \newblock \emph{Proceedings of the Symposium on Theory of Computing}, pp.~557--563.

\bibitem[{Bickel and Levina(2004)}]{Bickel:04} Bickel, P.~J. and Levina, E. (2004). Some theory for {F}isher's linear discriminant function, `naive {B}ayes', and some alternatives when there are more variables than observations.
\newblock \emph{Bernoulli}, \textbf{10}, 989--1010.

\bibitem[{Blaser and Fryzlewicz(2015)}]{BlaserFryzlewicz2015}Blaser, R. and Fryzlewicz, P. (2015). Random rotation ensembles.
\newblock \emph{J. Mach. Learn. Res.}, \textbf{17}, 1--26.

\bibitem[{Breiman(1996)}]{Breiman:96} Breiman, L. (1996). Bagging Predictors.
\newblock \emph{Machine Learning}, \textbf{24}, 123--140.

\bibitem[{Breiman(2001)}]{Breiman2001} Breiman, L. (2001).  Random Forests.
\newblock \emph{Machine Learning}, \textbf{45}, 5--32.

\bibitem[{Breiman et~al.(1984)Breiman, Friedman, Stone and Olshen}]{Breiman:84} Breiman, L., Friedman, J., Stone, C.~J. and Olshen, R.~A. (1984). \emph{Classification and Regression Trees}.
\newblock {Chapman and Hall}, New York.

\bibitem[{Cannings and Samworth(2016)}]{CanningsSamworth2016b} Cannings, T. I. and Samworth, R. J. (2016). \texttt{RPEnsemble}: Random projection ensemble classification. \newblock \texttt{R} package, v. 0.3. \url{https://cran.r-project.org/web/packages/RPEnsemble/index.html}

\bibitem[{Cannings and Samworth(2017)}]{CanningsSamworth:16} Cannings, T. I. and Samworth, R. J. (2017). Random-projection ensemble classification: supplementary material.  \newblock \emph{Available online}.

\bibitem[{Chikuse(2003)}]{Chikuse:03} Chikuse, Y. (2003). \emph{Statistics on Special Manifolds.} Lecture notes in statistics, volume 174.
\newblock {Springer--Verlag}, New York. 

\bibitem[{Cook(1998)}]{Cook:98} Cook, R.~D. (1998). \emph{Regression Graphics: Ideas for
Studying Regressions through Graphics}.
\newblock {Wiley}, New York.

\bibitem[{Cortes and Vapnik(1995)}]{Cortes:95} Cortes, C. and Vapnik, V. (1995). Support-vector networks.
\newblock \emph{Machine Learning}, \textbf{20}, 273--297.

\bibitem[{Dasgupta(1999)}]{Dasgupta:99} Dasgupta, S. (1999). Learning mixtures of Gaussians.
\newblock \emph{Proc. 40th Annual Symposium on Foundations of Computer Science}, 634--644.

\bibitem[{Dasgupta and Gupta(2002)}]{Dasgupta:02} Dasgupta, S. and Gupta, A. (2002). An elementary proof of the Johnson--Lindenstrauss Lemma.
\newblock \emph{Random Struct. Alg.},  \textbf{22}, 60--65.

\bibitem[{Devroye and Wagner(1976)}]{DevroyeWagner1976}Devroye, L. P. and Wagner, T. J. (1976). A distribution-free performance bound in error estimation.
\newblock \emph{IEEE Trans. Info. Th.}, \textbf{22}, 586--587.

\bibitem[{Devroye and Wagner(1979)}]{DevroyeWagner:1979}Devroye, L. P. and Wagner, T. J. (1979). Distribution-free inequalities for the deleted and hold-out error estimates.
\newblock \emph{IEEE Trans. Info. Th.}, \textbf{25}, 202--207.

\bibitem[{Devroye et~al.(1996)Devroye, Gy\"orfi and Lugosi}]{PTPR:96} Devroye, L., Gy\"orfi, L. and Lugosi, G. (1996). \emph{A Probabilistic Theory of Pattern Recognition}. \newblock {Springer}, New York.

\bibitem[{Durrant and Kab\'{a}n(2013)}]{Durrant:13} Durrant, R.~J. and Kab\'{a}n, A. (2013). Sharp generalization error bounds for randomly-projected classifiers.
\newblock \emph{J. Mach. Learn. Res. W \& CP}, \textbf{28}, 693--701.

\bibitem[{Durrant and Kab\'{a}n(2015)}]{Durrant:15} Durrant, R.~J. and Kab\'{a}n, A. (2015). Random projections as regularizers: learning a linear discriminant from fewer observations than dimensions. 
\newblock \emph{Machine Learning}, \textbf{99}, 257--286.

\bibitem[{Efron(1975)}]{Efron:75} Efron, B. (1975). The efficiency of logistic regression compared to normal discriminant analysis. 
\newblock \emph{J. Amer. Statist. Assoc.}, \textbf{70}, 892--898.

\bibitem[{Esseen(1945)}]{Esseen:45} Esseen, C.-G. (1945). Fourier analysis of distribution functions. A mathematical study of the Laplace--Gaussian law.
\newblock \emph{Acta Mathematica}, \textbf{77}, 1--125.

\bibitem[{Fan and Fan(2008)}]{Fan:08} Fan, J. and Fan, Y. (2008). High-dimensional classification using features annealed independence rules.
\newblock \emph{Ann. Statist.}, \textbf{36}, 2605--2637.

\bibitem[{Fan et~al.(2012)Fan, Feng and Tong}]{Fan:12}Fan, J., Feng, Y. and Tong, X. (2012). A road to classification in high dimensional space: the regularized optimal affine discriminant.
\newblock \emph{J. Roy. Statist. Soc., Ser. B.}, \textbf{72}, 745--771.

\bibitem[{Fisher(1936)}]{Fisher:36} Fisher, R.~A. (1936). The use of multiple measurements in taxonomic problems.
\newblock \emph{Annals of Eugenics}, \textbf{7}, 179--188.

\bibitem[{Fix and Hodges(1951)}]{Fix:51} Fix, E. and Hodges, J.~L. (1951). Discriminatory analysis -- nonparametric discrimination: Consistency properties.
\newblock Technical Report 4, USAF School of Aviation Medicine, Randolph Field, Texas.

\bibitem[{Friedman(1989)}]{Friedman:89}Friedman, J. (1989). Regularized discriminant analysis.
\newblock \emph{J. Amer. Statist. Assoc.}, \textbf{84}, 165--175.

\bibitem[{Goeman et~al.(2015)Goeman, Meijer and Chaturvedi}]{Goeman:15} Goeman, J., Meijer, R. and Chaturvedi, N. (2015). \texttt{penalized}: {L1} (lasso and fused lasso) and {L2} (ridge) penalized estimation in {GLM}s and in the {Cox} model. \newblock \texttt{R} package version 0.9-45, \url{http://cran.r-project.org/web/packages/penalized/}

\bibitem[{Gnedenko and Kolmogorov(1954)}]{Gnedenko:54} Gnedenko, B.~V. and Kolmogorov, A.~N. (1954). \emph{Limit Distributions for Sums of Independent Random Variables}, \newblock Addison-Wesley, Cambridge MA.

\bibitem[{Gul et~al.(2015)Gul, Perperoglou, Khan, Mahmoud, Miftahuddin, Adler and Lausen}]{Gul:15} Gul, A., Perperoglou, A., Khan, Z., Mahmoud, O., Miftahuddin, M., Adler, W. and Lausen, B. (2015). \texttt{ESKNN}: Ensemble of subset of K-nearest neighbours classifiers for classification and class membership probability estimation.  \newblock \texttt{R} package version 1.0, \url{http://cran.r-project.org/web/packages/ESKNN/} 

\bibitem[{Gul et~al.(2016)Gul, Perperoglou, Khan, Mahmoud, Miftahuddin, Adler and Lausen}]{Gul:16} Gul, A., Perperoglou, A., Khan, Z., Mahmoud, O., Miftahuddin, M., Adler, W. and Lausen, B. (2016). Ensemble of a subset of $k$NN classifiers. \newblock \emph{Adv. Data Anal. Classif.}, 1-14.

\bibitem[{Hall and Kang(2005)}]{Hall:05a} Hall, P. and Kang, K.-H. (2005). Bandwidth choice for nonparametric classification.
\newblock \emph{Ann. Statist.},  \textbf{33}, 284--306.

\bibitem[{Hall et~al.(2008)Hall, Park and Samworth}]{Hall:08}Hall, P., Park, B.~U. and Samworth, R.~J. (2008). Choice of neighbour order in nearest-neighbour classification.
\newblock \emph{Ann. Statist.},  \textbf{36}, 2135--2152.

\bibitem[{Hall and Samworth(2005)}]{Hall:05b}Hall, P. and Samworth, R.~J. (2005). Properties of bagged nearest neighbour classifiers.
\newblock \emph{J. Roy. Statist. Soc., Ser. B.}, \textbf{67}, 363--379.

\bibitem[{Hastie et~al.(1995)Hastie, Buja and Tibshirani}]{Hastie:95}Hastie, T., Buja, A. and Tibshirani, R. (1995). Penalized discriminant analysis.
\newblock \emph{Ann. Statist.}, \textbf{23}, 73--102.

\bibitem[{Hastie et~al.(2009)Hastie, Tibshirani and Friedman}]{ESL:09}Hastie, T., Tibshirani, R., and Friedman, J. (2009).
\newblock \emph{The Elements of Statistical Learning: Data Mining, Inference, and Prediction.} \newblock Springer Series in Statistics (2nd ed.). \newblock Springer, New York.

\bibitem[Hastie et~al.(2015) Hastie, Tibshirani, Narasimhan and Chu]{Hastie:2015} Hastie, T., Tibshirani, R., Narisimhan, B. and Chu, G. (2015). \texttt{pamr}: Pam: prediction analysis for microarrays.
\newblock \texttt{R} package version 1.55, \url{http://cran.r-project.org/web/packages/pamr/}

\bibitem[Karatzoglou, Smola and Hornik(2015)]{Karatzoglou:2015} Karatzoglou, A., Smola A. and Hornik, K. (2015). \texttt{kernlab}: Kernel-based Machine Learning Lab. \newblock \texttt{R} package version 0.9-20, \url{http://cran.r-project.org/web/packages/kernlab/}

\bibitem[{Khan et~al.(2015{\natexlab{a}})Khan, Gul, Mahmoud, Miftahuddin, Perperoglou, Adler and Lausen}]{Khan:15} Khan, Z., Gul, A., Mahmoud, O., Miftahuddin, M., Perperoglou, A., Adler, W. and Lausen, B. (2015{\natexlab{a}}). An ensemble of optimal trees for class membership probability estimation. \newblock In: Wilhelm, A., Kestler, H. A. (eds.), \emph{Analysis of Large and Complex Data, European Conference on Data Analysis, Bremen, July, 2014. Series: Studies in Classification, Data Analysis, and Knowledge Organization}. Springer-Verlag, Berlin.

\bibitem[{Khan et~al.(2015{\natexlab{b}})Khan, Gul, Mahmoud, Miftahuddin, Perperoglou, Adler and Lausen}]{Khan:15a} Khan, Z., Gul, A., Mahmoud, O., Miftahuddin, M., Perperoglou, A., Adler, W. and Lausen, B. (2015{\natexlab{b}}). \texttt{OTE}: Optimal trees ensembles for regression, classification and class membership probability estimation. \newblock \texttt{R} package version 1.0, \url{http://cran.r-project.org/web/packages/OTE/} 

\bibitem[{Larsen and Nelson(2016)}]{Larsen:16}Larsen, K. G. and Nelson, J. (2016). The Johnson--Lindenstrauss lemma is optimal for linear dimensionality reduction.
\newblock \emph{43rd International Colloquium on Automata, Languages and Programming}, \textbf{82}, 1-11.

\bibitem[{Le, Sarlos and Smola(2013)}]{LSS2013}Le, Q., Sarlos, T. and Smola, A. (2013). Fastfood --- approximating kernel expansions in loglinear time.
\newblock \emph{J. Mach. Learn. Res.}, W \& CP \textbf{28}, 244--252. 

\bibitem[{Lee et~al.(2013)Lee, Li and Chiaromonte}]{Lee:13}Lee, K.-Y., Li, B. and Chiaromonte, F. (2013). A general theory for nonlinear sufficient dimension reduction: formulation and estimation.
\newblock \emph{Ann. Statist.},  \textbf{41}, 221--249.

\bibitem[{Li(1991)}]{Li:91}Li, K.-C. (1991). Sliced inverse regression for dimension reduction. 
\newblock \emph{J. Amer. Statist. Assoc.},  \textbf{86}, 316--342.

\bibitem[Liaw and Wiener(2014)]{Liaw:2014} Liaw, A. and Wiener, M. (2014). \texttt{randomForest}: Breiman and Cutler's random forests for classification and regression. \newblock \texttt{R} package version 4.6-10, \url{http://cran.r-project.org/web/packages/randomForest/}

\bibitem[{Lopes(2016)}]{Lopes:16}Lopes, M. (2016). A sharp bound on the computation-accuracy tradeoff for majority voting ensembles. 
\newblock \emph{arXiv e-prints}, \texttt{1303.0727v2}.

\bibitem[{Lopes et~al.(2011)Lopes, Jacob and Wainwright}]{Lopes:11}Lopes, M., Jacob, L. and Wainwright, M.~J. (2011). A more powerful two-sample test in high dimensions using random projection.
\newblock \emph{Advances in Neural Information Processing Systems (NIPS)}.

\bibitem[{Marzetta et~al.(2011)Marzetta, Tucci and Simon}]{Marzetta:11}Marzetta, T., Tucci, G. and Simon, S. (2011). A random matrix-theoretic approach to handling singular covariance estimates.
\newblock \emph{IEEE Trans. Inf. Th.}, \textbf{57}, 6256--6271.

\bibitem[{McWilliams et~al.(2014)}]{McWilliams:14} McWilliams, B., Heinze, C., Meinshausen, N., Krummenacher, G. and Vanchinathan, H.~P. (2014). {LOCO}: distributing ridge regression with random projections.
\newblock \emph{arXiv e-prints}, \texttt{1406.3469v2}.

\bibitem[{Meinshausen and B\"uhlmann(2010)}]{MeinshausenBuhlmann2010}Meinshausen, N. and B\"uhlmann, P. (2010). Stability selection.
\newblock \emph{J. Roy. Statist. Soc., Ser. B (with discussion)}, \textbf{72}, 417--473.

\bibitem[Meyer et~al.(2015)]{Meyer:2015} Meyer, D., Dimitriadou, E., Hornik, K., Weingessel A., Leisch, F., Chang, C.-C. and  Lin, C.-C. (2015). \texttt{e1071}: Misc Functions of the Department of Statistics (e1071), TU Wien. \newblock \texttt{R} package version 1.6-4, \url{http://cran.r-project.org/web/packages/e1071/}

\bibitem[{Samworth(2012)}]{Samworth:12} Samworth, R.~J. (2012). Optimal weighted nearest neighbour classifiers.
\newblock \emph{Ann. Statist.}, \textbf{40}, 2733--2763.

\bibitem[{Shah and Samworth(2013)}]{ShahSamworth2013}Shah, R. D. and Samworth, R. J. (2013). Variable selection with error control: another look at stability selection.
\newblock \emph{J. Roy. Statist. Soc., Ser. B}, \textbf{75}, 55--80. 

\bibitem[Shin et~al.(2014) Shin, Wu, Zhang and Liu]{Shin:2014} Shin, S.~J., Wu, Y., Zhang, H.~H. and Liu, Y. (2014). Probability-enhanced sufficient dimension reduction for binary classification. \newblock \emph{Biometrics}, \textbf{70}, 546--555.

\bibitem[{Trefethen and Bau(1997)}]{Trefethen:97} Trefethen, L., N. and Bau, D., III (1997). \textit{Numerical Linear Algebra}. \newblock Society for Industrial and Applied Mathematics, Philadelphia.

\bibitem[{Tibshirani et~al.(2002)Tibshirani, Hastie, Narisimhan and Chu}]{Tibshirani:02}Tibshirani, R., Hastie, T., Narisimhan, B. and Chu, G. (2002). Diagnosis of multiple cancer types by shrunken centroids of gene expression.
\newblock \emph{Proceedings of the Natural Academy of Science, USA},  \textbf{99}, 6567--6572.

\bibitem[{Tibshirani et~al.(2003)Tibshirani, Hastie, Narisimhan and Chu}]{Tibshirani:03} Tibshirani, R., Hastie, T., Narisimhan, B. and Chu, G. (2003). Class prediction by nearest shrunken centroids, with applications to {DNA} microarrays.
\newblock \emph{Statist. Science}, \textbf{18}, 104--117.

\bibitem[{Vershynin(2012)}]{Vershynin2012}Vershynin, R. (2012). Introduction to the non-asymptotic analysis of random matrices.
\newblock In \emph{Compressed Sensing} (Eds. Y. C. Eldar and G. Kutyniok), pp. 210--268.  Cambridge University Press, Cambridge.


\bibitem[{Williams and Barber(1998)}]{Williams:98} Williams,  C. K. I.  and Barber, D. (1998). Bayesian classification with Gaussian processes. \newblock \emph{IEEE Transactions on Pattern Analysis and Machine Intelligence}, \textbf{20}, 1342--1351.

\bibitem[Witten(2011)]{Witten:2011} Witten, D. (2011). \textit{penalizedLDA}: Penalized classification using Fisher's linear discriminant.
\newblock \texttt{R} package version 1.0, \url{http://cran.r-project.org/web/packages/penalizedLDA/}

\bibitem[{Witten and Tibshirani(2011)}]{Witten:11}Witten, D.~M. and Tibshirani, R. (2011). Penalized classification using Fisher's linear discriminant.
\newblock \emph{J. Roy. Statist. Soc., Ser. B.}, \textbf{73}, 753--772.

\end{thebibliography}

\begin{thebibliography}{}
\bibitem[{Cannings and Samworth(2016)}]{CanningsSamworth:16b} Cannings, T. I. and Samworth, R. J. (2016). \texttt{RPEnsemble}: Random projection ensemble classification. \newblock \texttt{R} package, v. 0.3. \url{https://cran.r-project.org/web/packages/RPEnsemble/index.html}

\bibitem[{Cannings and Samworth(2017)}]{CanningsSamworth:16Main} Cannings, T. I. and Samworth, R. J. (2017). Random-projection ensemble classification.  \emph{J. Roy. Statist. Soc., Ser. B (with discussion)}, to appear.

\bibitem[{Gnedenko and Kolmogorov(1954)}]{Gnedenko:54} Gnedenko, B.~V. and Kolmogorov, A.~N. (1954). \emph{Limit Distributions for Sums of Independent Random Variables}, \newblock Addison-Wesley, Cambridge MA.

\end{thebibliography}

\begin{thebibliography}{}

\bibitem[{Achlioptas(2003)}]{Achlioptas2003} Achlioptas, D. (2003). Database-friendly random projections: Johnson--Lindenstrauss with binary coins.
\newblock \emph{J. Comp. Sys. Sci.}, \textbf{66}, 671--687.

\bibitem[{Biau and Devroye(2010)}]{BiauDevroye2010}Biau, G. and Devroye, L. (2010). On the layered nearest neighbour estimate, the bagged nearest neighbour estimate and the random forest method in regression and classification.
\newblock \emph{J. Mult. Anal.}, \textbf{101}, 2499--2518.

\bibitem[{Dahl et al.(2013)}]{Dahl2013}Dahl, G. E., Stokes, J. W., Deng, L. and Yu, D. (2013). 
Large-scale malware classification using random projections and neural networks. 
\newblock \emph{IEEE International Conference on Acoustics, Speech and Signal Processing, Vancouver}, 3422-3426.

\bibitem[{Fan and Lv(2008)}]{FanLv2008}Fan, J. and Lv, J. (2008). Sure independence screening for ultrahigh dimensional feature space. 
\newblock \emph{J. Roy. Statist. Soc. Ser. B. (with discussion)}, \textbf{70}, 849--911.

\bibitem[{Fan, Samworth and Wu(2009)}]{FSW2009}Fan, J., Samworth, R. and Wu, Y. (2009). Ultrahigh dimensional feature selection: beyond the linear model. 
\newblock \emph{J. Machine Learning Research}, \textbf{10}, 2013--2038.

\bibitem[{Hall and Samworth(2005)}]{HallSamworth2005}Hall, P. and Samworth, R. J. (2005). Properties of bagged nearest neighbour classifiers.
\newblock \emph{J. Roy. Statist. Soc. Ser. B.}, \textbf{67}, 363--379.

\bibitem[{Meinshuasen and B\"{u}hlmann(2010)}]{MeinshausenBuhlmann2010} Meinshausen, N. and B\"{u}hlmann, P. (2010). Stability Selection.
\newblock \emph{J. Roy. Statist. Soc. Ser. B. (with discussion)}, \textbf{72}, 417--473.

\bibitem[{Samworth(2012)}]{Samworth2012}Samworth, R.~J. (2012). Optimal weighted nearest neighbour classifiers.
\newblock \emph{Ann. Statist.}, \textbf{40}, 2733--2763.

\bibitem[{Shah and Samworth(2013)}]{ShahSamworth2013}Shah, R. D. and Samworth, R. J. (2013) Variable selection with error control: Another look at Stability Selection.
\newblock \emph{J. Roy. Statist. Soc., Ser. B}, \textbf{75}, 55--80.

\bibitem[{Wang(2016)}]{Wang2016}Wang, T. (2016). Spectral methods and computational trade-offs in high-dimensional statistical inference.  
\newblock Ph.D. thesis, University of Cambridge.

\bibitem[{Wedin(1972)}]{Wedin1972}Wedin, P.-\AA. (1972). Perturbation bounds in connection with singular value decomposition.
\newblock \emph{BIT Numerical Mathematics} \textbf{12}, 99--111.

\bibitem[{Yu, Wang and Samworth(2015)}]{YWS2015}Yu, Y., Wang, T. and Samworth, R. J. (2015). A useful variant of the Davis--Kahan theorem for statisticians. 
\newblock \emph{Biometrika}, \textbf{102}, 315--323.

\end{thebibliography}
\end{document}